\documentclass[a4paper,10pt]{article}

\usepackage{amsthm,amsmath,amssymb,amsfonts,graphics,graphicx}
\usepackage{epstopdf}

\def \max {\, \mbox{max}}
\def \diag {\, \mbox{diag}\, }

\def\bbbn{{\mathbb N}}
\def\bbbc{{\mathbb C}}
\def\bbbd{{\mathbb D}}
\def\bbbz{{\mathbb Z}}
\def\bbbr{{\mathbb R}}

\def\bv{{\mathbf e}}
\def\bn{{\mathbf n}}
\def\bm{{\mathbf m}}
\def\be{{\mathbf e}}
\def\i{{\rm i}}
\def\ii{{\iota}}

\def\cL{{\cal L}}
\def\cM{{\cal M}}

\def\cS{{\cal S}}

\def\hu{{\hat u}}
\def\hv{{\hat v}}

\def\tit{{\tilde \tau}}

\def\tr{\mbox{tr}}

\def\tr{\mbox{tr}}

\newtheorem{Def}{Definition}
\newtheorem{The}{Theorem}
\newtheorem{Pro}{Proposition}
\newtheorem{Lem}{Lemma}
\newtheorem{Rem}{Remark}

\begin{document}
\title{Wave fronts and cascades of soliton interactions in the periodic two 
dimensional 
Volterra system}
\author{Rhys Bury, Alexander V. Mikhailov$^{\star}$ and Jing Ping 
Wang$ ^\dagger $
\\
$\dagger$ School of Mathematics, Statistics \& Actuarial Science, University of 
Kent, UK \\
$\star$ Applied Mathematics Department, University of Leeds, UK
}
\date{}

\maketitle

\begin{abstract}
In the paper we develop the dressing method for the solution of the 
two-dimensional periodic Volterra system with a period $N$. 
We derive  soliton solutions of arbitrary rank $k$ and give a 
full classification of rank $1$  solutions. We have found a new class of exact solutions corresponding to 
wave fronts which represent  smooth interfaces 
between two nonlinear periodic waves or a periodic wave and a trivial (zero) solution. The 
wave fronts are non-stationary and they propagate with a constant average 
velocity. The system also has soliton solutions similar to breathers, which resembles soliton webs in the KP theory. 
We associate the classification of soliton solutions with the Schubert 
decomposition of the Grassmanians ${\rm Gr}_\bbbr(k,N)$ and ${\rm Gr}_\bbbc(k,N)$.  
\end{abstract}

\section{Introduction}

Construction of explicit exact solutions for integrable systems is an important and well developed 
area of research. There is a variety of methods designed to tackle this problem 
and vast literature concerning soliton solutions of rank one.
In this paper we develop the dressing method in application to a periodic two-dimensional Volterra 
system and derive explicitly soliton solutions of arbitrary rank. We found  
solutions resembling breathers (in the theory of 
the sine-Gordon equation), nonlinear periodic waves and a new type of exact 
solution for integrable systems, which are  smooth interfaces 
between two nonlinear periodic waves or a periodic wave and a trivial (zero) 
solution.

The dressing method for Lax 
integrable systems was originally formulated and developed in 
\cite{zash, mr80c:81115}. Its predecessor 
was proposed by Bargmann (1949) \cite{bargmann}, where the author performed the dressings
of the Schr\"odinger operator and discovered potentials, which now we would 
associate with the profiles
of one and two soliton solutions for the Korteweg de--Vries (KdV) equation. The 
connection of the potentials
of the Schr\"odinger operators with solutions of the KdV equation was 
established much  later by Gardner, Greene, Kruskal and Miura,  who discovered the 
inverse spectral transform \cite{ggkm}. A year later an elegant interpretation 
of their results 
was given by Lax in \cite{lax68}, where the concept of Lax pair has first appeared.

In this paper, we develop the dressing method and study exact solutions for the  $2$--dimensional generalisation
of the periodic Volterra lattice \cite{mik79,mik81}
\begin{eqnarray}\label{2+1}
\left\{ \begin{array}{ll}\phi^{(i)}_{t}= \theta^{(i)}_{x} + \theta^{(i)} 
\phi^{(i)}_{x}
-e^{2 \phi^{(i-1)}}+ e^{2 \phi^{(i+1)}},\qquad & \phi^{(i+N)}=\phi^{(i)}, \quad 
\theta^{(i+N)}=\theta^{(i)}, \\ & \\
\theta^{(i+1)}-\theta^{(i)}+ \phi^{(i+1)}_{x}+\phi^{(i)}_{x}=0,& \sum_{i=1}^N 
\phi^{(i)}=\sum_{i=1}^N\theta^{(i)}=0,
\end{array}\right.  .
\end{eqnarray}
System (\ref{2+1}) can be regarded as an 
integrable discretisation of the Kadomtsev-Petviashvili (KP) equation (see Section \ref{sec33}).
The KP equation, originally 
derived 
for ion-acoustic waves of small amplitude in plasma \cite{KP70}, is
 a $2+1$--dimensional integrable generalisation of the KdV 
equation. Its mathematical theory made a deep impact to the theory of 
integrable equations and give rise to useful notions such as $\tau$ 
function and Sato Grassmanian \cite{satoKP}. The KP equation possesses a rich 
set of exact solutions, whose classification require  advanced techniques from 
cluster algebra, tropical geometry and combinatorics developed in 
\cite{kodama2011kp, Kodama2013, kodama2014kp}.

Equation (\ref{2+1}) was first derived in 1979 motivated by the reduction 
group theory for Lax representation \cite{mik79}. 
For a fixed period $N$, the variables $\theta ^{(i)}$ can be eliminated and 
thus 
(\ref{2+1}) can be rewritten as a system of $(N-1)$-component second order 
evolutionary equations. In the simplest nontrivial case $N=3$, the system 
(\ref{2+1}) becomes
\begin{eqnarray}\label{eqn3}
\left\{\begin{array}{l}3 \phi^{(1)}_{t}=\phi^{(1)}_{xx}+2\phi^{(2)}_{xx} + 2 
\phi^{(1)}_{x}
\phi^{(2)}_{x}
+{\phi^{(1)}_{x}}^2 +3 e^{2 \phi^{(2)}} -3 e^{-2 \phi^{(1)}-2\phi^{(2)}}\\
3 \phi^{(2)}_{t}=-2\phi^{(1)}_{xx}-\phi^{(2)}_{xx} - 2 \phi^{(1)}_{x} 
\phi^{(2)}_{x} -{\phi^{(2)}_{x}}^2 -3 e^{2
\phi^{(1)}} +3 e^{-2 \phi^{(1)}-2\phi^{(2)}}
\end{array}\right. 
\end{eqnarray}
and after a point transformation it takes the form of a  nonlinear 
Schr\"odinger 
type equation (system ``u4'' in \cite{mr89g:58092})
\[  iu_T=u_{XX}+(u_X^\star)^2+e^{-2u-2u^\star}+\omega^\star e^{-2\omega 
u-2\omega^\star u^\star}+\omega e^{-2\omega^\star u-2\omega u^\star},\quad 
\omega=e^{\frac{2\pi 
i}{3}}\, ,\]
where $\star$ denotes complex conjugation. In this case, the system is  
bi-Hamiltonian. A recursion operator and bi-Hamiltonian structure for 
system (\ref{eqn3}) are 
explicitly constructed from its Lax representation in \cite{wang09}.
A certain class of Darboux transformations for arbitrary fixed period $N$
has recently been constructed in \cite{MPW}.

For infinite $N$, equation (\ref{2+1}) is an integrable differential-difference
equation in $2 + 1$ dimensions. It appeared in \cite{FNR13} where the authors
classified a family of equations with
the non-locality of intermediate long wave type. Its infinitely many 
symmetries and conserved densities are constructed using its master symmetry 
in \cite{wang15}.

Bargmann's potentials correspond 
to a rational (in the wave number) factor to the Jost function \cite{bargmann}. In the dressing method we also
start with a rational in the spectral parameter $\lambda$ matrix factor $\Phi(\lambda)$, which modifies the fundamental solution of the
``undressed'' Lax pair. In the case of system (\ref{2+1}) the Lax operators 
contain $N\times N$ matrices and are invariant with respect to a reduction 
group isomorphic to $\bbbz_2\times\bbbd_N$. We construct the reduction group 
invariant dressing factors $\Phi(\lambda)$ which have $N$ or $2N$ simple poles 
belonging to the orbits generated by transformations $\lambda\mapsto \omega 
\lambda,\ \lambda\mapsto \lambda^*$, where $\omega=\exp(\frac{2\pi i}{N})$. The 
case of $N$ simple poles leads to a new class of solutions, which we call  
kink solutions, while solutions corresponding to the orbits with $2N$ poles we 
call breathers. This terminology is borrowed from the  sine-Gordon theory where 
a kink solution  corresponds to a dressing factor with one pole and  two poles 
factor leads to a breather solution  \cite{fad87, MPW15}. We could also 
construct $(n,m)$ multisoliton solutions with $n$ kinks and $m$ breathers, but 
this generalisation is rather straightforward and therefore in this paper we 
focus on solutions corresponding to a single orbit (i.e. one kink and one 
breather solutions).

A kink solution can 
be parametrised by a real number $\nu\notin \{\pm 1,0\}$ and a point 
on a real Grassmannian ${\rm Gr}_\bbbr (k,N)$, while a breather solution can be 
uniquely parametrised by a complex number 
$\mu\in\bbbc$ such that $ |\mu|\notin \{1,0\},\ \mbox{Im}\,\mu^N\ne 0$ and a 
point on a complex Grassmannian ${\rm Gr}_\bbbc (k,N)$. The number $k$ in ${\rm Gr}(k,N)$ is the rank of the soliton solution. 
There is a difference between the cases of even and odd $N$. When $N$ is even, 
there are two 
different orbits with $N$ points, namely $\{\nu\omega^k\}_{k=1}^N$ and 
$\{\nu\omega^{k+\frac{1}{2}}\}_{k=1}^N$. They results in two different kink 
solutions. A fine classification of wave 
interfaces (in the kink case) and soliton 
interactions (in the breather case) can be naturally given in the terms of the 
invariant
Schubert cell decomposition of the Grassmannian. In particular, elementary  
line breathers and periodic kink solutions correspond to 
one-dimensional invariant Schubert cells of the Grassmannians.

Kink solutions represent regions 
filled by non-linear periodic waves with moving interfaces between the regions, see Figure \ref{T}.
Thus we also call them wave fronts.
\begin{figure}[h!]
\centering
\includegraphics[angle=0,scale=.5]{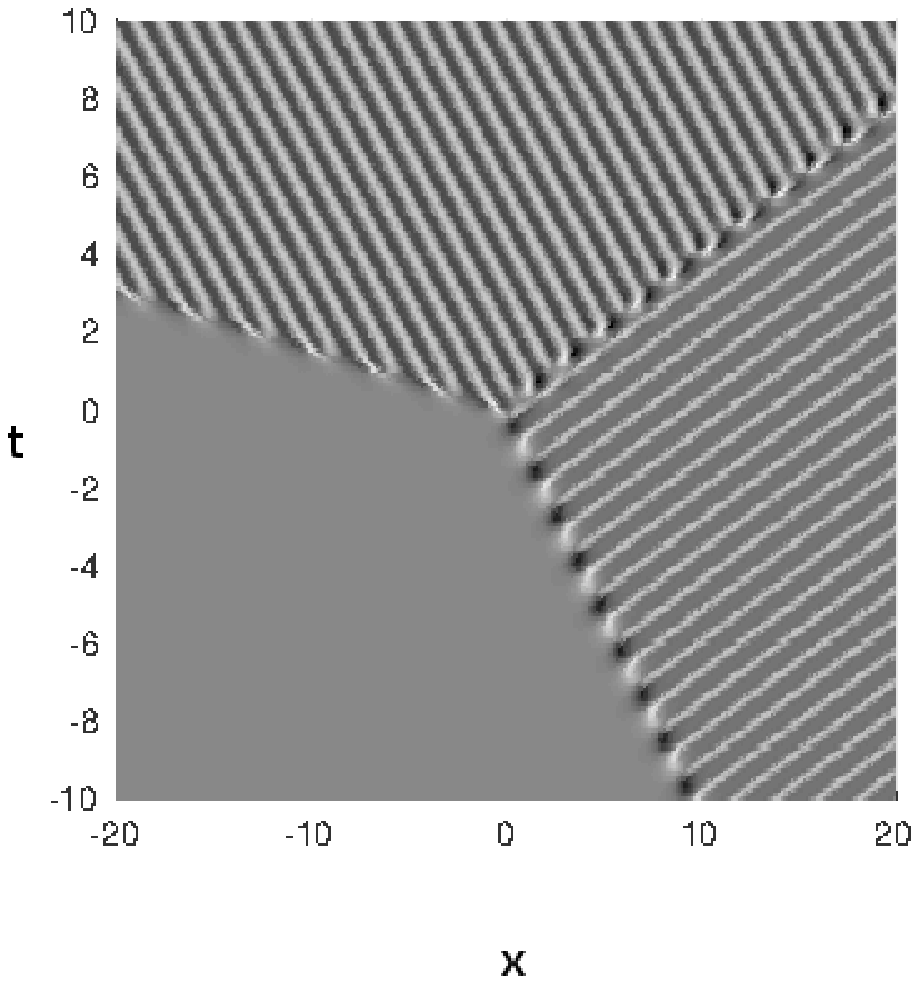}\includegraphics [ 
angle=0 
, scale=0.5]{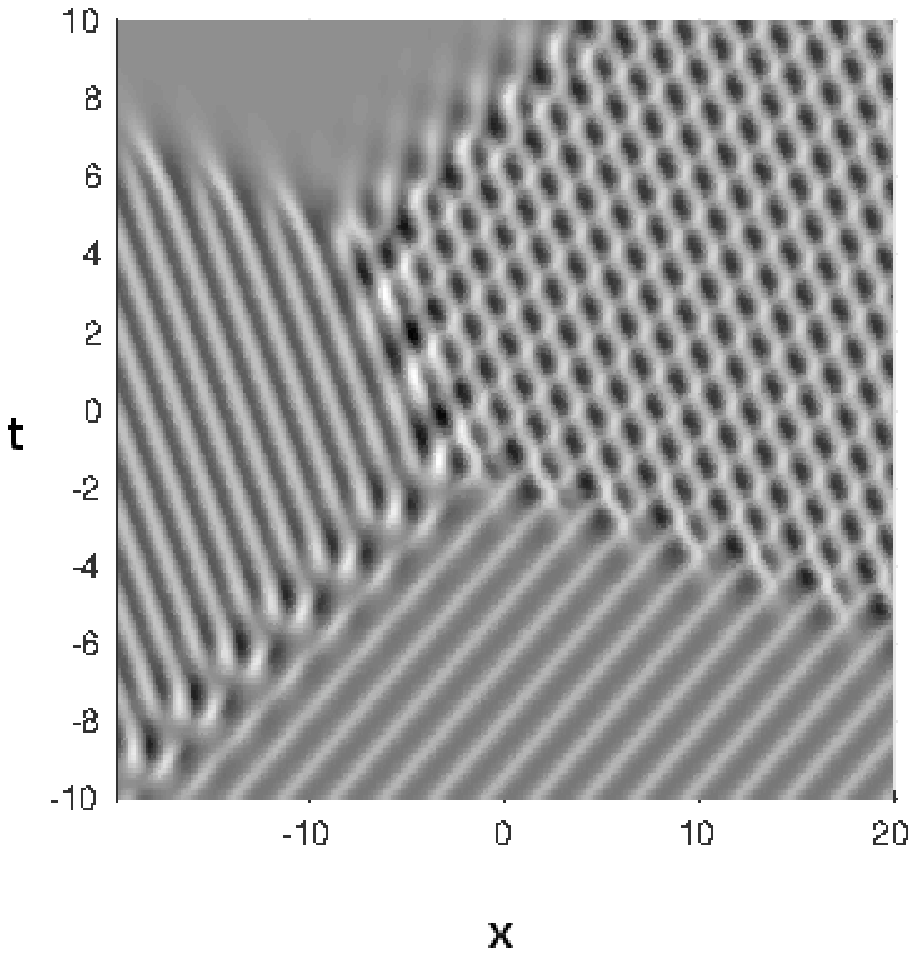}
\caption{Density plots of $\phi^{(1)}(x,t)$ for $N=5$  of rank 1 and rank 2 
kink 
solutions}
\label{T}
\end{figure}
Breathers correspond to a 
cascade of soliton decays and fusions. The density plots of a breather 
solution in the $(x,t)$-plane resemble soliton webs of the KP 
equations in the $(x,y)$-plane for a fixed moment of time, see Figure \ref{Tb}.  
\begin{figure}[h!]
\centering
\includegraphics [angle=0,scale=0.45]{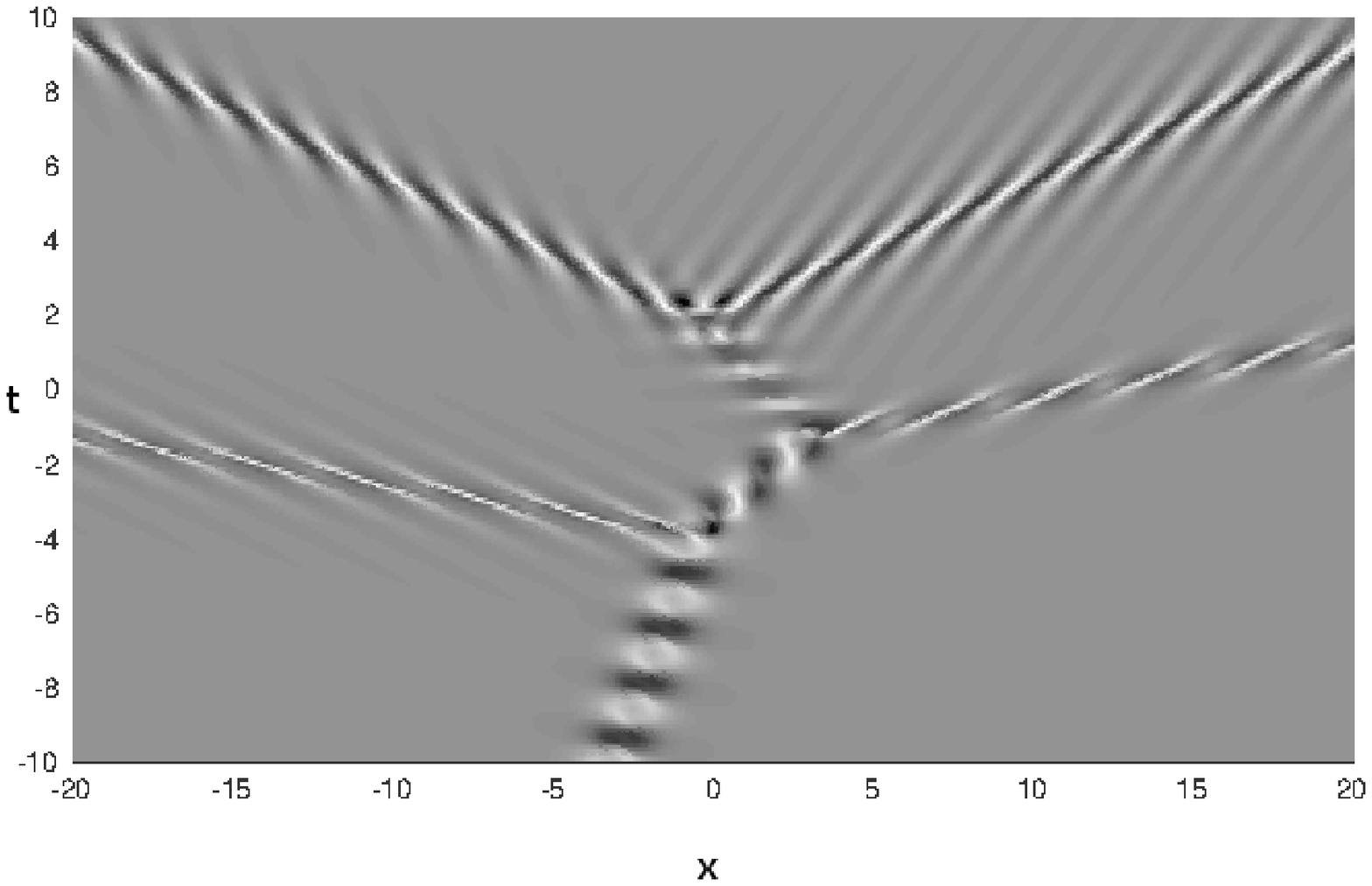} 
\includegraphics[angle=0,scale=0.45]{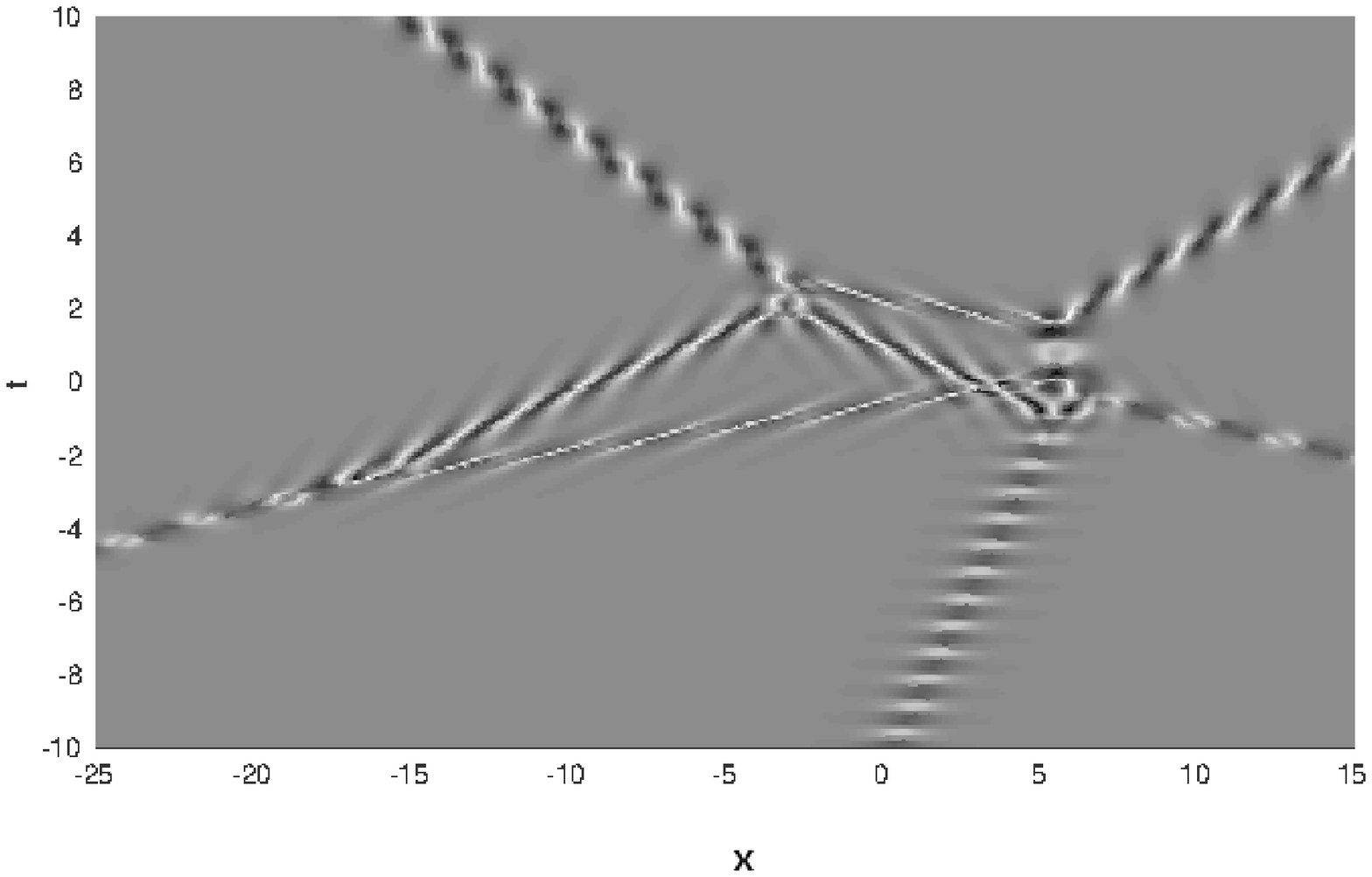}
\caption{Density plots of $\phi^{(1)}(x,t)$ for $N=5$ of rank 1 and rank 2 
breather
solutions}\label{Tb}
\end{figure}
In the paper we give explicit and detailed derivation for these two types of 
soliton solutions of arbitrary rank for the 
two dimensional 
Volterra system (\ref{2+1}) and give a complete classification of rank 
$1$ kink and breather solutions.

The arrangement of this paper is as follows: In Section \ref{sec2}, we recall the Lax 
representation for equation (\ref{2+1}) and its  reduction group. In Section \ref{sec3}, 
we discuss the dressing method in the presence of the reduction group. 
We explicitly derive both kink 
and breather solutions for equation (\ref{2+1}) on the trivial 
background using the dressing method. All exact solutions emerging from the 
dressing method can be written in the form
\begin{eqnarray}\label{tt}
\phi^{(i)}=\frac{1}{2}\ln \frac{\tau_{i-1} \tau_{i+1}}{\tau_i^2}.
\end{eqnarray}
Under a certain continuous limit $N\to \infty$, equation (\ref{2+1}) converges to 
the well known KP equation. In this limit $\tau_i$ in (\ref{tt}) can be related to 
the Hirota $\tau$-function for the bilinear form of the KP equation 
\cite{hirota}. 
In Section \ref{sec4}, we classify both kink and breather solutions of rank $1$ 
according to the eigenspaces of the constant matrix $\Delta$ in the Lax 
operators of equation (\ref{2+1}). For rank $1$ kink solutions, we start with 
a description of all possible rank 1 kink solutions in the cases of $N=3, 4$ and 
further prove the general results for arbitrary dimensions. For rank $1$ breather solutions we 
present some typical configurations and the general result on the number 
of possible distinct configurations. Our definition soliton graphs based on 
tropicalisation is motivated by \cite{kodama2011kp} although we do not have 
structure associated with Wronskian of solutions.
In the Conclusion we summarise our results and discuss the feasibility of a 
full classification of 
higher rank solutions. 

\section{Lax representation and the dihedral reduction group}\label{sec2}
Let us consider general matrix operators of the form
\begin{eqnarray}
\hat{L}(\lambda) = D_x + \mathbf{X}(x,t,\lambda), \qquad
\hat{M}(\lambda)= D_t + \mathbf{T}(x,t,\lambda),\label{gL}
\end{eqnarray}
where $D_x$ and $D_t$ are total derivatives in $x$ and $t$ respectively, $\lambda\in \bbbc$ 
is a spectral parameter, 
$\mathbf{X}$ and $\mathbf{T}$ are $N\times N$ traceless matrices
\begin{eqnarray*} \mathbf{X}(x,t,\lambda) &=&  \mathbf{X}_0 + \lambda^{-1} 
\mathbf{U} + \lambda \overline{\mathbf{U}}\\
\mathbf{T}(x,t,\lambda)& =& \mathbf{T}_0 + \lambda^{-1} \mathbf{A} + \lambda \overline{\mathbf{A}} +
\lambda^{-2} \mathbf{B} + \lambda^2 \overline{\mathbf{B}}
\end{eqnarray*}
and the matrices $\mathbf{X}_0, \mathbf{U},\overline{\mathbf{U}},\mathbf{T}_0 
,\mathbf{A}, \overline{\mathbf{A}}, \mathbf{B}$ and $\overline{\mathbf{B}} $ are 
functions of $x$ and $t$.
The compatibility condition $[\hat{L}, \hat{M}]=0$, that is,
\begin{equation}
  \mathbf{T}_x - \mathbf{X}_t + [\mathbf{X},\mathbf{T}] = 0\label{compcond}
\end{equation}
gives $7(N^2-1)$ partial differential equations (coefficients of  
$\lambda^{-3},\ldots ,\lambda^{3}$) for  $7N^2$ matrix entries. 

We define a group of  automorphisms generated by the following two transformations
for an operator ${\bf d}(\lambda)$: the first one is
\begin{equation}
\ii: {\bf d}(\lambda)\mapsto -{\bf d}^{\dagger}(\lambda^{-1}),
\label{redgroup1}
\end{equation}
where ${\bf d}^{\dagger}$ is the adjoint operator of operator ${\bf d}$.  The 
second one is
%\begin{equation}
%r:{\bf d}(\lambda)\rightarrow \bf d^{\star} (\lambda^{\star}),
%\label{redgroup2}
%\end{equation}
%where $^\star$ means its complex conjugate. The last one is
\begin{equation}
s:{\bf d}(\lambda)\mapsto Q {\bf d}(\omega^{-1} \lambda) Q^{-1},\qquad Q=\diag 
(\omega^i),\quad
\omega=\exp\frac{2\pi {\rm i}}{N} .
\label{redgroup3}
\end{equation}
These two non-commuting transformations satisfy
\begin{eqnarray*}
 \ii^2=s^N = {\rm id},\quad \ii s\ii = s^{-1}  
\end{eqnarray*}
and therefore generate the dihedral group denoted by $\bbbd_N$.
We call the 
group generated by transformations (\ref{redgroup1}) and (\ref{redgroup3}) the 
reduction group \cite{mik79,mik81, LM04, LM05, sara, bury}.
Note that the transformation $\ii$ (\ref{redgroup1}) is an outer automorphism of 
the Lie algebra 
$sl(n)$ over the Laurent polynomial ring $\mathbb{C}[\lambda,\lambda^{-1}]$.

\begin{Pro} If the linear operators (\ref{gL})  are  invariant with respect to the 
reduction group $\bbbd_N$, then they can be written in the 
form
 \begin{eqnarray}
 \hat{ L}&=&D_x + \lambda^{-1} \mathbf{u}\Delta - \lambda \Delta^{-1}\mathbf{u}\label{L1}\\
\hat{M}&=&D_t + \lambda^{-1} \mathbf{a}\Delta - \lambda
\Delta^{-1}\mathbf{a} + \lambda^{-2}  \mathbf{b}\Delta ^2 - \lambda ^2\Delta^{-2}\mathbf{b} \label{M1}
\end{eqnarray}
where $\mathbf{u},\mathbf{a},\mathbf{b}$ are diagonal matrices and 
$\Delta$ is the shift operator given by 
\begin{equation}
 \label{Delta}
 \Delta =\begin{pmatrix}
0& 1 & 0 & \dots & 0 \\
0 & 0 & 1 & \dots & 0 \\
0& 0& 0& \dots&1 \\
1 & 0&0 &  \dots& 0
\end{pmatrix}.
\end{equation}
\end{Pro}
\begin{proof}
%The invariance under the transformation $r$ implies that 
%matrices $\mathbf{X}(x,t,\lambda),\mathbf{T}(x,t,\lambda) $ are real for real 
%spectral $\lambda$. 
It follows from the invariance under the transformation 
$\ii$ (\ref{redgroup1}), that is,
\[ (\hat{L}(\lambda) ,\hat{M}(\lambda)) =(-\hat{L}^{\dagger} (\lambda^{-1}) 
,-\hat{M}^{\dagger}(\lambda^{-1})),\]
that 
\begin{eqnarray}\label{invs}
\overline{\mathbf{U}}=-\mathbf{U}^{\tr},\quad \overline{\mathbf{A}}=- 
\mathbf{A}^{\tr},\quad \overline{\mathbf{B}}=-\mathbf{B}^{\tr}, \quad  
\mathbf{X}_0=-\mathbf{X}_0^{\tr}, \quad  
\mathbf{T}_0=-\mathbf{T}_0^{\tr}, 
\end{eqnarray}
where $^{\tr}$ denotes matrix transposition.
The invariance under the transformation $s$ implies that 
$\mathbf{X}_0,\mathbf{T}_0$ are diagonal and the matrices 
$\mathbf{U},\mathbf{A}$ and $\mathbf{B}$ are of the form
\[ \mathbf{U}=\mathbf{u}\Delta,\quad \mathbf{A}=\mathbf{a}\Delta,\quad 
\mathbf{B}= \mathbf{b}\Delta ^2,\]
where  $\mathbf{u},\mathbf{a}$ and $\mathbf{b}$ are  diagonal matrices and $\Delta$ 
is given by (\ref{Delta}).
Combining with (\ref{invs}), we get
$\mathbf{X}_0=\mathbf{T}_0=0$ and further the expressions (\ref{L1}) and 
(\ref{M1}) in the 
statement.
\end{proof}

Let 
$\mathbf{u}=\diag(u^{(i)}), \mathbf{a}=\diag(a^{(i)})$ and  $\mathbf{b}=\diag(b^{(i)}
)$. Then the compatibility condition of Lax operators 
(\ref{L1}) and (\ref{M1}) leads to $3N$ equations 
\begin{eqnarray}
 && u^{(i)} b^{(i+1)}-b^{(i)} u^{(i+2)}=0\, , \label{lam3}\\
&& D_x(b^{(i)})+u^{(i)} a^{(i+1)}-a^{(i)} u^{(i+1)}=0\, ,\label{lam2}\\
&& D_t (u^{(i)})=D_x(a^{(i)})-u^{(i-1)} b^{(i-1)}+b^{(i)} u^{(i+2)}\label{lam1}
\end{eqnarray}
in $3N$ variables $u^{(i)}, a^{(i)}$ and $b^{(i)}$, where $i=1, \cdots, N$. In 
this paper we shall assume that all upper indices, taking value from $1$ to $N$, are counted modulo $N$, 
if not stated otherwise.
Take \begin{eqnarray}\label{u}
\mathbf{u}=\diag (\exp(\phi^{(1)}),\ldots,\exp(\phi^{(N)}))\, ,\quad 
\mathbf{a}=\diag (\theta^{(1)}\exp 
(\phi^{(1)}),\ldots,\theta^{(N)}\exp(\phi^{(N)})) .
\end{eqnarray} 
It follows from (\ref{lam3})-(\ref{lam1}) that we can set 
$b^{(i)}=\exp\left(\phi^{(i)}+\phi^{(i+1)}\right)$ and  $\sum_{i=1}^N 
\phi^{(i)}=\sum_{i=1}^N\theta^{(i)}=0$ without losing 
generality. In the variables $\phi^{(i)}$ and $\theta^{(i)}$ the system of equations 
(\ref{lam3})-(\ref{lam1}) leads to the 2-dimensional generalisation of the
Volterra system (\ref{2+1}). The corresponding operators (\ref{L1}) 
and (\ref{M1}) can be expressed as the invariant operators
under the reduction group $\bbbd_N $, namely,
\begin{eqnarray}\label{laxv}
\begin{array}{ll}
 \cL=D_{x}+U,& U= \lambda^{-1}\mathbf{u}
\Delta-\lambda \Delta^{-1} \mathbf{u},\\
\cM=D_{t}+V,& V= 
\lambda^{-2}\mathbf{u}\Delta\mathbf{u}\Delta+\lambda^{-1}\mathbf{a}
\Delta-\lambda \Delta^{-1}\mathbf{a}-\lambda^{2}  \Delta^{-1} 
\mathbf{u}\Delta^{-1}\mathbf{u} ,
\end{array}
\end{eqnarray}
where $\mathbf{u}$ and  $\mathbf{a}$ are defined by (\ref{u}) and the matrix $\Delta$ 
is given by (\ref{Delta}).
The condition of commutativity of these operators
\begin{equation}\label{zeroc}
 [\cL, \cM]=D_x(V)-D_{t}(U)+[U, V]=0
\end{equation}
leads to the $2$--dimensional generalisation of the Volterra lattice 
(\ref{2+1}) \cite{mik79, mik81}. 
This is often called a zero curvature representation or Lax representation of 
equation (\ref{2+1}).
These two operators, $ \cL$ and $\cM$, are conventionally called the Lax 
pair.

If we assume that the functions $\phi^{(k)},\theta^{(k)}$ in (\ref{u}) are 
real, then the Lax operators $\cL,\cM$ are also invariant with respect to 
transformation
\begin{equation}
 r\, :\, \cL(\lambda)\mapsto \cL^\star (\lambda^\star),\quad 
\cM(\lambda)\mapsto \cM^\star (\lambda^\star),
\end{equation}
where $^\star$ means its complex conjugate.
This transformation extends the dihedral group. The group generated by 
$s,\ii,r$ is isomorphic to $\bbbz_2\times\bbbd_N$.

\section{Rational dressing method for the generalised Volterra lattice}\label{sec3}
In this section, we use the method of rational dressing \cite{zash, mr80c:81115, 
mik81} to construct new exact solutions of (\ref{2+1}) 
starting from a known exact solution. Let us 
denote by
$U_0, V_0$ the matrices $U,V$ in which $\phi^{(i)}$ are 
replaced by the known exact solution $\phi_0^{(i)}, i=1, \cdots, N$ of (\ref{2+1}), that is,
$$
U_0= \lambda^{-1}\mathbf{u}_0\Delta-\lambda \Delta^{-1}\mathbf{u}_0, \qquad
V_0= \lambda^{-2}\mathbf{u}_0\Delta\mathbf{u}_0\Delta+\lambda^{-1}\mathbf{a}_0
\Delta-\lambda \Delta^{-1}\mathbf{a}_0-\lambda^{2}  \Delta^{-1} 
\mathbf{u}_0\Delta^{-1}\mathbf{u}_0 .
$$
The corresponding overdetermined linear system  
\begin{equation}
\label{naked}
\cL_0 \Psi_0=(D_x+U_0) \Psi_0=0, \qquad 
\cM_0 \Psi_0=(D_t+  V_0)\Psi_0=0 
\end{equation}
has a common fundamental solution $\Psi_0(\lambda,x,t)$.
%invariant under transformations (\ref{redgroup1})--(\ref{redgroup3}). 
Following \cite{zash, mr80c:81115} we shall 
assume that the fundamental solution $\Psi(\lambda,x,t)$ for the new 
(``dressed'') linear problems 
\begin{equation}
\label{dressed}
\cL \Psi=(D_x+U) \Psi=0, \qquad 
\cM \Psi=(D_t+ V)\Psi=0 
\end{equation}
is of the form 
\begin{eqnarray}\label{dress}
 \Psi=\Phi(\lambda) \Psi_0, \quad \det \Phi \not= 0,
\end{eqnarray}
where the dressing matrix $\Phi(\lambda)$ is assumed to be rational in the 
spectral parameter $\lambda$ and to be invariant with respect to the symmetries
\begin{eqnarray}
&& \Phi^{-1}(\lambda^{-1})=\Phi^{\tr}(\lambda);\label{Inphi1}\\
&& Q \Phi(\omega^{-1} \lambda) Q^{-1}=\Phi(\lambda). 
\label{Inphi2}
\end{eqnarray}
Conditions (\ref{Inphi1}) and (\ref{Inphi2}) guarantee that the corresponding Lax operators 
$\mathcal{L}$ and $\mathcal{M}$ are invariant under transformations
(\ref{redgroup1}) and (\ref{redgroup3}).

We are going to derive real solutions for the real equation. Thus we also 
require
\begin{equation}
\Phi^{\star}(\lambda^{\star})=\Phi(\lambda).
\label{redgroup2}
\end{equation} 
It follows from (\ref{naked}), (\ref{dressed}) and (\ref{dress}) that 
\begin{eqnarray}
&& \Phi (D_x+U_0) \Phi^{-1}=U; \label{chix}\\
&& \Phi (D_t+V_0) \Phi^{-1}=V. \label{chit}
\end{eqnarray}
These equations enable us to specify the form of the dressing matrix $\Phi$ and 
construct the corresponding ``dressed'' solution $\phi^{(i)}$ of equation 
(\ref{2+1}).

Let us consider the most
trivial case when the dressing matrix $\Phi$ does not 
depend on the spectral parameter $\lambda$. In this case the dressing matrix does not result
any new solutions.  
\begin{Pro}
Assume that $\Phi$ is a $\lambda$ independent dressing matrix for the two 
dimensional generalisation of the Volterra lattice (\ref{2+1}) and $\phi_0^{(i)}$ is a real solution. If it satisfies
 (\ref{Inphi1})--(\ref{redgroup2}),  
then the
matrix $\Phi=\pm I_N$, where $I_N$ is the $N\times N$identity matrix and the real solutions on the background $\phi_0^{(i)}$ are
$\phi^{(i)}=\phi_0^{(i)}$.
\end{Pro}
\begin{proof} 
Under the assumption that the dressing matrix $\Phi$ is
independent of the spectral parameter $\lambda$, it follows from (\ref{chix}) 
and (\ref{chit}) that
\begin{eqnarray}\label{consd}
D_x \Phi=0; \ D_t \Phi=0;\ \Phi \mathbf{u}_0 \Delta=\mathbf{u} \Delta \Phi; \ 
\Phi \Delta^{-1} \mathbf{u}_0 =\Delta^{-1} \mathbf{u} \Phi;\  \Phi \mathbf{a}_0 
\Delta=\mathbf{a} \Delta \Phi; \ 
\Phi \Delta^{-1} \mathbf{a}_0 =\Delta^{-1} \mathbf{a} \Phi.
\end{eqnarray}
It is obvious that the matrix $\Phi$ is independent
of $x$ and $t$.
Since $\Phi$ satisfies (\ref{Inphi1})--(\ref{redgroup2}) , we 
deduce that matrix $\Phi$ is real,  $\Phi \Phi^{\tr}=I_N$ and $\Phi$ is diagonal. 
Thus the constant matrix $\Phi$ has $\pm 1$ on the diagonal.
Substituting such $\Phi$ into (\ref{consd}), we get $\Phi=\pm I_N$ and $\phi^{(i)}=\phi_0^{(i)}$
since both $\phi^{(i)}$ and $\phi_0^{(i)}$ are real.
\end{proof}

A $\lambda$--dependent dressing matrix $\Phi(\lambda)$, which is
invariant with respect to the symmetries (\ref{Inphi1})--(\ref{redgroup2}) has 
poles at the orbits of 
the reduction group.
Simplest ``one soliton'' dressing corresponds to the cases when the matrix
$\Phi(\lambda)$ has only simple poles belonging to a single 
orbit. 

Notice that if  $\Phi(\lambda)$ is invariant under the reduction group, so is $\Phi^{-1}(\lambda)$.  Instead of specifying the poles for 
$\Phi(\lambda)$, we first specify the poles for $\Phi^{-1}$, and then 
determine $\Phi$ from the relation (\ref{Inphi1}).
If $\Phi^{-1}(\lambda)$ has a pole at the
point $\mu$, then by the second relation (\ref{Inphi2}) (for $\Phi^{-1}$) it 
must also have poles at
the points $\omega^{-1}\mu$, $\omega^{-2}\mu$, $\dots$, $\omega^{-(N-1)}\mu$.
Due to (\ref{redgroup2}),
there are two non-trivial cases:
\begin{enumerate}
 \item[(1)] The matrix $\Phi^{-1}(\lambda)$ has $N$ poles: 
 \begin{itemize}
  \item[(i)] for arbitrary $N$ poles at $ \omega^{-k} \mu,\ 
k=0, \cdots, N-1,\ \mu\ne 0,\ \mu\neq \pm1,\ \mu=\mu^\star;$
\item[(ii)] when $N$ is even, i.e. $N=2m$, poles at $ \omega^{-k} \mu,\ 
k=0, \cdots, N-1,\ \mu=\nu \exp(\frac{\pi{\rm i}}{N}),  \nu\in \bbbr, 
\nu\notin \{\pm 1, 0\} ;$
 \end{itemize}
\item[(2)] The matrix $\Phi^{-1}(\lambda)$ has $2N$ complex poles at $\omega^{-k} 
\mu$ and
$\omega^{-k} \mu^\star$, where $k=0,\cdots, N-1$ and $$
\mu\in\bbbc,\ |\mu|\ne 1,\ \mu\ne \omega^k  \mu^\star,\ k=1,\ldots,N.$$ 
\end{enumerate} 
Note that when $N=2m+1$ is odd, the case (i) in (1) includes the case (ii) 
since $$\omega^n \mu=\omega^l \nu \exp(\frac{\pi{\rm i}}{N})=-\nu\in \bbbr.$$
The extra conditions on $\mu$ is to ensure that the poles for 
$\Phi$ and $\Phi^{-1}$ are distinct.
These cases correspond to the ``kink'' and ``breather'' solutions respectively. 

%A more general case in which the poles belong to a finite union of orbits
%corresponds to a multi-soliton solution. Using this approach we can construct
%explicitly multi kink-breather solutions 
%and analyse the results of kink-kink, breather-breather and kink-breather collisions.  

The explicit forms of the matrix $\Phi^{-1}(\lambda)$ corresponding the above two 
cases
and invariant with respect to the symmetries (\ref{Inphi2}) and  (\ref{redgroup2}) are 
\begin{eqnarray*}
({\rm 1})&& ({\rm i})\ \ 
\Phi^{-1}(\lambda)=C+\sum_{k=0}^{N-1}\frac{Q^{-k} A 
Q^{k}}{\lambda 
\omega^k- \mu}, \quad A=A^\star, \quad \mu=\mu^\star, \ \mu\ne 0,\ \mu\neq \pm1; \\
&& ({\rm ii})\ \ \Phi^{-1}(\lambda)=C+\sum_{k=0}^{2m-1}\frac{Q^{-k} A 
Q^{k}}{\lambda 
\omega^k- \mu}, \quad N=2m,\ A^\star=\omega^{-1} Q^{-1} A Q, \ \mu=\nu 
\exp(\frac{\pi{\rm i}}{2m}), \ \nu\in \bbbr, \nu\notin \{\pm 1, 0\};\\
({\rm 2})&&\Phi^{-1}(\lambda)=C+\sum_{k=0}^{N-1}\left(\frac{Q^{-k} A 
Q^{k}}{\lambda \omega^k- \mu}+\frac{Q^{-k} A^\star
Q^{k}}{\lambda \omega^k- \mu^\star}\right),\quad |\mu|\ne 1,\ \mu\ne \omega^k  
\mu^\star,\ k=1,\ldots,N, 
\end{eqnarray*}
where $C$ and $A$ are $\lambda$-independent matrices of size $N\times N$. 
Moreover, to satisfy (\ref{Inphi2}) and (\ref{redgroup2}), we have 
$C=QCQ^{-1}$ implying $C$ is diagonal and $C=C^\star$. Hence we assume that
\begin{eqnarray}\label{matc}
 C=\diag(c_1, \cdots, c_N),
\end{eqnarray}
where $c_i, i=1,\cdots, N$ are real functions of $x$ and $t$.

We now derive the conditions on the matrices $A$ and $C$ such that $\Phi^{-1}(\lambda)$ satisfies (\ref{Inphi1}). In this case, we have
$\Phi(\lambda)=\left(\Phi^{-1}(\lambda^{-1})\right)^{\tr}$. It follows that
\begin{eqnarray}
({\rm 1})&& \Phi(\lambda)=C+\sum_{k=0}^{N-1}\frac{Q^{k} A^{\tr} 
Q^{-k}}{\lambda^{-1} 
\omega^k- \mu}, \quad A=A^\star, \ \mu=\mu^\star,
 \mbox{or} \ N=2m,\ \mu=\nu 
\exp(\frac{\pi{\rm i}}{N}),
 A^\star=\omega^{-1} Q^{-1} A Q ; \label{phik}\\
({\rm 2})&&\Phi(\lambda)=C+\sum_{k=0}^{N-1}\left(\frac{Q^{k} A^{\tr} 
Q^{-k}}{\lambda^{-1} \omega^k- \mu}+\frac{Q^{k} A^{\star\tr}
Q^{-k}}{\lambda^{-1} \omega^k- \mu^\star}\right).\label{bphi}
\end{eqnarray}

\begin{Pro}\label{pro0}
Let $I_N$ denote the $N\times N$ identity matrix.
The dressing matrix satisfies (\ref{Inphi1}) if and only
if matrix $A$ 
and the real diagonal matrix $C$ satisfy the relations:
\begin{eqnarray}
&&\lim_{\lambda \to\infty} \Phi(\lambda) =C^{-1}; \qquad
\Phi(\mu) A=0.\label{cond}
\end{eqnarray}
\end{Pro}
\begin{proof} We verify that the above $\Phi(\lambda)$ is indeed the inverse matrix of $ \Phi^{-1}(\lambda)$
by checking $\Phi(\lambda)\Phi^{-1}(\lambda)=I_N$.
The product is a rational matrix function of $\lambda$. Taking the limit at $\lambda =\infty$ we obtain $\lim_{\lambda \to\infty} \Phi(\lambda) C=I_N$,
which implies the first equation in (\ref{cond}).
Under the assumptions on $\mu$, the poles of both $\Phi$ and $\Phi^{-1}$ are simple and distinct. 
Therefore, $\Phi(\lambda)\Phi^{-1}(\lambda)$ has $2N$ simple poles. Requesting the vanishing of the residue at $\lambda=\mu$ we 
obtain the second equation  in 
(\ref{cond}). The residues at all other points of the reduction group orbit 
will vanish due to the manifest invariants of the expression with respect to 
the reduction group.
\end{proof}

We now investigate the conditions (\ref{chix}) and (\ref{chit}) for
$\Phi(\lambda)$ which follow from the fact that it is a dressing matrix.
Notice that $\Phi, U_0$ and $\Phi^{-1}$  have distinct simple poles. Thus the left hand side of (\ref{chix}) has simple poles only.
%Due to the reduction group, we only need to compare the residues of both sides at $\lambda=\mu$ and $\lambda=\infty$.
We first compare the residues at the pole $\mu$. It follows that
\begin{eqnarray}\label{xker}
 \lim_{\lambda\to\mu} (\lambda-\mu)\Phi (D_x+U_0(\lambda)) \Phi^{-1}=\Phi(\mu) (D_x+U_0(\mu)) A=0 .
\end{eqnarray}
Thus  $(D_x+U_0(\mu))A\in\ker \Phi(\mu)$.  In a similar way it follows from the condition (\ref{chit}) that 
\begin{eqnarray}\label{tker}
\Phi(\mu) (D_t+V_0(\mu)) A=0 , \quad \mbox{that is},\quad (D_t+V_0(\mu))A\in\ker \Phi(\mu).
\end{eqnarray}
We compute the residue at $\lambda=\infty$ of both sides of (\ref{chix}) and we have
\begin{eqnarray*}
 \lim_{\lambda\to\infty} \frac{1}{\lambda}\Phi(\lambda) (D_x+U_0(\lambda)) \Phi^{-1}(\lambda)
 =-\lim_{\lambda\to\infty} \Phi(\lambda)\Delta^{-1}\mathbf{u}_0 C =-\Delta^{-1} \mathbf{u} .
\end{eqnarray*}
Using (\ref{cond}), we have
\begin{eqnarray}\label{infty}
 C^{-1} \Delta^{-1}\mathbf{u}_0 C =\Delta^{-1} \mathbf{u} .
\end{eqnarray}
This formula provides the relation between $\mathbf{u}_0$ and $\mathbf{u}$.
However, it is required to determine the diagonal matrix $C$ in the dressing matrix $\Phi$, which
depends on the choice of the form for $\Phi(\lambda)$. We will determine it 
when we compute the kink and breather solutions.

In the following two sections, we construct the exact solutions starting with the trivial solution $\phi^{(i)}_0=0, i=1, \cdots, N$
for the equation (\ref{2+1}). In this case, $\mathbf{u}_0=I_N$ and $\mathbf{a}_0=0$. Thus we have
$$\cL_0=D_x+\lambda^{-1} \Delta-\lambda \Delta^{-1}; \qquad
\cM_0=D_t+\lambda^{-2} \Delta^2-\lambda^2 \Delta^{-2} .
$$
It is easy to see that in this case the fundamental solution for (\ref{naked})
is
\begin{eqnarray}\label{psi0}
\Psi_0(x,t,\lambda)=\exp((-\lambda^{-1} \Delta+\lambda \Delta^{-1}) 
x-(\lambda^{-2} \Delta^2-\lambda^2 \Delta^{-2}) t) .
\end{eqnarray}
The matrix $\Psi_0(x,t,\lambda)$ obviously 
satisfies the reduction group symmetry conditions  
(\ref{Inphi1})--(\ref{redgroup2}).

On the trivial background (\ref{infty}) becomes $C^{-1} \Delta^{-1} 
C =\Delta^{-1} \mathbf{u}$. It follows that
\begin{eqnarray}\label{sol}
 \exp(\phi^{(j)})= \frac{c_j}{c_{j+1}}.
\end{eqnarray}
We will use this to construct solutions for (\ref{2+1}) on the trivial background 
later.

\subsection{Kink solutions}\label{Kink}
In this section, we derive the explicit formula for kink solutions of 
arbitrary ranks. As we discussed before, a kink solution for equation 
(\ref{2+1}) corresponds to the invariant dressing matrix 
with $N$ simple poles. It is of the form (\ref{phik}). There is a 
difference between the dimension $N$ being even or odd. If $N$ is odd, there is 
only one case when $A=A^\star$, $\mu\in  \mathbb{R}$ and $\mu\notin\{\pm1,0\}$. 
If $N$ is even, there 
is an extra case when $\mu=\nu \exp(\frac{\pi \i}{N})$ with $\nu \in\bbbr$ and 
$A^\star=\omega^{-1} Q^{-1} A Q$. This difference is caused by the real requirement (\ref{redgroup2}).
Hence, we first derive the expressions for $\mu$ and $A$ and then add the conditions for them.

For all cases,  the matrix $C$ defined by (\ref{matc}) is diagonal with
 real functions  $c_i, i=1,\cdots, N$ on the diagonal.
 Moreover, it follows from Proposition \ref{pro0} that
\begin{eqnarray}
&&\lim_{\lambda \to\infty} \Phi(\lambda)=C-\frac{1}{\mu}\sum_{k=0}^{N-1}Q^{k} 
A^{\tr} Q^{-k}
=C-\mu^{-1} N A_{\diag}=C^{-1};\label{AC1}\\
&& \Phi(\mu)A=(C+\sum_{k=0}^{N-1}\frac{Q^{k} A^{\tr} Q^{-k}}{\mu^{-1} \omega^k- 
\mu})A=0. \label{AC2}
\end{eqnarray}
\begin{Pro}\label{pro3}
If the matrix $A$ is nondegenerate in the dressing 
matrix 
(\ref{phik}), the real solutions for (\ref{2+1}) are
$\phi^{(i)} =0$ on the trivial background  
$\phi_0^{(i)}=0$, where $i=1, \cdots, N$.
\end{Pro}
\begin{proof} If $\det A\ne 0$, from (\ref{AC2}) we get
$$
C+\sum_{k=0}^{N-1}\frac{Q^{k} A^{\tr} Q^{-k}}{\mu^{-1} \omega^k- \mu}=0,
$$
which implies that matrix $A$ is diagonal since $C$ is 
diagonal, and thus 
$$
C =-\sum_{k=0}^{N-1}\frac{Q^{k} A^{\tr} Q^{-k}}{\mu^{-1} \omega^k- \mu}=-\mu 
\sum_{k=0}^{N-1}  \frac{1}{\omega^{k}-\mu^2} A=\frac{N \mu^{2N-1}}{\mu^{2N}-1} 
A\, .
$$
Here we used Lemma 2 proved in the appendix of paper \cite{MPW} and stating that for $x^N\ne 1$ and $\omega=\exp(\frac{2\pi i}{N})$
 \begin{equation}
\label{old}
  \sum_{j=0}^{N-1} \frac{\omega^{lj}}{x-\omega^j}=\frac{N
x^{(l-1)\!\!\! \mod\! N}}{x^N-1}.
 \end{equation}

Substituting this into (\ref{AC1}), we have 
$$
C^2-\frac{N}{\mu}AC=\frac{1}{\mu^{2N}}C^2=I_N, \quad \mbox{that is,} \quad 
C^2=\mu^{2N}I_N.
$$
So we have $c_i=\pm \mu$ when $\mu\in \mathbb{R}$, or $c_i=\pm \nu $ when 
$\mu=\nu\exp(\frac{\pi \i}{N})$ with $\nu\in\bbbr$.  It follows from 
(\ref{sol}) that
$$\exp(\phi^{(i)})= \frac{c_i}{c_{i+1}}= 1 
$$
since $\phi^{(i)}\in \bbbr$.
Thus we obtain the trivial solutions for $\phi^{(i)}$ as given in the 
statement.
\end{proof}
In order to construct solutions which depend on $x$ and $t$, we consider the case where matrix $A$ is of rank $r\leq N-1$,
and hence can be represented in the form 
$$
A=\bn \bm^{\tr},
$$
where $ \bn $ and $\bm$ are both $N\times r$ matrices of rank $r$. We define the rank of the kink solutions as the rank of matrix $A$ in the dressing matrix.
In this case, equation (\ref{AC2}) becomes
\begin{eqnarray}
(C+\sum_{k=0}^{N-1}\frac{Q^{k} \bm \bn^{\tr} Q^{-k}}{\mu^{-1} \omega^k- 
\mu})\bn=0 \label{mn} 
\end{eqnarray}
since $\bm$ is $N\times r $ matrix of rank $r$. We can use it to solve for $\bm$ in terms of $\bn$. Further we can also determine matrix $C$ in terms of $A$
using (\ref{AC1}).
\begin{Rem}\label{rem1}
The dressing matrix $\Phi(\lambda)$ (\ref{phik}) is 
parametrised by a matrix $\bn$ lying on a Grassmannian. Indeed, assume that we 
get $\bm=F(\bn)$ from (\ref{mn}). If we make a change $\hat \bn=\bn W,$ where 
$W$ is an invertible $r\times r$ matrix, the corresponding solution $\hat 
\bm=F(\bn) (W^{-1})^{\tr}= \bm (W^{-1})^{\tr}$. Therefore the matrix
 $$\hat A=\hat \bn \hat \bm^{\tr}=\bn W W^{-1} \bm^{\tr}=\bn \bm^{\tr}=A.$$
\end{Rem}
It follows from (\ref{mn}) that $\bn \in\ker \Phi(\mu)$. From (\ref{xker}) we 
get
$$
0=\Phi(\mu)  (D_x+U_0(\mu))\bn \bm^{\tr}=\Phi(\mu)  (D_x+U_0(\mu))(\bn) 
\bm^{\tr}+\Phi(\mu)(\bn \bm_x^{\tr}).
$$
This implies that $(D_x+U_0(\mu))(\bn) \in \ker \Phi(\mu)$. Thus there exists a 
scalar function $\gamma(x,t)$ such that
\begin{eqnarray}\label{nxker}
 (D_x+U_0(\mu))(\bn)=\gamma(x,t) \bn.  
\end{eqnarray}
Similarly we can show that there exists a scalar function  $\delta(x,t)$ such 
that 
\begin{eqnarray}\label{ntker}
(D_t+V_0(\mu))(\bn)=\delta(x,t)\bn .
\end{eqnarray}
Compatibility of the operators $D_x+U_0$ and $D_t+V_0$ implies that 
$\gamma_t=\delta_x$.  So let $\gamma=\eta_x$ and $\delta=\eta_t$, where $\eta$ 
is a potential function, whereupon we can deduce that
\[\bn = \exp(\eta)\Psi_0(x, t, \mu)\bn_0,\]
where $\bn_0$ is a constant $N\times r$ matrix and $\Psi_0(x, t, \mu)$ is the 
fundamental solution of the linear differential equations defined by  
$\cL_0(\mu)\Psi_0=0$ and $\cM_0(\mu) \Psi_0=0$. According to Remark \ref{rem1}, 
the 
dressing matrix $\Phi$ is invariant under a rescaling of the matrix 
$\bn$, we can simply take
\begin{eqnarray}\label{bn}
\bn = \Psi_0(x, t, \mu)\bn_0 
\end{eqnarray}
In what follows, we explicitly construct kink solutions of arbitrary ranks.
\subsubsection{Rank $1$ kink solutions}
Here we consider the matrix ${A}=\mathbf{n}\mathbf{m}^{\tr}$, where $\mathbf{n}$ 
and $\mathbf{m}$ are vectors. As we discussed before, we
first solve for $\mathbf{m}$ using the equation (\ref{mn}), that is,
\begin{eqnarray}\label{mn1}
(C+\sum_{k=0}^{N-1}\frac{Q^{k} \bm \bn^{\tr} Q^{-k}}{\mu^{-1} \omega^k- 
\mu})\bn=0  .
\end{eqnarray}
We then determine the diagonal matrix $C$ and write down the rank $1$ solutions 
as follows:
\begin{Lem}\label{lemk}
Let matrix $A$ be a bi-vector and ${A}=\mathbf{n}\mathbf{m}^{\tr}$. If the dressing matrix given by (\ref{phik}) satisfies
(\ref{Inphi1}), then the entries for diagonal matrix $C$ are given by
\begin{eqnarray}\label{mac}
 c_i^2=\mu^2\frac{\tau_{i-1}}{\tau_i},
\quad \tau_i=\frac{1}{\mu^{2N}-1}\sum_{l=1}^N 
n_l^2   \mu^{2 \{(i-l) \mod N\}}
\end{eqnarray}
where $n_i$ are the components of the vector $\bn$.
\end{Lem}
\begin{proof}
Under the assumption, we have that
$\mathbf{n}^{\tr}Q^{-k}\mathbf{n}$ is a
scalar function. So the matrix 
\[W = \sum_{k=0}^{N-1}\frac{Q^{k} \bn^{\tr} Q^{-k} \bn}{\mu^{-1} \omega^k- \mu}
=\sum_{k=0}^{N-1}\frac{\bn^{\tr} Q^{-k} \bn Q^{k} }{\mu^{-1} \omega^k- \mu}\]
is diagonal with the entries on the diagonal being
$$
W_{ii}=\sum_{k=0}^{N-1}\sum_{l=1}^N n_l^2 \omega^{-l k} \frac{ \omega^{i k} }{\mu^{-1} \omega^k- \mu}
=\mu\sum_{l=1}^N n_l^2 \sum_{k=0}^{N-1} \frac{ \omega^{(i-l) k} }{ \omega^k- \mu^2}
=-\mu\sum_{l=1}^N n_l^2  \frac{N \mu^{2 \{(i-l-1) \mod N\}}}{\mu^{2N}-1},
$$
where we used the  (\ref{old}). So $W$ is invertible since $\mu\neq 0$ and $|\mu|\neq 1$.
Substituting it into (\ref{mn1}), we get the vector $\bm$ with components
\begin{eqnarray}\label{m}
m_i= \frac{\mu^{2N}-1}{\mu N}\frac{c_i n_i}{\sum_{l=1}^N n_l^2   \mu^{2 
\{(i-l-1) \mod N\}}} .
\end{eqnarray}
The matrix $C$ can be determined using the equation (\ref{AC1}), which is equivalent to
\[c_i - c_i^{-1} = \mu^{-1} N m_i n_i=\frac{\mu^{2N}-1}{\mu^2}\frac{c_i n_i^2}{\sum_{l=1}^N n_l^2   \mu^{2 \{(i-l-1) \mod N\}}}.\]
This leads to
\begin{eqnarray*}
 c_i^2=\frac{\mu^2 \sum_{l=1}^N n_l^2   \mu^{2 \{(i-l-1) \mod N\}}}{\mu^2 
\sum_{l=1}^N n_l^2   \mu^{2 \{(i-l-1) \mod N\}}-(\mu^{2N}-1) 
n_i^2}=\mu^2\frac{\tau_{i-1}}{\tau_i},
\end{eqnarray*}
where $\tau_i$ is defined by (\ref{mac}).
\end{proof}
We now use (\ref{sol}) to derive the real solution for $\phi^{(i)}$. 
For the case when $\mu=\mu^\star=\nu$ and $A=A^{\star}$, we only need to choose 
$\mathbf{n}_0$ to be a real valued constant vector.  
Using (\ref{bn}) and (\ref{psi0}), we can determine the real vector $\bn$.
This leads to $c_i^2>0$. 
According to (\ref{sol}), the solutions are
\[\phi^{(i)}=\ln \left(\frac{c_i}{c_{i+1}}\right)=\frac{1}{2}\ln \left(
\frac{\tau_{i-1}\tau_{i+1}}{\tau_i^2}\right)\]
Thus we have the following result:
\begin{Pro}\label{pro4}
Let $\mathbf{n}_0$ be a constant real vector and $\nu\in\bbbr,\ 
\nu\ne0,\ \nu\ne\pm 1$. 
 A rank $1$ kink solution of the  system (\ref{2+1}) on a trivial background
$\phi^{(i)}= 0, i=1,\cdots, N$, is given by
\begin{eqnarray}\label{1kink}
 \phi^{(i)}=\frac{1}{2}\ln \left(\frac{\tau_{i-1}\tau_{i+1}}{\tau_i^2}\right), 
\quad \tau_i=\frac{1}{\nu^{2N}-1}\sum_{l=1}^N 
n_l^2   \nu^{2 \{(i-l) \mod N\}}
\end{eqnarray}
where $n_i$ are the components of the vector $$\bn=\exp((\nu 
\Delta^{-1}-\nu^{-1} \Delta) x-(\nu^{-2} \Delta^2-\nu^2 \Delta^{-2}) t)\bn_0. 
$$ 
\end{Pro}

For the case when $\mu=\nu \exp(\frac{\pi \i}{N})$ with $\nu\in \bbbr$ and 
$A^\star=\omega^{-1} Q^{-1} A Q$, to get the real solutions we use the 
following statement.
\begin{Pro}\label{pro42}
Let $\mathbf{n}_0$ be a constant vector satisfying $\mathbf{n}_0=Q 
\mathbf{n}_0^\star$. For $\mu=\nu \exp(\frac{\pi \i}{N})$, where $\nu\in\bbbr$  
and $\nu\ne0$,
$\nu \neq \pm1$, a rank $1$ kink solution of the  system (\ref{2+1}) on a 
trivial 
background
$\phi^{(i)}= 0, i=1,\cdots, N$, is given by
\begin{eqnarray}\label{1kink2}
 \phi^{(i)}=\frac{1}{2}\ln \left(\frac{\tau_{i-1}\tau_{i+1}}{\tau_i^2}\right), 
\quad \tau_i=\frac{1}{\mu^{2N}-1}\sum_{l=1}^N 
n_l^2   \mu^{2 \{(i-l) \mod N\}}
\end{eqnarray}
where $n_i$ are the components of the vector $$\bn=\exp((\mu 
\Delta^{-1}-\mu^{-1} \Delta) x-(\mu^{-2} \Delta^2-\mu^2 \Delta^{-2}) t)\bn_0 
. $$ 
\end{Pro}
\begin{proof} To prove the statement, we only need to show that $c_i$ given by 
(\ref{mac}) are real. 
First we notice that 
$$
\mu=\omega \mu^\star, \qquad \Delta Q=\omega Q \Delta, \qquad \Delta^{-1} Q=\omega^{-1} Q \Delta^{-1} .
$$
Using these identities, we are able to show that
\begin{eqnarray*}
&&\left((\mu\Delta^{-1}-\mu^{-1} \Delta) x-(\mu^{-2} \Delta^2-\mu^2 
\Delta^{-2}) t\right) Q=Q\left( (\nu \omega^{-\frac{1}{2}}\Delta^{-1}-\nu^{-1}  
\omega^{\frac{1}{2}}\Delta)x -(\nu^{-2} \omega\Delta^2-\nu^2 
\omega^{-1} \Delta^{-2})t\right)\\
&&\qquad=Q \left((\mu^\star\Delta^{-1}-{\mu^\star}^{-1} \Delta) 
x-{(\mu^\star}^{-2} \Delta^2-{\mu^\star}^2 \Delta^{-2}) t\right).
\end{eqnarray*}
This leads to $\Psi_0(x,t,\mu)Q=Q \Psi_0(x,t,\mu^\star)$. Using (\ref{bn}) we 
have
$$\bn=\Psi_0(x,t,\mu)\bn_0=\Psi_0(x,t,\mu) Q\bn_0^\star
=Q \Phi(x,t,\mu^\star) \bn_0^\star=Q \bn^\star ,
$$
that is, $n_l=\omega^l n_l^\star$.  Substituting these into (\ref{m}) we get
\begin{eqnarray}\label{mm}
m_i^\star= \frac{\nu^{2N}-1}{\mu \omega^{-1} N}\frac{c_i 
\omega^{-i} n_i }{\omega^{-2i+2}\sum_{l=1}^N {n_l}^2   {\mu}^{2\{(i-l-1) 
\mod N\}}} 
=\omega^{i-1} m_i,
\end{eqnarray}
which implies $\bm=\omega Q^{-1} \bm^\star$. Thus $A=\bn \bm^{\tr}=\omega Q 
A^{\star} Q^{-1}$. Finally we show that $c_i$ are real. Indeed,
\begin{eqnarray*}
\tau_i^{\star}=\frac{1}{{\mu^\star}^{2N}-1}\sum_{l=1}^N 
{n_l^\star}^2   {\mu^\star}^{2 \{(i-l) \mod N\}} =\omega^{-2i} \tau_i
\end{eqnarray*}
Therefore, we have ${c_i^\star}^2={\mu^\star}^2 
\frac{\tau^{\star}_{i-1}}{\tau_i^{\star}}=\mu^2 
\frac{\tau_{i-1}}{\tau_i}=c_i^2$. Using (\ref{sol}), we derive the real 
solution for $\phi^{(i)}$ as in the statement. 
\end{proof}
Note that this Proposition is valid for arbitrary dimension $N$. However, it 
only leads to new solutions different from the ones obtained in Proposition 
\ref{pro4} when $N$ is even.

\subsubsection{Rank $r>1$ kink solutions}
In this case, the rank $r$ matrix $A=\bn \bm^{\tr}$, where $\bn$ and $\bm$ are 
$N\times r$ matrices of rank $r$. As discussed before, we first solve $\bm$ in terms of 
$\bn$ using (\ref{mn}). It follows from (\ref{AC1}) that
\begin{equation}\label{ckr}
C=C^{-1}+\mu^{-1} \sum_{k=0}^{N-1} Q^k \bm \bn^{\tr} Q^{-k}.
\end{equation}
Substituting  it into (\ref{mn}), we get
\begin{eqnarray}\label{mnrk}
 C^{-1}\bn +\frac{1}{\mu^2} \sum_{k=0}^{N-1}\frac{\omega^k Q^k\bm \bn^{\tr} 
Q^{-k} \bn}{\omega^k \mu^{-1}-\mu}=0 .
\end{eqnarray}
Let $\tilde{\bm}=C \bm$. Then (\ref{ckr}) and (\ref{mnrk}) become
\begin{eqnarray}\label{tilde}
 C^2=I_N+\mu^{-1}N (\tilde{\bm} \bn^{\tr})_{\diag};\qquad
 \frac{1}{\mu^2} \sum_{k=0}^{N-1}\frac{\omega^k Q^k\tilde{\bm} \bn^{\tr} 
Q^{-k} \bn}{\mu-\omega^k \mu^{-1}}=\bn .
\end{eqnarray}
We denote 
$j$-th rows of $\bn$ and $\tilde{\bm}$ as $\bn_{j}$ and $\tilde{\bm}_{j}$  
respectively. It follows from the second identity in (\ref{tilde}) that
\begin{eqnarray}\label{mj}
 &&\bn_j=\tilde{\bm}_j \frac{1}{\mu^2} \sum_{k=0}^{N-1}\frac{\omega^{(j+1)k}
\bn^{\tr} Q^{-k} \bn}{\mu-\omega^k \mu^{-1}}=\tilde{\bm}_j \frac{1}{\mu^2} 
\bn^{\tr} \sum_{k=0}^{N-1}\frac{\omega^{(j+1)k}
 Q^{-k} }{\mu-\omega^k \mu^{-1}}\bn=\frac{N}{\mu (\mu^{2N}-1)}\tilde{\bm}_j 
\bn^{\tr} S(j) \bn,
\end{eqnarray}
where $S(j)$ is an $N \times N$ diagonal matrix with $i$-th diagonal entry 
equal to $\mu^{2\{(j-i)\!\!\mod N \}}$. We can determine $\bm$ and 
further the matrix $C$ in the dressing matrix as follows:
\begin{Lem}\label{lemkr}
Let rank $r$ matrix $A=\bn \bm^{\tr}$, where $\bn$ and $\bm$ are 
$N\times r$ matrices of rank $r$. If the dressing matrix given by (\ref{phik}) satisfies
(\ref{Inphi1}), then the entries for diagonal matrix $C$ are given by
\begin{eqnarray}\label{macr}
 c_j^2=\mu^{2r}\frac{\tau_{j-1}}{\tau_j},
\quad \quad \tau_j=\det R(j), 
\quad R(j)=\frac{1}{\mu^{2N}-1}\bn^{\tr} S(j) 
\bn,
\end{eqnarray}
where $S(j)$ is an $N \times N$ diagonal matrix with $i$-th diagonal entry 
equal to $\mu^{2\{(j-i)\!\!\mod N \}}$.
\end{Lem}
\begin{proof}
Using the notations given in the statement, it follows from (\ref{mj}) that
$$
\tilde \bm_j=\frac{\mu}{N} \bn_j R(j)^{-1}.
$$
We substitute it into the first equation in (\ref{tilde}) and determine that the 
$j$-th diagonal entry of the diagonal matrix $C$ satisfies
\begin{eqnarray}\label{cjr}
c_j^2=1+ \sum_{\alpha,\beta=1}^r \bn_{j\alpha} 
\{{R(j)^{-1}}\}_{\alpha,\beta} \bn_{j \beta} .
\end{eqnarray}
The explicit formula for the entries at $(\alpha, \beta)$ of $R(j)$ is equal to
\begin{eqnarray*}
 R(j)_{\alpha,\beta}=\frac{1}{\mu^{2N}-1}\sum_{k=1}^N n_{k \alpha} n_{k\beta} \mu^{2\{(j-k)\!\!\mod 
N \}} .
\end{eqnarray*}
It is easy to show the identity between the entries between matrices $R(j)$ and 
$R(j-1)$:
\begin{eqnarray*}
  R(j)_{\alpha,\beta}= \mu^2 R(j-1)_{\alpha,\beta}- n_{j \alpha} 
n_{j \beta} .
\end{eqnarray*}
This implies 
$$
R(j)=\mu^2 R(j-1)- \bn_j^{\tr}  \bn_j .
$$
Using Sylvester's determinant theorem, it leads to
$$
\mu^{2 r} \det R(j-1)=\det R(j) (1+\bn_j R(j)^{-1} \bn_j^{\tr} ).
$$
Comparing it with (\ref{cjr}), we obtain that
$$
c_j^2=\frac{\mu^{2r} \det R(j-1)}{\det R(j)}=\frac{\mu^{2r} \tau_{j-1}}{\tau_j},
$$
which is (\ref{macr}) in the statement.
\end{proof}
Using (\ref{sol}) we get the solutions in the statement, where $\bn$ is 
determined by (\ref{bn}) and (\ref{psi0}). For the case when $\mu=\mu^\star=\nu$ 
and $A=A^{\star}$, we only need to choose 
$\mathbf{n}_0$ to be a real valued constant matrix of size $N\times r$ to 
guarantee that the solutions are real. We state the result as follows: 
\begin{Pro}\label{pro4r}
Let $\mathbf{n}_0$ be a rank $r$ constant real matrix of size $N\times r$ and
$\mu=\mu^\star=\nu\in\bbbr,\ 
\nu\ne0,\ \nu\ne\pm 1$. 
 A rank $r$ kink solution of the  system (\ref{2+1}) on a trivial background
$\phi^{(i)}= 0, i=1,\cdots, N$, is given by
\begin{eqnarray}\label{rkink}
 \phi^{(j)}=\frac{1}{2}\ln \left(
\frac{\tau_{j-1}\tau_{j+1}}{\tau_j^2}\right), \quad \tau_j=\det R(j), 
\quad R(j)=\frac{1}{\nu^{2N}-1}\bn^{\tr} S(j) 
\bn,
\end{eqnarray}
where $S(j)$ is an $N \times N$ diagonal matrix with $i$-th diagonal entry 
equal to $\nu^{2\{(j-i)\!\!\mod N \}}$
and $$\bn=\exp((\nu \Delta^{-1}-\nu^{-1} \Delta) x-(\nu^{-2} \Delta^2-\nu^2 
\Delta^{-2}) t)\bn_0 .$$ 
\end{Pro}
Notice that taking $r=1$ in Proposition \ref{pro4r} we get the results in 
Proposition \ref{pro4}.

For the case when $\mu=\nu \exp(\frac{\pi \i}{N})$ with $\nu>0$ and 
$A^\star=\omega^{-1} Q^{-1} A Q$, we have the similar result as Proposition 
\ref{pro42} to get kink solutions of rank $r$.
\begin{Pro}\label{pro4r2}
Let $\mathbf{n}_0$ be a rank $r$ constant real matrix of size $N\times r$  
satisfying $\mathbf{n}_0=Q 
\mathbf{n}_0^\star$. For $\mu=\nu \exp(\frac{\pi \i}{N})$, where $\nu>0$  and 
$\nu \neq 1$, a rank $r$ kink solution of the  system (\ref{2+1}) on a trivial 
background
$\phi^{(i)}= 0, i=1,\cdots, N$, is given by
\begin{eqnarray}\label{rkink2}
 \phi^{(j)}=\frac{1}{2}\ln \left(
\frac{\tau_{j-1}\tau_{j+1}}{\tau_j^2}\right), \quad \tau_j=\det R(j), 
\quad R(j)=\frac{1}{\mu^{2N}-1}\bn^{\tr} S(j) 
\bn,
\end{eqnarray}
where $S(j)$ is an $N \times N$ diagonal matrix with $i$-th diagonal entry 
equal to $\mu^{2\{(j-i)\!\!\mod N \}}$
and $$\bn=\exp((\mu \Delta^{-1}-\mu^{-1} \Delta) x-(\mu^{-2} \Delta^2-\mu^2 
\Delta^{-2}) t)\bn_0 .$$  
\end{Pro}
\begin{proof}
Similar to the proof of Proposition \ref{pro42}, under the assumption we 
have
$
\bn=Q \bn^\star .
$
This leads to 
\begin{eqnarray*}
\tau_j^\star=\det R(j)^{\star}=\det 
\left( \frac{1}{{\mu^\star}^{2N}-1}\bn^{\star\tr} S(j)^\star
\bn^\star\right)=\det\left(\frac{\omega^{-2j}}{{\mu}^{2N}-1}
\bn^{\tr} S(j)\bn\right)
=\omega^{-2rj} \tau_j
\end{eqnarray*}
Therefore, we have ${c_j^\star}^2={\mu^\star}^{2 r}
\frac{\tau^{\star}_{j-1}}{\tau_j^{\star}}=\mu^{2 r} 
\frac{\tau_{j-1}}{\tau_j}=c_j^2$. Using (\ref{sol}), we derive the real 
solution for $\phi^{(i)}$ as in the statement.
\end{proof}
Similar to Proposition \ref{pro42}, although this Proposition is valid for 
arbitrary dimension $N$, it 
only leads to new solutions different from the ones obtained in Proposition 
\ref{pro4r} when $N$ is even.

\subsection{Breather solutions}
A breather solution corresponds to the simple poles at points of a 
generic orbit of the reduction group. The corresponding dressing matrix is of 
the form
%\begin{eqnarray}\label{bphi}
%\Phi^{-1}(\lambda)=C+\sum_{k=0}^{N-1}\left(\frac{Q^{-k} A 
%Q^{k}}{\lambda \omega^k- \mu}+\frac{Q^{-k} A^\star
%Q^{k}}{\lambda \omega^k- \mu^\star}\right),
%\Phi(\lambda)=C+\sum_{k=0}^{N-1}\left(\frac{Q^{k} A^{\tr} 
%Q^{-k}}{\lambda^{-1} \omega^k- \mu}+\frac{Q^{k} A^{\star\tr}
%Q^{-k}}{\lambda^{-1} \omega^k- \mu^\star}\right),\quad \omega^k\mu\neq 
%\omega^l \mu^{\star}, k,l=0,\cdots,N-1 ,
%\end{eqnarray}
where $A$ is a $\lambda$-independent matrix of size $N\times N$ and the matrix 
$C$ defined by (\ref{matc}) is diagonal with
 real functions  $c_i, i=1,\cdots, N$ on the diagonal. Moreover, it follows from 
Proposition \ref{pro0} that
\begin{eqnarray}
&&\lim_{\lambda \to\infty} \Phi(\lambda)=C-\frac{1}{\mu}\sum_{k=0}^{N-1}Q^{k} 
A^{\tr} Q^{-k}-\frac{1}{\mu^\star}\sum_{k=0}^{N-1}Q^{k} 
A^{\star\tr} Q^{-k}\nonumber\\
&&\qquad =C-\frac{1}{\mu} N A_{\diag}-\frac{1}{\mu^\star} N 
A^{\star}_{\diag}=C^{-1};\label{ACb1}\\
&& \Phi(\mu)A=(C+\sum_{k=0}^{N-1}\frac{Q^{k} A^{\tr} Q^{-k}}{\mu^{-1} \omega^k- 
\mu}+\sum_{k=0}^{N-1}\frac{Q^{k} A^{\star\tr}
Q^{-k}}{\mu^{-1} \omega^k- \mu^\star})A=0. \label{ACb2}
\end{eqnarray}
If matrix $A$ is invertible then, in a similar manner to in Proposition \ref{pro3} for the kink 
case, we find the real solutions for $\phi^{(i)}=0$ on the trivial background. 
Hence we assume that the rank of $A$ is $ r\leq N-1$. 

In the same way as we did for the case of kinks in Section \ref{Kink}, we 
present $
A=\bn \bm^{\tr},
$
where $ \bn $ and $\bm$ are both $N\times r$ matrices of rank $r$. It is 
obvious that the dressing matrix $\Phi(\lambda)$ (\ref{bphi}) is again
parametrised by a matrix $\bn$ lying on a Grassmannian and we 
also have
\begin{eqnarray}\label{bnb}
 \bn=\Psi_0(x,t,\mu) \bn_0,
\end{eqnarray}
where $\bn_0$ is an $N\times r$ matrix of rank $r$.
\subsubsection{Rank $1$ breather solutions}
In a similar manner to in the case of the rank $1$ kink, we consider the matrix 
${A}=\mathbf{n}\mathbf{m}^{\tr}$, where $\mathbf{n}$ and $\mathbf{m}$ are 
vectors. Then (\ref{ACb2}) becomes
\begin{eqnarray}\label{mnb1}
(C+\sum_{k=0}^{N-1}\frac{Q^{k} \bm \bn^{\tr} Q^{-k}}{\mu^{-1} \omega^k- 
\mu}+\sum_{k=0}^{N-1}\frac{Q^{k} \bm^{\star}\bn^{\star\tr}
Q^{-k}}{\mu^{-1} \omega^k- \mu^\star})\bn=0  .
\end{eqnarray}
We first use it to determine $\mathbf{m}$ in terms of $\bn$. Then we use 
(\ref{sol}) to compute the solutions $\phi^{(i)}$ for equation (\ref{2+1}).
\begin{Pro}\label{pro5}
Let $\mathbf{n}_0$ be a constant vector and $
\mu\in\bbbc,\ |\mu|\ne 1,\ \mu\ne \omega^k  \mu^\star,\ k=1,\ldots,N$. 
 A rank $1$ breather solution of the  system (\ref{2+1}) on a trivial background
$\phi^{(i)}= 0, i=1,\cdots, N$, is given by
\begin{eqnarray}
&& \phi^{(i)}=\frac{1}{2}\ln \left(
\frac{\tau_{i-1}\tau_{i+1}}{\tau_i^2}\right),\qquad 
\tau_i=\rho(i)^2-|\sigma(i)|^2, 
\label{1breather}
\end{eqnarray}
where
\begin{eqnarray}
&&\sigma(i)=\frac{1}{\mu^{2N}-1}\sum_{l=1}^N n_l^2   \mu^{2 \{(i-l) \mod N\}}; \qquad
\rho(i)=\frac{1}{|\mu|^{2N}-1} \sum_{l=1}^N |n_l|^2   |\mu|^{2 \{(i-l) \mod N\}}\label{srho}
\end{eqnarray}
and the $n_i$ are the components of vector $$\bn=\exp((\mu 
\Delta^{-1}-\mu^{-1} \Delta) x-(\mu^{-2} \Delta^2-\mu^2 \Delta^{-2}) 
t)\bn_0 .$$ 
\end{Pro}
\begin{proof}
Under the assumption, we have that
$\mathbf{n}^{\tr}Q^{-k}\mathbf{n}$ and $\bn^{\star\tr} Q^{-k}\mathbf{n}$ are
scalar functions. So we define 
\begin{eqnarray*}
&&D=\sum_{k=0}^{N-1}\frac{\bn^{\tr} Q^{-k} \bn Q^{k} }{\mu^{-1} \omega^k- 
\mu};\qquad
E=\sum_{k=0}^{N-1}\frac{\bn^{\star\tr} Q^{-k} \bn Q^{k} }{\mu^{-1} \omega^k- 
\mu^{\star}}.
\end{eqnarray*}
They are diagonal with the entries on diagonal being
\begin{eqnarray*}
&&D_{ii}=\sum_{k=0}^{N-1}\sum_{l=1}^N n_l^2 \omega^{-l k} \frac{\omega^{i k} 
}{\mu^{-1} \omega^k- \mu}
=\mu\sum_{l=1}^N n_l^2 \sum_{k=0}^{N-1} \frac{ \omega^{(i-l) k} }{ \omega^k- 
\mu^2}
=-\mu\sum_{l=1}^N n_l^2  \frac{N \mu^{2 \{(i-l-1) \mod N\}}}{\mu^{2N}-1};\\
&&E_{ii}=\sum_{k=0}^{N-1}\sum_{l=1}^N |n_l|^2 \omega^{-l k} \frac{\omega^{i k} 
}{\mu^{-1} \omega^k- \mu^{\star}}
=\mu\sum_{l=1}^N |n_l|^2 \sum_{k=0}^{N-1} \frac{ \omega^{(i-l) k} }{ \omega^k- 
|\mu|^2}
=-\mu\sum_{l=1}^N |n_l|^2  \frac{N |\mu|^{2 \{(i-l-1) \mod N\}}}{|\mu|^{2N}-1},
\end{eqnarray*}
where we used (\ref{old}). Using the notations defined by 
(\ref{srho}), we rewrite them as
\begin{eqnarray}\label{DE}
 D_{ii}=-\mu N \sigma(i-1); \qquad 
E_{ii}=-\mu N \rho(i-1).
\end{eqnarray}

Writing out the entries of (\ref{mnb1}), we have
$$
c_i n_i+D_{ii} m_i+E_{ii} m_i^{\star}=0.
$$
We solve it for $m_i$ and it follows that
$$
m_i=c_i \frac{D_{ii}^{\star} n_i-E_{ii} n_i^{\star}}{|E_{ii}|^2-|D_{ii}|^2}.
$$
The matrix $C$ can be determined using the equation (\ref{ACb1}), which is 
equivalent to
\[c_i - c_i^{-1} = N \left(\frac{1}{\mu}  m_i n_i + \frac{1}{\mu^{\star}} 
m_i^{\star} n_i^{\star} \right)
=\frac{N c_i}{|E_{ii}|^2-|D_{ii}|^2} \left(\frac{1}{\mu}  D_{ii}^{\star} 
n_i^2  + \frac{1}{\mu^{\star}} D_{ii} n_i^{\star 2}
-(\frac{1}{\mu} E_{ii}+\frac{1}{\mu^{\star}} E_{ii}^{\star} )|n_i|^2  
\right).\]
This leads to
\begin{eqnarray*}
 c_i^2=\frac{|E_{ii}|^2-|D_{ii}|^2}{|E_{ii}|^2-|D_{ii}|^2-\frac{N}{\mu}  
D_{ii}^{\star} 
n_i^2  - \frac{N}{\mu^{\star}} D_{ii} n_i^{\star 2}
+N (\frac{1}{\mu} E_{ii}+\frac{1}{\mu^{\star}} E_{ii}^{\star} )|n_i|^2},
\end{eqnarray*}
Substituting (\ref{DE}) into it, we get
\begin{eqnarray*}
c_i^2=|\mu|^4\frac{ \rho(i-1)^2-|\sigma(i-1)|^2}{ \rho(i)^2-|\sigma(i)|^2}
\end{eqnarray*}
where $\sigma(i)$ and $\rho(i)$ are defined in (\ref{srho}). Here we used the 
identities between $\sigma(i-1)$ and $\sigma(i)$, and $\rho(i-1)$ and $\rho(i)$ 
as follows:
\begin{eqnarray*}
 \mu^2 \sigma(i-1)-n_i^2=\sigma(i); \qquad  |\mu|^2 
\rho(i-1)-|n_i|^2=\rho(i) .
\end{eqnarray*}
It follows from (\ref{sol}) that
\[\phi^{(i)}=\ln \left(\frac{c_i}{c_{i+1}}\right)=\frac{1}{2}\ln \left(\frac{
\tau_{i-1}\tau_{i+1}}{\tau_i^2}\right),\]
where $\tau_i$ is defined by (\ref{1breather}). The vector $\bn$ is determined
using (\ref{bnb}) and (\ref{psi0}). 
\end{proof}
\subsubsection{Rank $r>1$ breather solutions}
In this case, the rank $r$ matrix $A=\bn \bm^{\tr}$, where $\bn$ and $\bm$ are 
$N\times r$ matrices.  It follows from (\ref{ACb1}) that
\begin{equation}\label{cbr}
C=C^{-1}+\frac{1}{\mu}\sum_{k=0}^{N-1}Q^{k} 
\bm \bn^{\tr} Q^{-k}+\frac{1}{\mu^\star}\sum_{k=0}^{N-1}Q^{k} 
\bm^{\star} \bn^{\star\tr} Q^{-k}.
\end{equation}
Substituting  it into (\ref{ACb2}), we get
\begin{eqnarray}\label{mnrb}
 C^{-1}\bn +\frac{1}{\mu^2} \sum_{k=0}^{N-1}\frac{\omega^k Q^k\bm \bn^{\tr} 
Q^{-k} \bn}{\omega^k \mu^{-1}-\mu}+\frac{1}{\mu^{\star 2}} 
\sum_{k=0}^{N-1}\frac{\omega^k Q^k\bm^{\star} \bn^{\star \tr} 
Q^{-k} \bn}{\omega^k \mu^{\star -1}-\mu}=0 .
\end{eqnarray}
Let $\tilde{\bm}=C \bm$. Then (\ref{cbr}) and (\ref{mnrb}) become
\begin{eqnarray}
&& C^2=I_N+\frac{1}{\mu} N (\tilde{\bm} 
\bn^{\tr})_{\diag}+\frac{1}{\mu^{\star}} N (\tilde{\bm}^{\star} 
\bn^{\star \tr})_{\diag};\label{tilde1}\\
&& \frac{1}{\mu^2} \sum_{k=0}^{N-1}\frac{\omega^k Q^k\tilde{\bm} \bn^{\tr} 
Q^{-k} \bn}{\mu-\omega^k \mu^{-1}}+\frac{1}{\mu^{\star 2}} 
\sum_{k=0}^{N-1}\frac{\omega^k Q^k\tilde{\bm}^{\star} \bn^{\star \tr} 
Q^{-k} \bn}{\mu-\omega^k \mu^{\star -1}}
=\bn .\label{tilde2}
\end{eqnarray}
We denote the
$j$-th rows of $\bn$ and $\tilde{\bm}$ by $\bn_{j}$ and $\tilde{\bm}_{j}$  
respectively. It follows from (\ref{tilde2}) that
\begin{eqnarray}
 &&\bn_j=\tilde{\bm}_j \frac{1}{\mu^2} \sum_{k=0}^{N-1}\frac{\omega^{(j+1)k}
\bn^{\tr} Q^{-k} \bn}{\mu-\omega^k \mu^{-1}}+\tilde{\bm}^{\star}_j 
\frac{1}{\mu^{\star 2}} 
\sum_{k=0}^{N-1}\frac{\omega^{(j+1)k} \bn^{\star \tr} 
Q^{-k} \bn}{\mu-\omega^k \mu^{\star -1}}\nonumber\\
&&\quad =\tilde{\bm}_j \frac{1}{\mu^2} 
\bn^{\tr} \sum_{k=0}^{N-1}\frac{\omega^{(j+1)k} Q^{-k} }{\mu-\omega^k 
\mu^{-1}}\bn+\tilde{\bm}^{\star}_j 
\frac{1}{\mu^{\star 2}} \bn^{\star \tr}
\sum_{k=0}^{N-1}\frac{\omega^{(j+1)k}  
Q^{-k} }{\mu-\omega^k \mu^{\star -1}} \bn. \label{mjb}
\end{eqnarray}
Let us introduce the following notations for $r\times r$ 
matrices with entry at $(\alpha,\beta)$ being
\begin{eqnarray}
 && R(j)_{\alpha, \beta}=\frac{\mu}{N}\left(\frac{1}{\mu^2}\bn^{\tr} 
\sum_{k=0}^{N-1}\frac{\omega^{(j+1)k} Q^{-k} 
}{\mu-\omega^k\mu^{-1}} \bn\right)_{\alpha,\beta}= \frac{1}{\mu^{2N}-1} 
\sum_{l=1}^N n_{l \alpha} n_{l\beta} \mu^{2\{(j-l)\!\!\mod N\}};\label{rjb} \\
&& P(j)_{\alpha, \beta}=\frac{\mu^\star}{N}\left(\frac{1}{\mu^{\star 2}} 
\bn^{\star \tr}
\sum_{k=0}^{N-1}\frac{\omega^{(j+1)k} Q^{-k} }{\mu-\omega^k \mu^{\star -1}} 
\bn \right)_{\alpha, \beta}= \frac{1}{|\mu|^{2N}-1} 
\sum_{l=1}^N n_{l \alpha}^{\star} n_{l\beta} |\mu|^{2\{(j-l)\!\!\mod N\}} 
.\label{pjb}
\end{eqnarray}
Notice that $R(j)=R(j)^{\tr}$ and $P(j)^{\dagger}=P(j)^{\star\tr}=P(j)$, where the
notation $\dagger$ denotes the conjugate transpose of a matrix.
It is easy to show the identity between the entries between $R(j)$ and 
$R(j-1)$, and between $P(j)$ and $P(j-1)$:
\begin{eqnarray*}
  R(j)_{\alpha,\beta}= \mu^2 R(j-1)_{\alpha,\beta}- n_{j \alpha} 
n_{j \beta}; \quad P(j)_{\alpha,\beta}= |\mu|^2 
P(j-1)_{\alpha,\beta}- n_{j \alpha}^\star 
n_{j \beta} .
\end{eqnarray*}
These imply that 
\begin{eqnarray}\label{re01}
R(j)=\mu^2 R(j-1)- \bn_j^{\tr}  \bn_j; \quad P(j)=|\mu|^2 P(j-1)- 
\bn_j^{\dagger}  \bn_j  .
\end{eqnarray}
It follows from (\ref{mjb}) that
$$
\bn_j=\frac{N}{\mu}\tilde\bm_j R(j)+\frac{N}{\mu^\star} \tilde \bm_j^{\star} 
P(j)=\hat \bm_j R(j)+ \hat \bm_j^{\star} P(j),
$$
where $\hat \bm=\frac{N}{\mu} \tilde \bm$. From it, we obtain the solution for 
$\hat \bm$:
and this leads to
$$
\hat\bm_j=\left(\bn_j P(j)^{-1}- \bn_j^{\star} R(j)^{ \star-1}\right) 
\left(R(j) P(j)^{-1}-P(j)^{\star} R(j)^{\star -1}\right)^{-1}.
$$
We substitute it into (\ref{tilde1}) and determine that the 
$j$-th diagonal entry in the diagonal matrix $C$ satisfies
\begin{eqnarray}
&&c_j^2=1+\sum_{\alpha=1}^r \left(\hat \bm_{j\alpha}  \bn_{j 
\alpha} 
+\hat \bm_{j\alpha}^\star  \bn_{j \alpha}^\star \right)=1+\hat \bm 
\bn^{\tr}+\hat\bm^{\star} \bn^{\dagger}\nonumber\\
&&\quad=1+\bn_j \left(R(j)-P(j)^{\star} R(j)^{\star -1} P(j)\right)^{-1} 
\bn^{\tr} +\bn_j^{\star} \left(P(j)^{\star}-R(j) p(j)^{-1} 
R(j)^{\star}\right)^{-1} \bn^{\tr}\nonumber\\
&&\qquad\quad + \bn_j^{\star} \left(R(j)^{\star}-P(j) 
R(j)^{-1} P(j)^{\star}\right)^{-1} 
\bn^{\dagger} +\bn_j \left(P(j)-R(j)^{\star} p(j)^{\star -1} 
R(j)\right)^{-1} \bn^{\dagger}.\label{cjb}
\end{eqnarray}
\begin{Lem}\label{detcj} Matrices $R(j)$ and $P(j)$ defined by (\ref{rjb}) 
and (\ref{pjb}) respectively and scalar $c_j$ given by (\ref{cjb})
satisfy the identity
\begin{eqnarray}\label{detin}
&&\tau_{j-1}= |\mu|^{-4r} \tau_j  c_j^2, \qquad  \tau_j=\det \left( R(j) 
R(j)^{\star}-R(j) P(j) R(j)^{-1} P(j)^{\star} \right) .
\end{eqnarray}
\end{Lem}
\begin{proof} 
We first apply Sylvester's determinant theorem to (\ref{re01}). It leads to
\begin{eqnarray}\label{r01}
\mu^{2 r} \det R(j-1)=\det R(j) (1+\bn_j R(j)^{-1} \bn_j^{\tr} ).
\end{eqnarray}
Using the Sherman-Morrison formula, we find that
$$
R(j-1)^{-1}=\mu^2 R(j)^{-1} \left(1-\frac{1}{1+\alpha} \bn_j^{\tr} \bn_j 
R(j)^{-1}\right), \qquad \alpha=\bn_j R(j)^{-1} \bn_j^{\tr}.
$$
Using $\alpha$, we rewrite (\ref{r01}) as
\begin{eqnarray}\label{r01a}
 \det R(j-1)=\frac{1}{\mu^{2 r}}\det R(j) (1+\alpha ).
\end{eqnarray}
Using (\ref{re01}) we now compute 
\begin{eqnarray*}
&&\quad R(j-1)^{\star}- P(j-1) R(j-1)^{-1} P(j-1)^{\star}\\
&&=\frac{1}{\mu^{\star 2}}
\left( R(j)^{\star}+\bn_j^{\dagger} \bn_j^{\star}-(P(j)+\bn_j^{\dagger} \bn_j)
R(j)^{-1} \left(1-\frac{1}{1+\alpha} \bn_j^{\tr} \bn_j R(j)^{-1}\right) 
(P(j)^{\star}+\bn_j^{\tr} \bn_j^{\star})\right)\\
&&=\frac{1}{\mu^{\star 2}} \left(R(j)^{\star}- P(j) R(j)^{-1} P(j)^{\star}+\frac{1}{1+\alpha} 
(\bn_j^{\dagger}-P(j) R(j)^{-1} \bn_j^{\tr}) (\bn_j^{\star}-\bn_j R(j)^{-1} P(j)^{\star})\right).
\end{eqnarray*}
Let $W(j)=R(j)^{\star}- P(j) R(j)^{-1} P(j)^{\star}$. Using Sylvester's determinant theorem, we obtain
\begin{eqnarray*}
&&\quad\det W(j-1)
=\frac{1}{\mu^{\star 2r}}\det W(j)
\left(1+\frac{1}{1+\alpha}(\bn_j^{\star}-\bn_j R(j)^{-1} P(j)^{\star}) W(j)^{-1} 
(\bn_j^{\dagger}-P(j) R(j)^{-1} \bn_j^{\tr})\right).
\end{eqnarray*}
Combining it with (\ref{r01a}) and using the notation in (\ref{detin}), we have
\begin{eqnarray*}
&&\qquad \tau_{j-1}=\det R(j-1) \det W(j-1)\\
&&=\frac{1}{|\mu|^{4r}} \det R(j) \det W(j)
 \left(1+\alpha+(\bn_j^{\star}-\bn_j R(j)^{-1} P(j)^{\star}) W(j)^{-1} 
(\bn_j^{\dagger}-P(j) R(j)^{-1} \bn_j^{\tr})\right).
\end{eqnarray*}
We compare the scalar expression inside the bracket with $c_j^2$ given by (\ref{cjb}) and use the identity
$$
R(j)^{-1}+R(j)^{-1} P(j)^{\star} W(j)^{-1}P(j) R(j)^{-1}=\left(R(j)-P(j)^{\star} R(j)^{\star -1} P(j)\right)^{-1} 
$$
and we get the formula (\ref{detin}) in the statement.
\end{proof}
We are now able to write down the rank $r$ breather 
solutions as follows:
\begin{Pro}\label{pro5b}
Let $\mathbf{n}_0$ be a rank $r$ constant matrix of size $N\times r$ and $
\mu\in\bbbc,\ |\mu|\ne 1,\ \mu\ne \omega^k  \mu^\star,\ k=1,\ldots,N$. 
A rank $r$ breather solution of the  system (\ref{2+1}) on a trivial background
($\phi^{(i)}= 0, i=1,\cdots, N$) is given by
\begin{eqnarray}\label{rbreather}
 \phi^{(j)}=\frac{1}{2}\ln \left(
\frac{\tau_{j-1}\tau_{j+1}}{\tau_j^2}\right), \qquad \tau_j=\det \left( R(j) 
R(j)^{\star}-R(j) P(j) R(j)^{-1} P(j)^{\star} \right) 
\end{eqnarray}
where $r\times r$ matrices $R(j)$ and $P(j)$ are defined by (\ref{rjb}) and (\ref{pjb}) respectively
and $$\bn=\exp((\mu \Delta^{-1}-\mu^{-1} \Delta) x-(\mu^{-2} \Delta^2-\mu^2 
\Delta^{-2}) t)\bn_0 .$$ 
\end{Pro}
\begin{proof} It follows from Lemma \ref{detcj} that
$$
c_j^2=|\mu|^{4r} \frac{\tau_{j-1}}{\tau_j} .
$$
Using (\ref{sol}) we get the solutions in the statement, where $\bn$ is 
determined by (\ref{bnb}) and (\ref{psi0}).
\end{proof}

Notice that taking $r=1$ in Proposition \ref{pro5b} we get the results in 
Proposition \ref{pro5}.

\subsection{The $\tau$-function and continuous limits}\label{sec33}
In the last two sections, we showed that both kink and breather solutions are 
expressed in the form
\begin{eqnarray}\label{phitau}
 \phi^{(i)}=\ln \frac{c_i}{c_{i+1}}=\frac{1}{2} \ln 
\frac{\tau_{i-1}\tau_{i+1}}{\tau_i^2}
\end{eqnarray}
according to Propositions \ref{pro4}--\ref{pro5b}. In this section, we derive 
the equations for $c_i$ and $\tau_i$. To simplify the notations, we 
drop the index $i$ and introduce the shift operator $\cS$ mapping the index $i$ 
to $i+1$, that is, $c_i=c$, $c_{i+1}=\cS c$, $c_{i-1}=\cS^{-1} c$ and the same 
for $\tau_i$.  The shift operator satisfies $\cS^N \tau =\tau$ and $\cS^N c=c$.

Let $c=e^{u}$ and $\tau=e^{2 v}$. It follows from (\ref{phitau}) that
\begin{eqnarray}\label{uvi}
\phi= (1-\cS) u=  \cS^{-1} (\cS-1)^2 v.
\end{eqnarray}
This leads to
\begin{eqnarray}\label{thuvi}
\theta=-(\cS-1)^{-1} (\cS+1) \phi_x=(\cS+1) u_x.
\end{eqnarray}
Substituting (\ref{uvi}) and (\ref{thuvi}) into (\ref{2+1}), we get
$$
(1-\cS) u_t=(\cS+1) u_{xx}+\left((\cS+1)u_x \right)
\left((1-\cS)u_x \right)+(\cS^2-1) e^{2 (\cS^{-1}-1)u} .
$$
Thus we have
\begin{eqnarray}\label{equu}
u_t=(\cS+1)(1-\cS)^{-1} u_{xx}+u_x^2 -(\cS+1)\left( e^{2 
(\cS^{-1}-1)u} -1\right),
\end{eqnarray}
where we choose the integration constant to be $1$ such that $u=0$ is a
solution of (\ref{equu}). Since $u=\ln c$, we have  the equation for function 
$c$ as follows:
\begin{eqnarray}\label{equc}
\frac{c_t}{c}=(\cS+1)(1-\cS)^{-1} 
\left(\frac{c_{xx}}{c}-\frac{c_x^2}{c^2}\right)+\frac{c_x^2}{c^2} 
-(\cS+1)\left( e^{2(\cS^{-1}-1)\ln c} -1\right).
\end{eqnarray}

We now derive the equation for $v$. From (\ref{uvi}) we have 
$u=(\cS^{-1}-1) v$. Note that
$$
u_x^2=(1+\cS^{-1}) \left(v_x^2-v_x \cS v_x\right)+(1-\cS^{-1}) v_x \cS v_x
$$
Substituting these into (\ref{equu}), we obtain
\begin{eqnarray}\label{equv}
 v_t=(\cS+1)(1-\cS)^{-1} \left(v_{xx}+v_x^2-v_x \cS v_x 
-e^{2(\cS-2+\cS^{-1})v}+1\right) -v_x \cS v_x .
\end{eqnarray}
Now the $\tau$-function is related to $v$ by $v=\frac{1}{2} \ln \tau$. 
It follows from (\ref{equv}) that
\begin{eqnarray}\label{tauf}
 \frac{\tau_t}{\tau}=(\cS+1)(1-\cS)^{-1} 
\left(\frac{\tau_{xx}}{\tau}-\frac{1}{2} \left(
\frac{\tau_x^2}{\tau^2}+\frac{\tau_x \cS \tau_x}{\tau
\cS \tau}\right)-2 \frac{(\cS \tau) (\cS^{-1}\tau)}{\tau^2}+2
\right)-\frac{1}{2} \frac{\tau_x \cS \tau_x}{\tau \cS \tau} .
\end{eqnarray}
Thus we have proved the following statement:
\begin{Pro}\label{pro6}
If function $c$ satisfies (\ref{equc}), then $\phi=\ln \frac{c}{\cS c}$ 
satisfies the 2-dimensional Volterra equation (\ref{2+1}).  If function $\tau$ 
satisfies (\ref{tauf}), then $\phi=\ln \frac{(\cS \tau) 
(\cS^{-1}\tau)}{\tau^2}$ 
satisfies the 2-dimensional Volterra equation (\ref{2+1}).
\end{Pro}

We know that the continuous limit of system (\ref{2+1}) goes to the KP 
equation. Indeed, for the continuous limit as $N\to \infty$ and $h=N^{-1}$, we 
set
\begin{eqnarray}\label{lmr}
 T=h^3 t, \quad X=i h+4ht,\quad  Y=h^2 x, 
\end{eqnarray}
which imply that
\begin{eqnarray}\label{lmrd}
 \frac{\partial }{\partial t}=h^3 \frac{\partial }{\partial T}+4 h 
\frac{\partial }{\partial X}, \quad \frac{\partial }{\partial x}
=h^2\frac{\partial }{\partial Y} .
\end{eqnarray}
Let $\phi^{(i)}(x,t)=h^2 w(X, Y, T)$. In the new variables system (\ref{2+1}) takes the form 
\[ w_T=\frac{2}{3}w_{XXX}+8ww_X-2D_X^{-1}w_{YY}+O(h^2),\] 
and it  goes to the KP 
equation in the limit $h\to 0$. We can compute the continuous limits of equations (\ref{equu}) and (\ref{equv}) by 
setting
\begin{eqnarray}\label{cfun}
u(x,t)= h \hu(X,Y,T),\quad v(x,t)=\hv(X,Y,T) %, \quad 
%\tau(x,t)=\hta(X,Y,T)
\end{eqnarray}
respectively. Notice that 
$$
\cS u=\hu(X+h,Y,T)=\sum_{n=0}^{\infty}\frac{h^n}{n!} 
\frac{\partial^n \hu}{\partial X^n}=e^{h \frac{\partial}{\partial X}} \hu .
$$
Hence we replace the shift operator $\cS$ by $e^{h \frac{\partial}{\partial 
X}}$. This leads to
\begin{eqnarray}\label{cons}
{h \frac{\partial}{\partial X}} 
(1+\cS)(1-\cS)^{-1}=-2-\frac{h^2}{6}\frac{\partial^2}{\partial X^2}+O(h^4) . 
\end{eqnarray}
Substituting (\ref{lmrd}), (\ref{cfun}) and (\ref{cons})  into (\ref{equu}), we 
obtain
$$
h^5 \hu_{XT}+4h^3 \hu_{XX}=\left(-2-\frac{h^2}{6}\frac{\partial^2}{\partial 
X^2}\right)h^5 \hu_{YY}+4h^3 \hu_{XX}+\frac{2}{3}h^5 \hu_{XXXX}-8 h^5 \hu_X 
\hu_{XX}+O(h^6)
$$
implying
$$
\hu_{XT}+2 \hu_{YY}-\frac{2}{3} \hu_{XXXX}+8  \hu_X \hu_{XX}=O(h) .
$$
In the same way, we substitute (\ref{lmrd}), (\ref{cfun}) and (\ref{cons})  
into (\ref{equv}) and obtain its continuous limit
$$
\hv_{XT}+2 \hv_{YY}-\frac{2}{3} \hv_{XXXX}-4 \hv_{XX}^2=O(h) .
$$
In the variable $\tit=e^{\hv}$ and in the limit $h\to 0$, it becomes
$$
3(\tit \tit_{XT}-\tit_X \tit_T)+6 (\tit \tit_{YY}-\tit_Y^2)-2 \tit \tit_{4X}+8 
\tit_X \tit_{3X}-6 \tit_{XX}^2=0,
$$
which is the standard bilinear form for the KP equation. This gives 
us the link between the  Hirota $\tau$-function for the bilinear form and  the 
functions $\tau_i$ defined in Propositions 
\ref{pro4}-\ref{pro5b} in the continuous limit $\tau_i=\tit^2$.

\section{Classification of rank $1$  solutions}\label{sec4}

In this section we describe and analyse kink and breather solutions given in 
Propositions \ref{pro4}, \ref{pro42} and Proposition 
\ref{pro5}. Solutions are completely characterised by 
the choice of the poles of the dressing matrix $\Phi(\lambda)$ and a constant 
vector $\bn_0$. In the case of kink solutions the invariant dressing matrix has 
$N$ poles while in 
the case of breather solutions the dressing matrix has $2N$ poles. 

It is convenient to use the basis
\begin{equation}\label{bek}
 \be_k=(\omega^k,\omega^{2k},\ldots,\omega^{(N-1)k},1)^{\rm tr},\qquad 
k=1,\ldots ,N
\end{equation}
 of eigenvectors $\Delta\be_k=\omega^k\be_k$ of the matrix 
$\Delta$
for representation of the vector 
\begin{eqnarray}\label{nalpha}
\bn_0=\sum_{k=1}^N \alpha_k \be_k.
\end{eqnarray}
In this basis we have
\[
 \Psi_0 (x, t, \mu)\be_k=\exp 
\left((\mu\omega^{-k}-\mu^{-1}\omega^{k})x+(\mu^2\omega^{-2k}-\mu^{-2}\omega^{2k
} )t\right)\be_k
\]
and thus 
\[
 \bn= \Psi_0 (x, t, \mu)\bn_0=\sum_{k=1}^N \alpha_k \exp 
\left((\mu\omega^{-k}-\mu^{-1}\omega^{k})x+(\mu^2\omega^{-2k}-\mu^{-2}\omega^{2k
} )t\right)\be_k .
\]
Obviously $\alpha_k=N^{-1}\be^{\tr}_{-k}\, \bn_0$ in (\ref{nalpha}). The 
vector $\bn_0$ 
in this basis is 
given by a matrix $\alpha=(\alpha_1,\ldots ,\alpha_N)$.

\subsection{Classification of rank $1$ kink solutions}
In this section, we classify the kink solutions of rank $1$ given by 
Propositions \ref{pro4} and \ref{pro42}.
We begin with the description of possible kink solutions in the cases $N=3,4$ 
and then give an overview of the general case. We draw attention to the fact 
that the properties of solutions for even and odd values of $N$ are slightly 
different. In particular, in the case of even $N$ there is an obvious solution 
\begin{equation}\label{solf}
 \phi^{(j)}=(-1)^j f(x),\qquad \theta_j=0,\quad j=1,\ldots,2N
\end{equation}
of the system (\ref{2+1}), where $f(x)$ is an arbitrary differentiable function 
of $x$. Moreover, Proposition \ref{pro42} gives new solutions only when $N$ is 
even.

In the case of kink solutions obtained in Proposition \ref{pro4} when 
$\mu=\mu^*=\nu\in \bbbr$ and $\nu\notin\{\pm1,0\}$, the vector $\bn_0$ is real 
and thus we require that
\[
\alpha_N=\alpha_N^\star,\qquad  \alpha_{N-k}=\alpha_k^\star,\quad k=1,\ldots 
,N-1.
\]
In the case of kink solutions obtained in Proposition \ref{pro42} when 
$\mu=\nu\exp(\frac{\pi\i}{N}), \nu\in \bbbr$ and $\nu\notin\{\pm1,0\}$, the 
vector $\bn_0=Q \bn^\star$. Notice that $Q \bn^\star=\alpha_1^\star \be_N+\alpha_2^\star \be_{N-1}+\cdots +\alpha_N^\star\be_1$
thus we require that
\[  \alpha_{k}=\alpha_{N-k+1}^\star,\quad k=1,\ldots ,N.
\]
In particular, when $N=2m$, it reduces to
\begin{eqnarray}
\alpha_{k}=\alpha_{2m-k+1}^\star,\quad k=1,\ldots ,m. \label{even}
\end{eqnarray}

\subsubsection{Classification of rank $1$ kink solutions in the case where $N=3$}
In this section we set $N=3$ for equation (\ref{2+1}), that is, equation 
(\ref{eqn3}). In this case it is sufficient to study rank $1$ solutions due to 
the fact that ${\rm Gr}(1,3)\simeq {\rm Gr}(2,3)$.

We will classify possible solutions in terms of the constant matrix 
$ \alpha=( \alpha_1,\alpha_2,\alpha_3)$, which represents the real vector 
$\bn_0$.
In variables
\begin{eqnarray*}
 \xi=(\nu -\nu^{-1}) x-(\nu^{-2} -\nu^2) t; \quad \eta=\frac{\sqrt{3}}{2} 
((\nu +\nu^{-1}) x-(\nu^{-2} +\nu^2) t) 
\end{eqnarray*}
 we have
$$\Psi_0 (x, 
t,\nu)\bv_1=e^{-\frac{\xi}{2}-\i \eta} \bv_1 ;
\quad \Psi_0 (x, t, \nu)\bv_2=e^{-\frac{\xi}{2}+\i \eta} \bv_2; \quad \Psi_0(x, 
t, \nu)\bv_3=e^{\xi} \bv_3.
$$
To get real solutions, it is required that $\bn$ and $\nu$ be real. Hence 
there are 
three cases: 
\begin{enumerate}
 \item $\alpha_3\neq 0\in \bbbr, \alpha_1=\alpha_2=0$. So we have
$$
\bn=\alpha_3 e^{\xi} \bv_3=\alpha_3 e^{\xi} (1, 1, 1)^{\tr} .
$$
This leads to $\tau_i=\alpha_3^2 e^{2\xi} (1+\nu^2+\nu^4)/(\nu^6-1)$ in 
Proposition \ref{pro4}.
It follows from (\ref{1kink}) that the solution is trivial, i.e.
$\phi^{(i)}=0$.
When we consider the 
classification of solutions later on,  we won't count this case any more.
\item $\alpha_3=0, \alpha_2= \alpha_1^\star \neq 0$. Without the loss of 
generality, we 
take $\alpha_1=e^{\i \beta },$ where $\beta\in \bbbr$ is constant. Then we  
take 
$\alpha_2=e^{-\i\beta}$ such that $\bn$ is real. Indeed,
\begin{eqnarray*}
 \bn=2 e^{-\frac{\xi}{2}} (\cos(\eta-\beta-\frac{2\pi}{3}) , 
\cos(\eta-\beta+\frac{2\pi}{3}),  \cos(\eta-\beta))^{\tr}.
\end{eqnarray*}
Using (\ref{1kink}), we obtain the solution 
\begin{eqnarray*}
&& \phi^{(j)}=\frac{1}{2}\ln\left(\frac{ 
\tau_{j-1}\tau_{j+1}}{\tau_j^2}\right), \\
&& \tau_j=\cos^2(\eta-\beta-\frac{2\pi}{3})  \nu^{2 \{(j-1)\!\! \mod 3\}} + 
\cos^2(\eta-\beta+\frac{2\pi}{3})  \nu^{2 \{(j-2)\!\! \mod 3\}} + 
\cos^2(\eta-\beta)  \nu^{2 \{j\!\! \mod 3\}}
\end{eqnarray*}
In this case, solutions are periodic functions of the variable $\eta$. 
\item $ \alpha_1 \alpha_2 \alpha_3\neq 0$. Let $\alpha_1= e^{\i 
\beta+\gamma}$,  where $\beta, \gamma$ are constants. We take 
$\alpha_3=e^\delta \in 
\bbbr$ and $\alpha_2= e^{-\i \beta+\gamma}$ in order $\bn$ to be real.
\begin{eqnarray*}
 \bn=e^{-\frac{1}{2} \xi+\gamma}\left(\begin{array}{c} 2  \cos( 
\eta-\beta-\frac{2\pi}{3})+ e^{\frac{3\xi}{2}+\delta-\beta}
\\
2 \cos( \eta-\beta+\frac{2 \pi}{3})+ e^{\frac{3\xi}{2}+\delta-\beta}
\\
2 \cos ( \eta-\beta)+ e^{\frac{3\xi}{2}+\delta-\beta}
\end{array}\right). 
\end{eqnarray*}
Using (\ref{1kink}), we are able to write down the solutions for $\phi^{(i)}$. 
Here we omit the tedious formula and only show the density plot. 
Notice that when $\xi\rightarrow +\infty$, solutions $\phi^{(i)}\rightarrow 0$;
when $\xi\rightarrow -\infty$, the contribution of $\alpha_1$ and $\alpha_2$ is 
dominant, which leads to periodic solutions. A line on the $(x, t)$-plane
given by $ e^{\frac{3\xi}{2}+\delta-\beta}=2$ corresponds to 
the wave front propagation. It has a slope equal to $-\nu/(1+\nu^2)$.

We now choose $\nu=0.4$ and
$\alpha_1=\alpha_2=\alpha_3=1$.  In Figure \ref{kink3} on the left we show a
density 
plot of $\phi^{(1)}$ in
the $(x,t)$-plane and on the right a snapshot of the solution $\phi^{(1)}$ at 
$t=0$.  Notice that
the solution is a periodic oscillating wave, oscillating in half of space (the
$x$-axis) only and moving to the left as time progresses.  Furthermore,
the frontier of the wave does not have a stationary profile and oscillates
in a rather complicated way.
\begin{figure}[ht]
\centering
\includegraphics[scale=0.5]{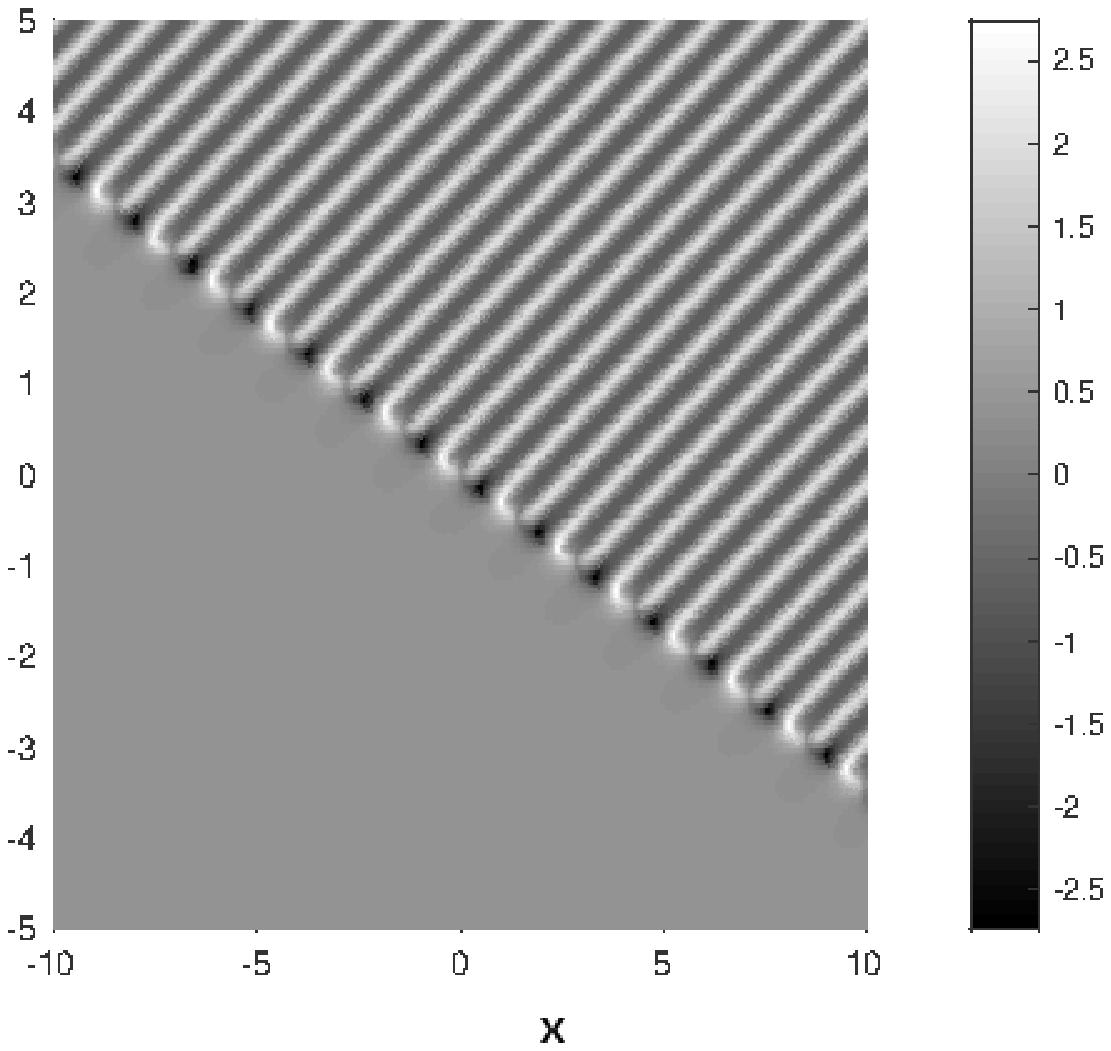}
\includegraphics[scale=0.467]{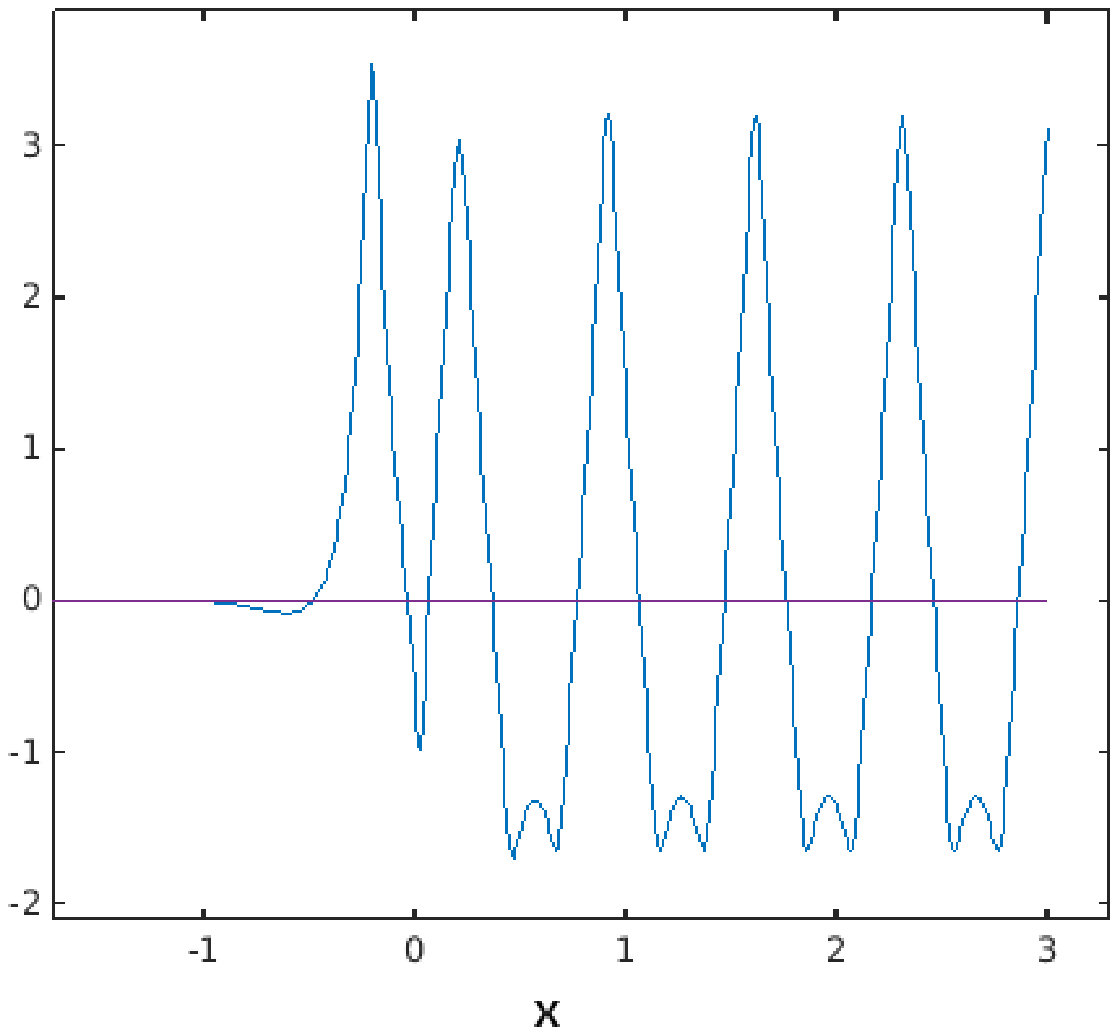}
\caption{Density  plot of $\phi^{(1)}(x,t)$  and a snapshot of $\phi^{(1)}$ at 
t=0 ($\alpha=(1,1,1)$, $\nu=0.4$) .
}
\label{kink3}
\end{figure}
\end{enumerate}
Therefore, in the case $N=3$ we have only two types of kink 
solutions. To the best of our knowledge the wave front solutions (see Figure
\ref{kink3}) represent a new class of exact solutions for integrable models.

\subsubsection{Classification of rank $1$ kink solutions in the case where 
$N=4$}\label{N4k}
For $N=4$  equation (\ref{2+1}) can be rewritten in the form 
\begin{eqnarray}\label{eq4}
\begin{array}{l} 2 \phi^{(i)}_{t}=\phi^{(i+1)}_{xx}-\phi^{(i+3)}_{xx} + 
\phi^{(i)}_{x}(\phi^{(i+1)}_{x}-\phi^{(i+3)}_{x})
+2 e^{2 \phi^{(i+1)}} -2 e^{2 \phi^{(i-1)}},
\end{array} 
\end{eqnarray}
where $\phi^{(i+4)}=\phi^{(i)}$ and $\sum_{i=1}^4\phi^{(i)}=0 $.
We classify its all possible rank $1$ kink solutions.

We first consider the case when $\mu=\mu^\star=\nu$ and the constant 
real vector 
$$
\bn_0=\alpha_1 \bv_1 +\alpha_2 \bv_2 +\alpha_3 \bv_3 +\alpha_4 \bv_4,
$$
where $\alpha=(\alpha_1,\alpha_2,\alpha_3,\alpha_4),\ 
\alpha_3=\alpha_1^\star\in \bbbc,\ \alpha_2,\alpha_4\in\bbbr$.
In variables
$\xi=(\nu -\nu^{-1}) x$, $\zeta=(\nu+\nu^{-1})x$ and 
$\eta=(\nu^{-2} -\nu^2) t$  we have
$$
\Psi_0 (x, t, \nu)\bv_1=e^{-\zeta\i+\eta}\bv_1; \quad \Psi_0 (x, t, 
\nu)\bv_2=e^{-\xi-\eta}\bv_2;
\quad \Psi_0 (x, t, \nu)\bv_3=e^{\zeta\i+\eta}\bv_3;\quad \Psi_0 (x, t, 
\nu)\bv_4=e^{\xi-\eta}\bv_4.
$$

There are four cases (excluding the trivial solutions) :
\begin{enumerate}
 \item $\alpha_2$ and $\alpha_4$ are both non-zero real numbers, and 
$\alpha_1=\alpha_3=0$. We can take $\alpha_4=1$,
then
 \begin{eqnarray*}
&& \bn= (e^{\xi-\eta}-\alpha_2 e^{-\xi-\eta}, e^{\xi-\eta}+\alpha_2 
e^{-\xi-\eta}, e^{\xi-\eta}-\alpha_2 e^{-\xi-\eta}, e^{\xi-\eta}+\alpha_2 
e^{-\xi-\eta} )^{\tr}\\
&&\quad=e^{\xi-\eta}(1-\alpha_2 e^{-2\xi}, 1+\alpha_2 e^{-2\xi}, 1-\alpha_2 
e^{-2\xi}, 1+\alpha_2 e^{-2\xi} )^{\tr}.
\end{eqnarray*}
This leads to 
\begin{eqnarray*}
&& \tau_j=e^{2\xi-2\eta}\left((1-\alpha_2 e^{-2\xi})^2  (\nu^{2 \{(j-1)\!\! 
\mod 4\}} +\nu^{2 \{(j-3)\!\! \mod 4\}})\right. \\
&&\qquad \left.+(1+\alpha_2 e^{-2\xi})^2  (\nu^{2 \{(j-2)\!\! \mod 4\}} +\nu^{2 
\{j\!\! \mod 4\}})\right)/(\nu^8-1)\\
&&\qquad=e^{2\xi-2\eta} \left(\frac{1+\alpha_2^2 e^{-4\xi}}{\nu^2-1}- (-1)^j   
\frac{2 \alpha_2 e^{-2 \xi}}{\nu^2+1}\right).
\end{eqnarray*}
Using (\ref{1kink}), we obtain the solution 
\begin{eqnarray*}
&& \phi^{(j)}=\frac{1}{2}\ln\left(\frac{ 
\tau_{j-1}\tau_{j+1}}{\tau_j^2}\right)=\ln\left|\frac{(1+\alpha_2^2 
e^{-4\xi})(\nu^2+1)
+2 \alpha_2 e^{-2 \xi} (-1)^{j}   (\nu^2-1)}{(1+\alpha_2^2 
e^{-4\xi})(\nu^2+1)-2 
\alpha_2 e^{-2 \xi} (-1)^j   (\nu^2-1)}
\right|,
\end{eqnarray*}
which is independent of time $t$ (see left plot in Figure \ref{kink4el}). 
\begin{figure}[ht]
\centering
\includegraphics[scale=0.5]{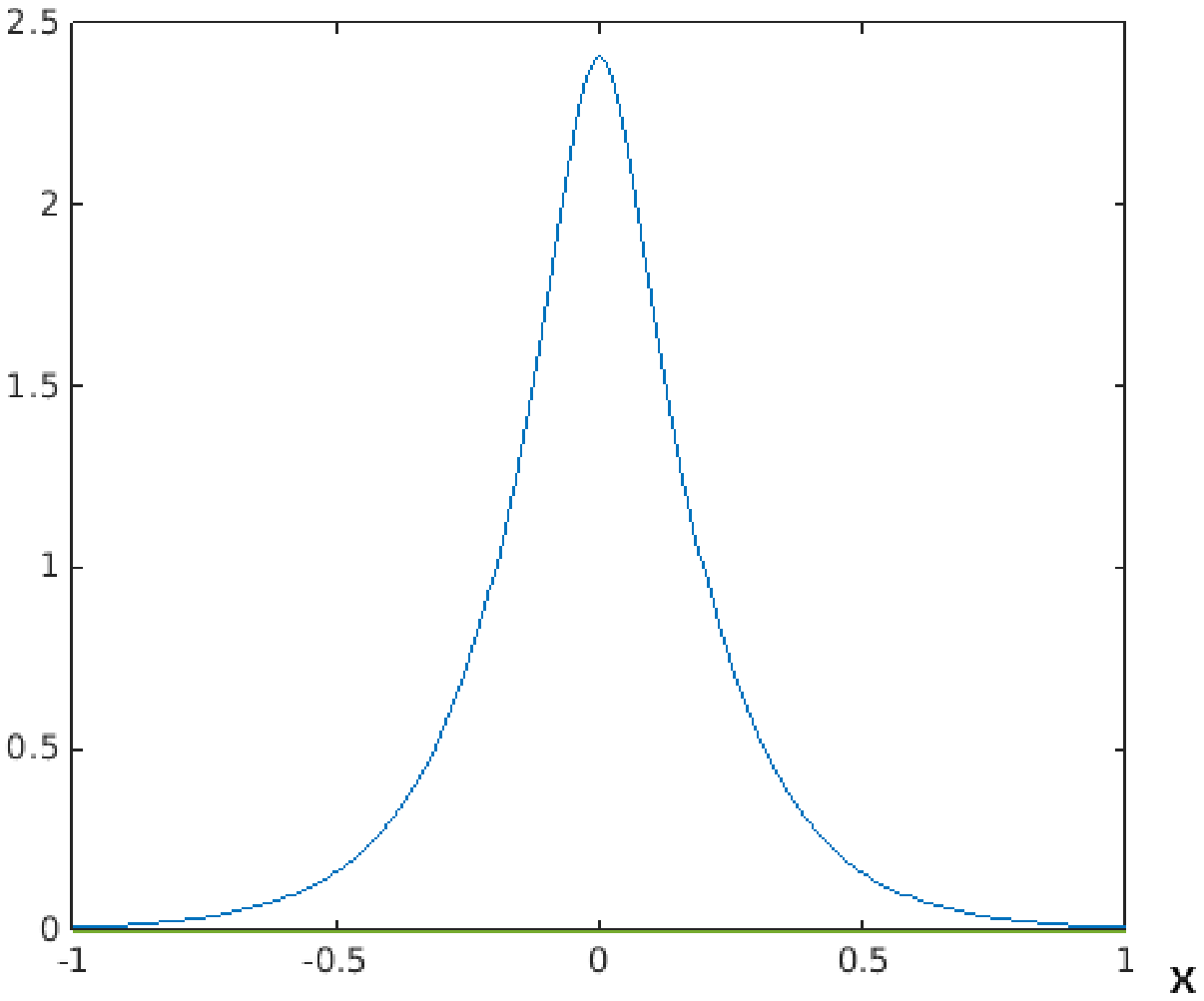}
\includegraphics[scale=0.5]{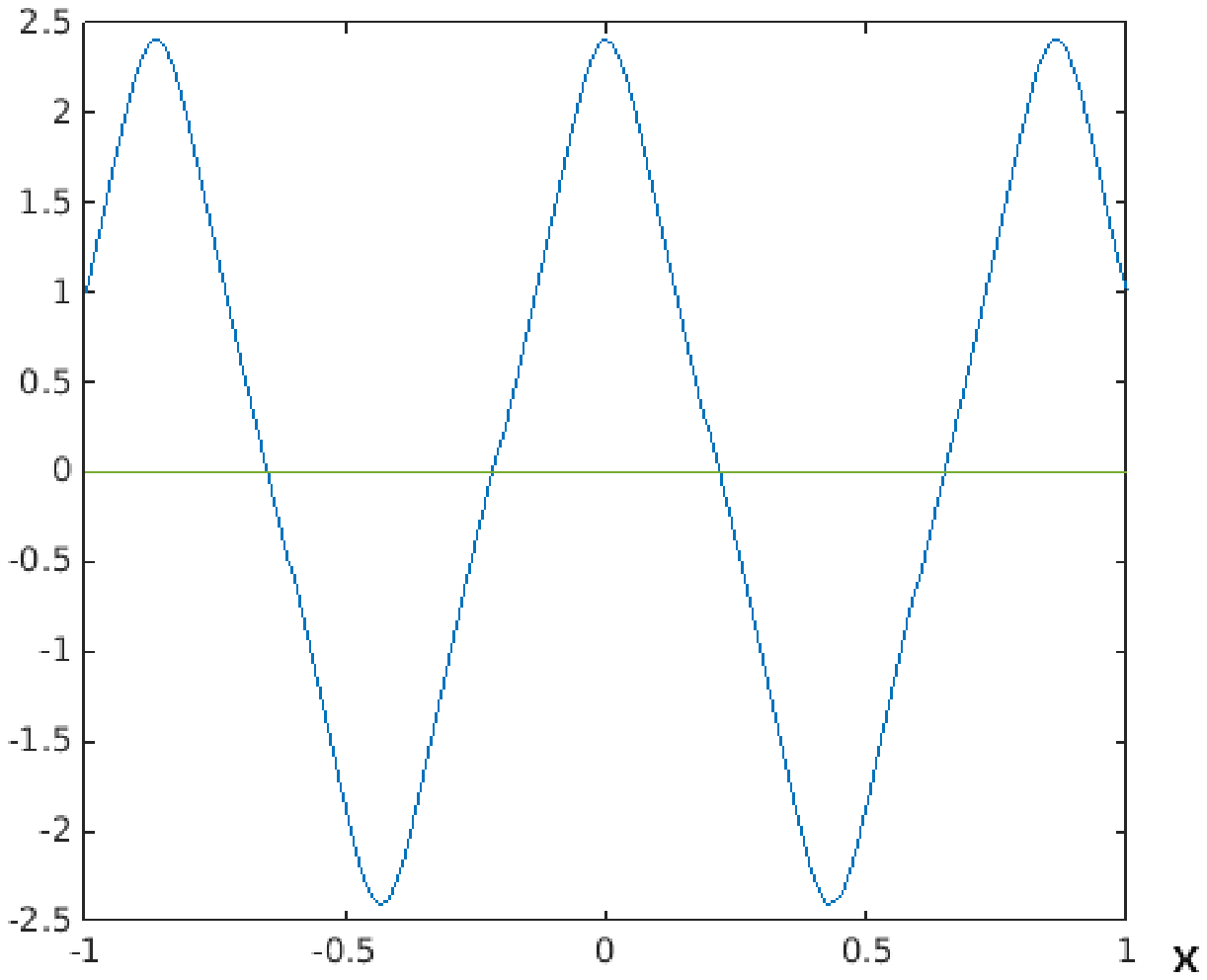}
\caption{Graph of $\phi^{(1)}(x,t)$  for kink solutions with $\nu=0.3$:
 on the left $\alpha=(0,1,0,1)$ and on the right  
$\alpha=(1,0,1,0)$
}\label{kink4el}\end{figure}
This solution 
is of type (\ref{solf}).
\item $\alpha_2=\alpha_4=0, \alpha_1 \alpha_3\neq 0$. Without the loss of 
generality, we 
take $\alpha_1=e^{\i \beta }=\alpha_3^\star,$ where $\beta\in \bbbr$ is a real
constant. Then
\begin{eqnarray*}
 \bn=2 e^{\eta} ( \sin(\zeta-\beta), -\cos(\zeta-\beta), -\sin(\zeta-\beta), 
\cos(\zeta-\beta)
)^{\tr}.
\end{eqnarray*}
Using (\ref{1kink}), we obtain the solution 
\begin{eqnarray*}
&& \phi^{(j)}=\frac{1}{2}\ln\left(\frac{ 
\tau_{j-1}\tau_{j+1}}{\tau_j^2}\right), \\
&& \tau_j=\sin^2(\zeta-\beta)  (\nu^{2 \{(j-1)\!\! \mod 4\}} -\nu^{2 
\{(j-3)\!\! \mod 4\}})
+\cos^2(\zeta-\beta)  (\nu^{2 \{j\!\! \mod 4\}} -\nu^{2 \{(j-2)\!\! \mod 4\}})
\end{eqnarray*}
Here we get periodic solutions (see right plot in Figure \ref{kink4el}). 
This is also a
solution of type (\ref{solf}).
\item Only one of $\alpha_2$ and $\alpha_4$ is nonzero, and $\alpha_1=e^{\i 
\beta }, 
\alpha_3=e^{-\i\beta}$,  where $\beta\in \bbbr$ is constant. Using Proposition 
\ref{pro4}, we are able to write down the solutions for 
$\phi^{(j)}$. Here we ignore the tedious formula and only show their plots, see 
 first two density
plots in Figure \ref{kinkcl}. 
\item $\alpha_1 \alpha_2 \alpha_3\alpha_4\neq 0$. Let $\alpha_1= e^{\i 
\beta},\alpha_3= e^{-\i \beta}$ ,  where $\beta\in \bbbr$ is constant. We take 
$\alpha_2, \alpha_4\in 
\bbbr$. The right plot in Figure \ref{kinkcl} is its density plot.
\begin{figure}[ht]
\centering
\includegraphics[scale=0.35]{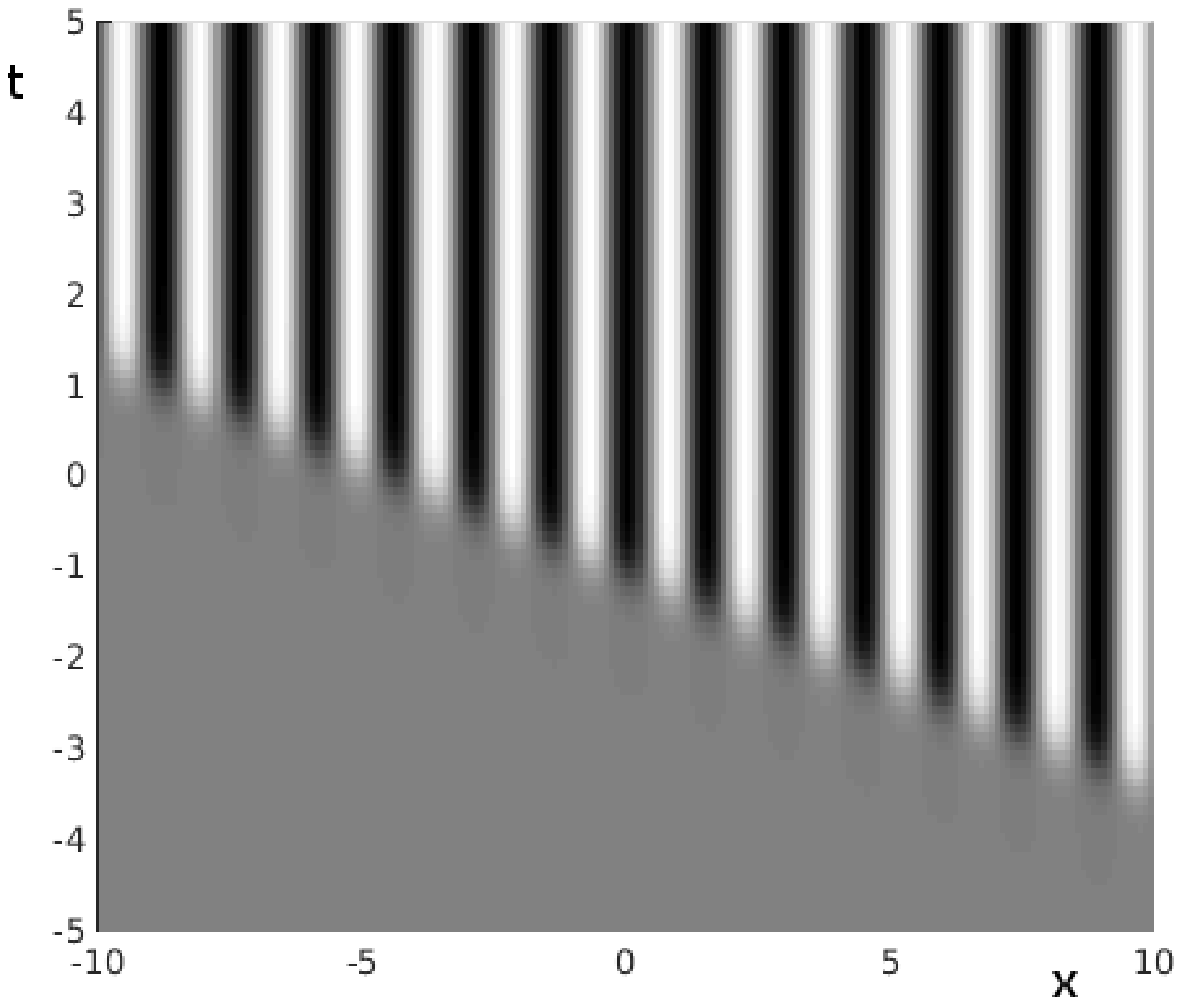}
\includegraphics[scale=0.35]{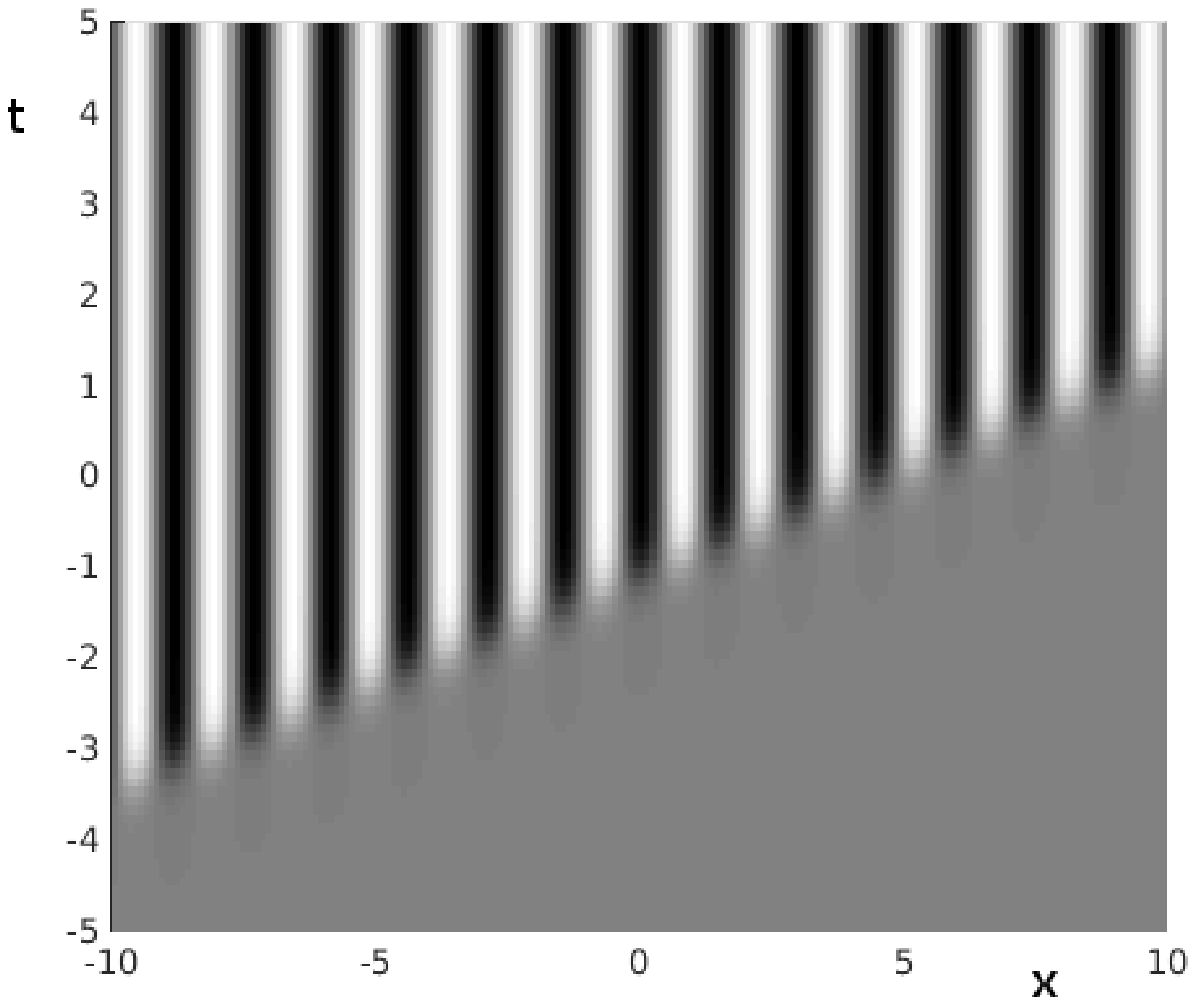} 
\includegraphics[scale=0.35]{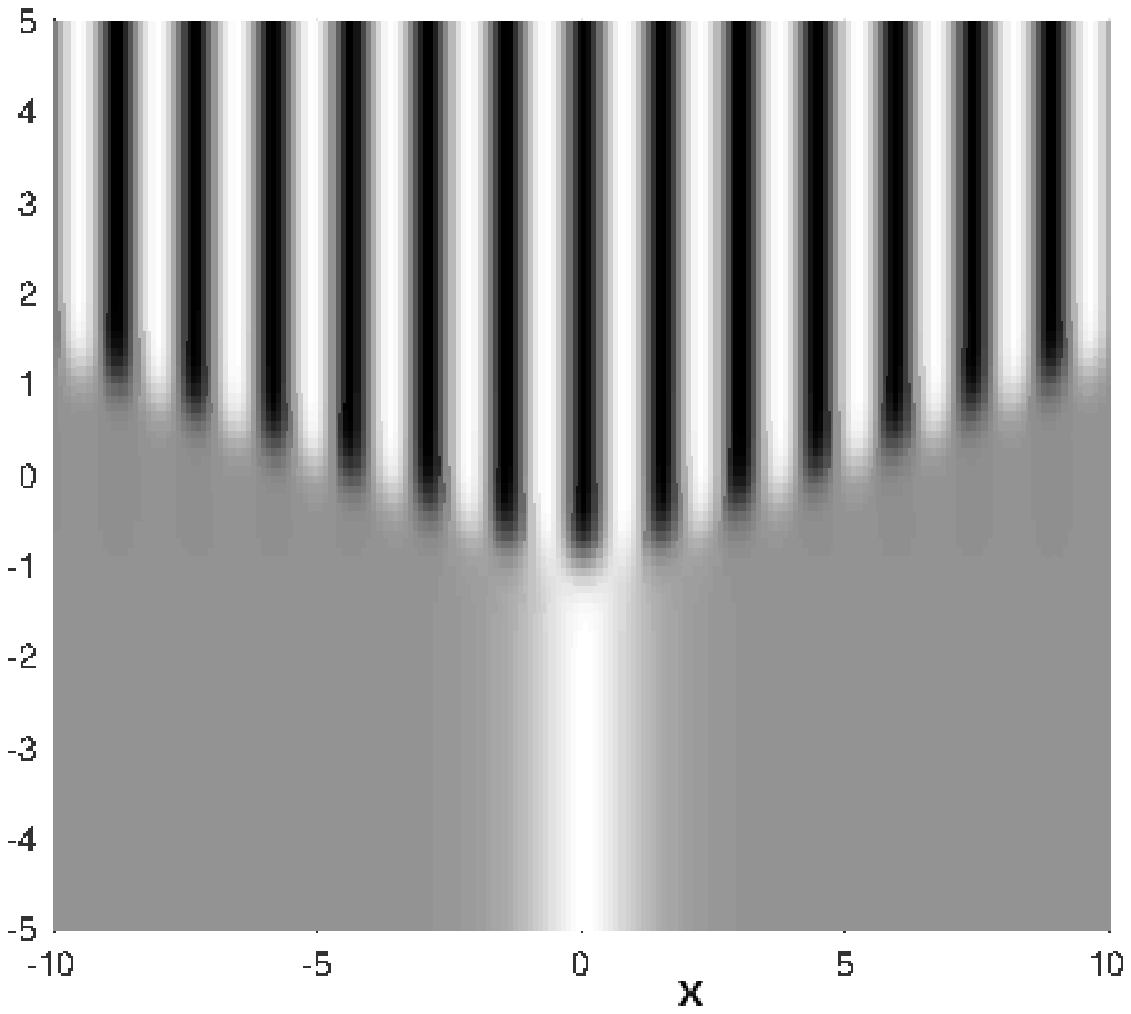}
\caption{Density  plot of $\phi^{(1)}(x,t)$  for kink solutions with $\nu=0.7$ 
and  $\alpha$ equals to $(1,1,1,0)$ for the left plot, $(1,0,1,1)$ for the 
middle plot and $(1,1,1,1)$ for the right plot.
}
\label{kinkcl}
\end{figure}
\end{enumerate}

For an even dimension, we also need to consider the case stated in Proposition 
\ref{pro42}. Let $\mu=\nu \exp(\frac{\pi \i}{4})$ and the constant vector
$$
\bn_0=\alpha_1 \bv_1 +\alpha_2 \bv_2 +\alpha_3 \bv_3 +\alpha_4 \bv_4,
$$
where $\alpha_i\in \bbbc$. Notice that
$$
\bn_0^\star=\alpha_1^\star \bv_3 +\alpha_2^\star \bv_2 +\alpha_3^\star \bv_1 
+\alpha_4^\star \bv_4, \quad Q \bv_i=\bv_{i+1}.
$$
The requirement that $\bn_0=Q \bn_0^\star=\alpha_1^\star \bv_4 +\alpha_2^\star 
\bv_3 +\alpha_3^\star \bv_2 
+\alpha_4^\star \bv_1$ implies that $\alpha_1=\alpha_4^\star$ and 
$\alpha_2=\alpha_3^\star$.

In variables
$\xi=\frac{\sqrt 2}{2} (\nu -\nu^{-1}) x$, $\zeta=\frac{\sqrt 2}{2} (\nu + \nu^{-1}) x$ 
and 
$\eta=(\nu^{-2}+\nu^2) t$  we have
\begin{eqnarray*}
&&\Psi_0 (x, t, \mu)\bv_1=e^{\xi-(\zeta+\eta)\i}\bv_1; \quad \Psi_0 (x, t, 
\mu)\bv_2=e^{-\xi-(\zeta -\eta)\i}\bv_2;\\
&& \Psi_0 (x, t, \mu)\bv_3=e^{-\xi+(\zeta-\eta)\i}\bv_3;\quad \Psi_0 (x, t, 
\mu)\bv_4=e^{\xi+(\zeta+\eta)\i}\bv_4.
\end{eqnarray*}

There are three cases:
\begin{enumerate}
 \item $\alpha_1=\alpha_4=0$. Without the loss of generality, we take 
$\alpha_2=e^{\i \beta}=\alpha_3^\star$, where $\beta\in \bbbr$. Then
\begin{eqnarray*}
 \bn=e^{-\xi}\left(\begin{array}{c}  (1+\i) \left(\sin(\zeta-\eta-\beta)-
\cos(\zeta-\eta-\beta)\right) \\ -2 \i \sin(\zeta-\eta-\beta)\\
(-1+\i) \left(\sin(\zeta-\eta-\beta)+\cos(\zeta-\eta-\beta)\right)\\
2 \cos(\zeta-\eta-\beta)
 \end{array}\right)
\end{eqnarray*}
Let $\theta=2(\zeta-\eta-\beta) $. Using (\ref{1kink2}) we get 
\begin{eqnarray*}
&&\tau_1=2 \i e^{-2\xi}  \left( (1-\cos\theta)\nu^6+(1+\sin  \theta) 
\nu^4+(1+\cos\theta)\nu^2+1-\sin \theta\right); \\
&&\tau_2=-2 e^{-2\xi}  \left( (1+\sin\theta)\nu^6+(1+\cos  \theta) 
\nu^4+(1-\sin\theta)\nu^2+1-\cos \theta\right);\\
&&\tau_3=-2 \i e^{-2\xi}  \left( (1-\cos\theta)\nu^2+1+\sin  \theta 
+(1+\cos\theta)\nu^6+(1-\sin \theta)\nu^4\right);\\
&&\tau_4=2 e^{-2\xi}  \left( (1+\sin\theta)\nu^2+1+\cos  
\theta+(1-\sin\theta)\nu^6+(1-\cos \theta)\nu^4\right).
\end{eqnarray*}
Hence we obtain periodic solutions. The left plot in Figure \ref{kinkcb} is its 
density plot.

\item $\alpha_2=\alpha_3=0$. Without the loss of generality, we take 
$\alpha_1=e^{\i \beta}=\alpha_4^\star$, where $\beta\in \bbbr$. 
In this case, we also get periodic solution similar to the above case.
The middle plot in Figure \ref{kinkcb} is its density plot.

\item $\alpha_1 \alpha_2 \alpha_2 \alpha_4\neq 0$. Let $\alpha_1=e^{\i 
\beta}=\alpha_4^\star$ and $\alpha_2=\rho e^{\i 
\gamma}=\alpha_3^\star, \rho\neq 0$,  where $\beta, \gamma, \rho\in \bbbr$. Here we ignore the tedious formula 
and only show their density plots, see the right plot of Figure \ref{kinkcb}.
\begin{figure}[ht]
\centering
\includegraphics[scale=0.35]{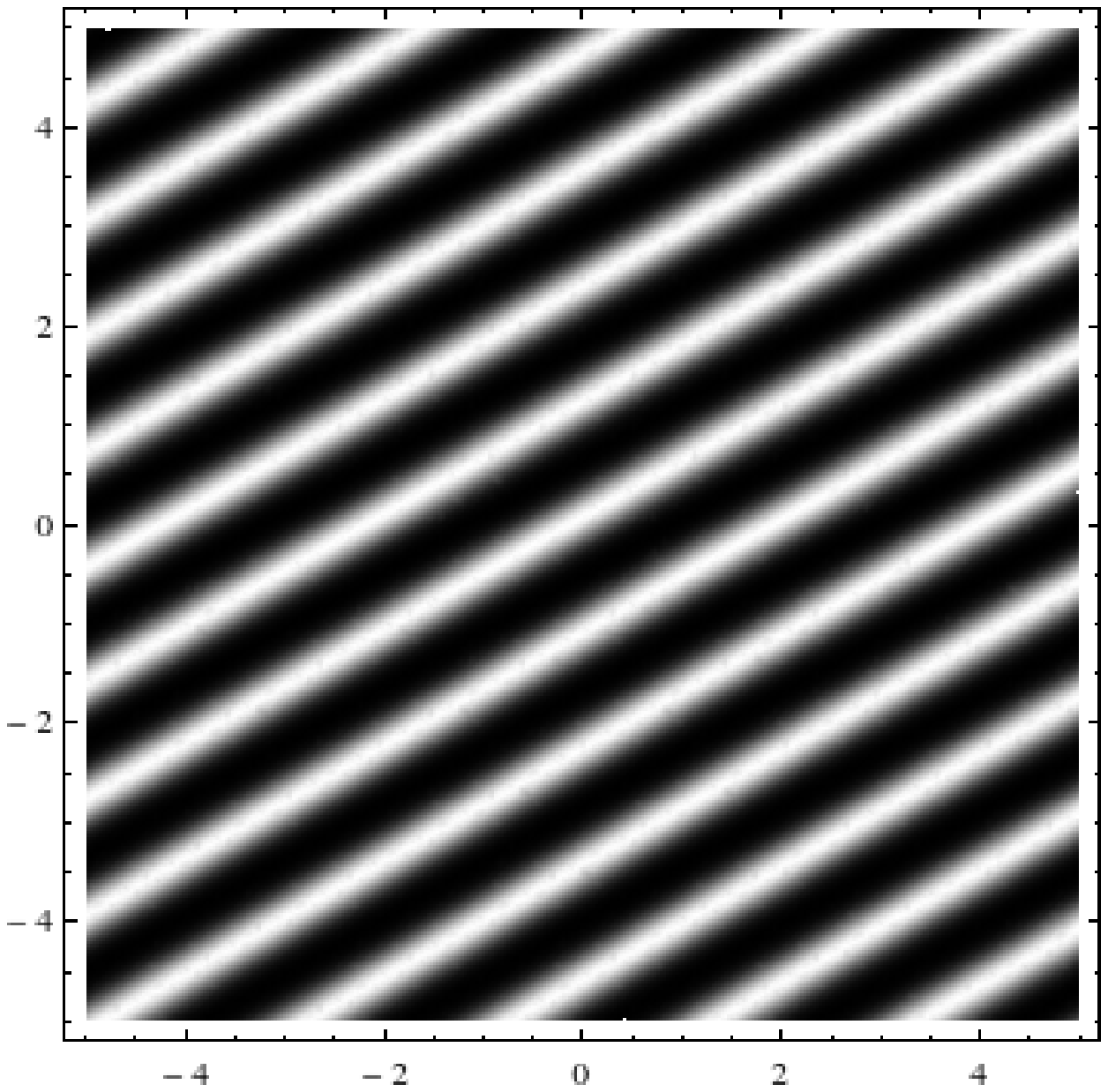}\ \ \
\includegraphics[scale=0.35]{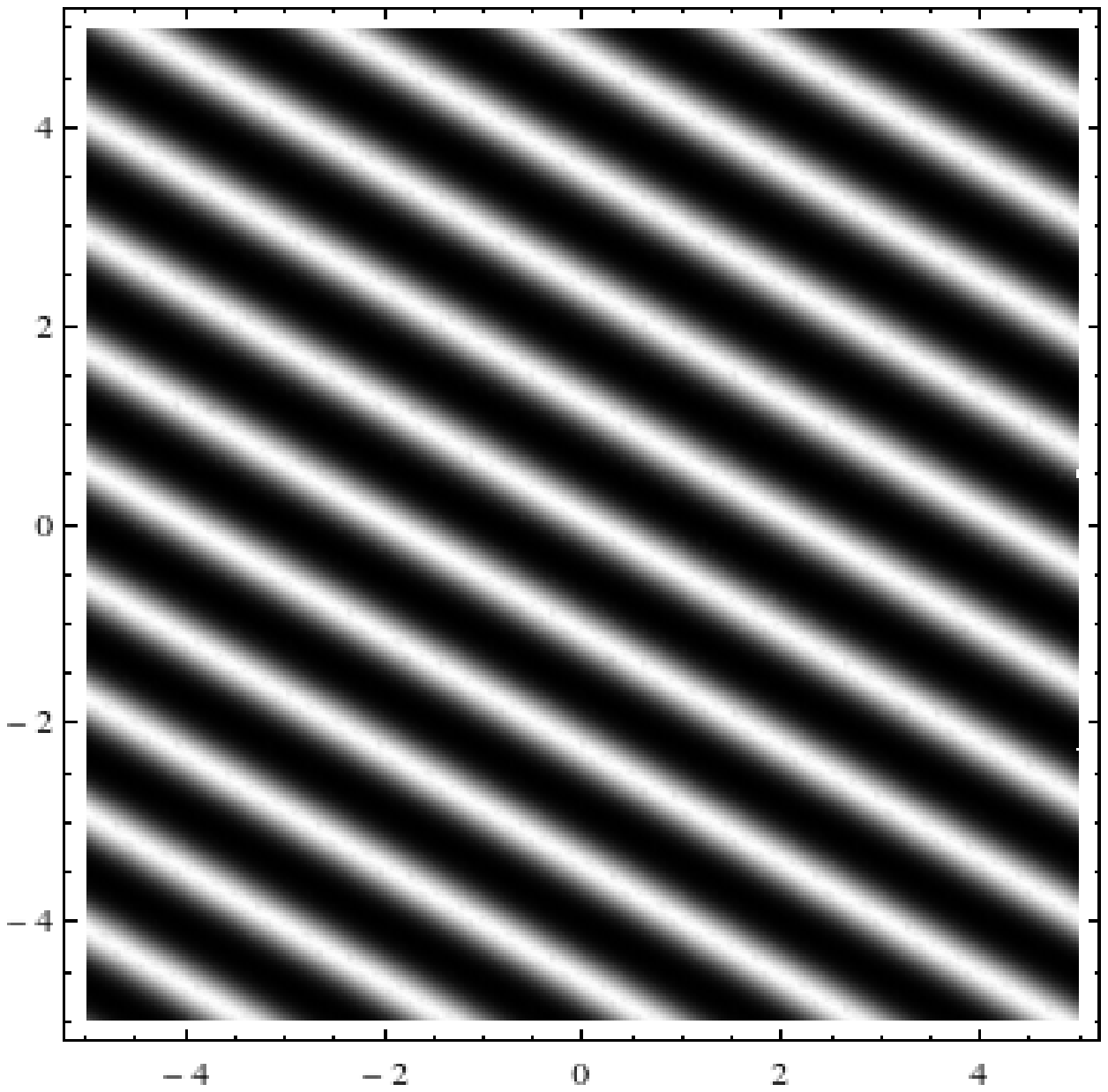} \ \ \
\includegraphics[scale=0.35]{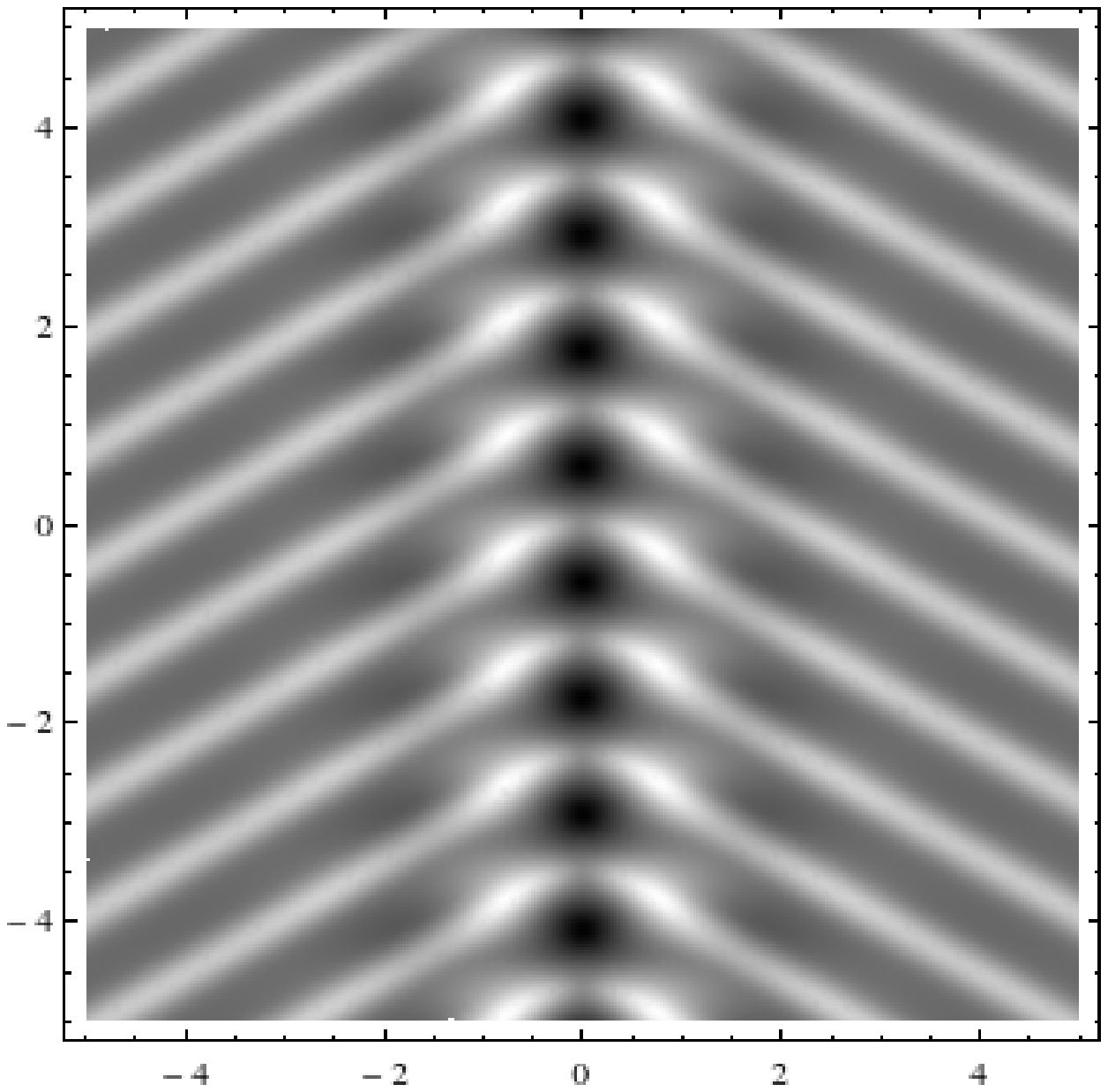}
\caption{Density  plot of $\phi^{(1)}(x,t)$  for kink solutions with 
$\mu=\frac{3}{2}e^{\frac{\pi\i}{4}}$ 
and  $\alpha$ equals to $(0,1,1,0)$ for the left plot, $(1,0,0,1)$ for the 
middle plot and $(1,1,1,1)$ for the right plot.
}
\label{kinkcb}
\end{figure}
\end{enumerate}
Therefore, in the case $N=4$ we have eight different rank $1$ kink solutions.

\subsubsection{Classification of rank $1$ kink solutions for arbitrary 
dimensions}
In this section, we classify all possible rank $1$ kink solutions for 
arbitrary dimension $N$. We have already explored how the solutions for lower dimensions $N=3, 4$.
There is a difference between the dimension $N$ being even or odd.

For arbitrary $N$,  according to Proposition \ref{pro4}, our kink 
solution depends only on $\bn_0\in \bbbr^N$ and $\nu\in \bbbr, \nu \notin 
\{\pm1,0\}$.
We can decompose $\mathbb{R}^N$ as a direct sum of invariant 
subspaces of $\Delta$ as follows:
\begin{eqnarray*}
&&N=2m-1,\quad \mathbb{R}^N=
E^1_0\bigoplus_{p=1}^{m-1}E^2_p ;\\
&& N=2m,\qquad \mathbb{R}^N=E^1_0\bigoplus E^1_m\bigoplus_{p=1}^{m-1}E^2_p,
\end{eqnarray*}
where
\[ E^1_0=\mbox{span}_{\mathbb{R}}({\bf e}_N),\  E^1_m=\mbox{span}_{\mathbb{R}}({\bf e}_m),\   
E^2_p=\mbox{span}_{\mathbb{R}} (\mbox{Re}\, ({\bf e}_p), \mbox{Im}\, ({\bf e}_p)).\]

We define ``elementary waves'' as solutions 
corresponding to the case where $\bn_0$ is simply a 
combination of two eigenvectors. When $N=2m$, there are $m$ elementary 
waves: one pair of
real eigenvalues and $m-1$ pairs of complex conjugate eigenvalues.
When $N=2m-1$, there are $m-1$ elementary wave solutions since there is only 
one real eigenvalue, which leads to 
trivial solutions and we exclude it. The other solutions can be built from 
these elementary wave solutions together with
trivial solutions. 

We are able to write down the elementary wave solutions for arbitrary $N$. To 
do so, we make use the following identity:
For fixed $p\in \bbbn$ and $\omega^N=1$, by direct computation we have
\begin{eqnarray}\label{imid}
\sum_{l=1}^N \omega^{pl} \mu^{2 \{(j-l)\!\! \mod 
N\}}=\frac{\mu^{2N}-1}{\mu^2\omega^{-p}-1} \omega^{pj} , \quad \mu\in \bbbc.
\end{eqnarray}
\begin{The}\label{th1} For any given nonzero constants $\beta, \nu\in \bbbr$, $N\in \bbbn,\ \nu^2\ne 
1,\ N>2$ and $p\in\{1, \cdots \lfloor 
\frac{N-1}{2}\rfloor\}$,
system (\ref{2+1})
 has elementary periodic wave solution of rank $1$  given by
\begin{eqnarray*}
\phi^{(j)}=\frac{1}{2} \ln \frac{\tau_{j-1} \tau_{j+1}}{\tau_j^2}, \qquad 
\tau_j=\frac{\nu^2 \cos(2b-\frac{4p(j+1)\pi}{N}) 
-\cos(2b-\frac{4pj\pi}{N})}{|\nu^2\omega^{2p}-1|^2} +\frac{1}{\nu^2-1}, 
\end{eqnarray*}
where
\begin{eqnarray}\label{bb}
 b=(\nu+\frac{1}{\nu})\sin(\frac{2p \pi}{N}) x+(\nu^2+\frac{1}{\nu^2}) 
\sin(\frac{4p\pi}{N}) t-\beta.
\end{eqnarray}
For even $N=2m$, there is also a time independent rank 
$1$ elementary kink solution of the form 
\begin{eqnarray*}
\phi^{(j)}=\ln\left|\frac{(\beta^2+ e^{-4\xi})(\nu^2+1)
+2 \beta e^{-2 \xi} (-1)^{j}   (\nu^2-1)}{(\beta^2+ e^{-4\xi})(\nu^2+1)-2 \beta 
e^{-2 \xi} (-1)^j   (\nu^2-1)}
\right|, \qquad \xi=(\nu-\frac{1}{\nu})x.
\end{eqnarray*}
\end{The}
\begin{proof}Let us take $\bn_0=e^{\i\beta}\be_p+e^{-\i\beta}\be_{N-p}$ then  
the 
$k$-th component of the vector $\bn=\Psi_0(x,t,\nu)\bn_0$ can be written as 
follows:
\begin{eqnarray*}
&&\bn_k=e^{(\frac{\nu}{\omega^p}-\frac{\omega^p}{\nu})x +
(\frac{\nu^2}{\omega^{2p}}-\frac{\omega^{2p}}{\nu^2})t+\beta \i} ({\bf e}_p)_k 
+e^{(\nu\omega^p-\frac{1}{\omega^p \nu})x +
(\nu^2\omega^{2p}-\frac{1}{\omega^{2p}\nu^2})t-\beta \i} ({\bf e}_{N-p})_k 
=e^{a}\left(e^{-b\i}\omega^{kp} 
+e^{b\i}\omega^{-kp}\right),
\end{eqnarray*}
where
\begin{eqnarray*}
&&a=(\nu-\frac{1}{\nu})\cos(\frac{2p \pi}{N}) x+(\nu^2-\frac{1}{\nu^2}) 
\cos(\frac{4p\pi}{N}) t
\end{eqnarray*}
and $b$ is defined by (\ref{bb}).
It follows from (\ref{1kink}) that
\begin{eqnarray*}
&&\tau_j=\frac{1}{\nu^{2N}-1}e^{2a} 
\sum_{k=1}^N\left(\left(e^{-b\i}\omega^{kp} 
+e^{b\i}\omega^{-kp}\right)^2 \right)   \nu^{2 \{(j-k) \mod N\}}\\
&&\qquad=e^{2a} 
\left(e^{-2b\i}\frac{\omega^{2pj}}{\nu^2\omega^{-2p}-1}+e^{2b\i}\frac{\omega^{
-2pj}}{\nu^2\omega^{2p}-1}+\frac{2}{\nu^2-1} \right)\\
&&\qquad =2 e^{2a}\left(\frac{\nu^2 \cos(2b-\frac{4p(j+1)\pi}{N}) 
-\cos(2b-\frac{4pj\pi}{N})}{|\nu^2\omega^{2p}-1|^2} +\frac{1}{\nu^2-1} \right),
\end{eqnarray*}
which leads to the periodic solutions for $\phi^{(j)}$ given in the statement.

Similarly, in the case $N=2m$, we compute the solution corresponding to
 $\bn_0=\bv_m +\beta \bv_{2m}$. Now we have
$$\bn_k=\beta e^{(\nu-\frac{1}{\nu})x+(\nu^2-\frac{1}{\nu^2})t}+(-1)^k 
e^{(-\nu+\frac{1}{\nu})x+(\nu^2-\frac{1}{\nu^2})t}
=e^{(\nu-\frac{1}{\nu})x+(\nu^2-\frac{1}{\nu^2})t} \left( \beta +(-1)^k 
e^{-2(\nu-\frac{1}{\nu})x})\right),
$$
where $k=1, \cdots N=2m $. This leads to
\begin{eqnarray*}
\tau_j=e^{2(\nu-\frac{1}{\nu})x+2(\nu^2-\frac{1}{\nu^2})t} \left( 
\frac{\beta^2+ e^{-4(\nu-\frac{1}{\nu})x}}{\nu^2-1}
+(-1)^{j+1}\frac{2\beta e^{-2(\nu-\frac{1}{\nu})x}}{\nu^2+1}\right),
\end{eqnarray*}
which gives us the solutions $\phi^{(j)}$ independent of time $t$ given in the 
statement.
\end{proof}
When $N=2m$, according to Proposition \ref{pro42}, to get kink solutions we take $\mu=\nu e^{\frac{\pi \i}{2m}}$
and $\bn_0=Q\bn_0^\star$. It follows from (\ref{even}) there are also $m$ elementary waves.
Similar to Theorem \ref{th1}, we explicitly derive the elementary solutions in this case.
\begin{The}\label{th2} Let $N=2m$, where the integer $m\geq2$. For any given nonzero constants 
$\beta, \nu\in \bbbr$ and $p\in\{1, \cdots m\}$,
system (\ref{2+1})
 has elementary periodic wave solution of rank $1$  given by
\begin{eqnarray*}
\phi^{(j)}=\frac{1}{2} \ln \frac{\tau_{j-1} \tau_{j+1}}{\tau_j^2}, \qquad 
\tau_j=\frac{\nu^2 \cos(2b-\frac{(2p-1)(j+1)\pi}{m}) 
-\cos(2b-\frac{(2p-1)j\pi}{m})}{|\nu^2\omega^{2p-1}-1|^2} +\frac{1}{\nu^2-1}, 
\end{eqnarray*}
where
\begin{eqnarray}\label{bb2}
b=(\nu+\frac{1}{\nu})\sin(\frac{(2p-1) \pi}{2m}) x+(\nu^2+\frac{1}{\nu^2}) 
\sin(\frac{(2p-1)\pi}{m}) t-\beta .
\end{eqnarray}
\end{The}
\begin{proof} For $\mu=\nu\exp(\frac{\pi\i}{2m})$, we take $\bn_0=e^{\i\beta}\be_p+e^{-\i\beta}\be_{2m-p+1}$
following from (\ref{even}) and then  
the $k$-th component of the vector $\bn=\Psi_0(x,t,\nu)\bn_0$ can be written as 
follows:
\begin{eqnarray*}
&&\bn_k=e^{(\frac{\nu}{\omega^{p-\frac{1}{2}}}-\frac{\omega^{p-\frac{1}{2}}}{\nu})x +
(\frac{\nu^2}{\omega^{2p-1}}-\frac{\omega^{2p-1}}{\nu^2})t+\beta \i} ({\bf e}_p)_k 
+e^{(\nu\omega^{p-\frac{1}{2}}-\frac{\omega^{\frac{1}{2}-p}}{\nu})x
+(\nu^2\omega^{2p-1}-\frac{1}{\omega^{2p-1}\nu^2})t-\beta \i} 
({\bf e}_{2m-p+1})_k \\
&&=e^{a}\left(e^{-b\i}\omega^{kp} 
+e^{b\i}\omega^{-k(p-1)}\right),
\end{eqnarray*}
where
\begin{eqnarray*}
&&a=(\nu-\frac{1}{\nu})\cos(\frac{(2p-1)\pi}{2m}) x+(\nu^2-\frac{1}{\nu^2}) 
\cos(\frac{(2p-1)\pi}{m}) t
\end{eqnarray*}
and $b$ is defined by (\ref{bb2}).
It follows from (\ref{1kink2}) that
\begin{eqnarray*}
&&\tau_j=\frac{1}{\nu^{2N}-1}e^{2a} 
\sum_{k=1}^N\left(\left(e^{-b\i}\omega^{kp} 
+e^{b\i}\omega^{-k(p-1)}\right)^2 \right)   \mu^{2 \{(j-k) \mod N\}}\\
&&\qquad=e^{2a} 
\left(e^{-2b\i}\frac{\omega^{2pj}}{\mu^2\omega^{-2p}-1}+e^{2b\i}\frac{\omega^{
(2-2p)j}}{\mu^2\omega^{2p-2}-1}+\frac{2\omega^j}{\mu^2\omega^{-1}-1} \right)\\
&&\qquad =2 e^{2a}\omega^j\left(\frac{\nu^2 \cos(2b-\frac{(2p-1)(j+1)\pi}{m}) 
-\cos(2b-\frac{(2p-1)j\pi}{m})}{|\nu^2\omega^{2p-1}-1|^2} +\frac{1}{\nu^2-1} \right),
\end{eqnarray*}
which leads to the periodic solutions for $\phi^{(j)}$ given in the statement. 
\end{proof}

Elementary rank $1$ kink solutions correspond to two dimensional 
$\Delta$--invariant subspaces of $\bbbr^N$. Other rank 1 solutions correspond 
to invariant subspaces of dimension $3,4,\ldots, N$. The number of all 
possible $\Delta$--invariant subspaces gives us the number of all rank 1 
solutions.
\begin{The} 
Equation (\ref{2+1}) with odd  $N=2m-1$ has $2^m-2$ different 
rank $1$ kink solutions. In the case of even  $N=2m$ it has
 $3 \cdot 2^{m}-4$ different rank $1$ kink solutions.
\end{The}
\begin{proof} 
When $N=2m-1$, there are $m-1$ elementary solutions from $\bn_0\in E^2_p$ for each $p=1, 2, \cdots, m-1$, 
and one constant solution
from $\bn_0\in E^1_0$.  We can build up other solutions by taking
any combination of them. For example, there are $C_m^2=\frac{m(m-1)}{2}$ 
different solutions if we take any two combinations.
Thus, the total number of different rank $1$ kink solutions is
$$
m-1+\sum_{k=2}^{m} C_m^k=2^m-2 .
$$
When $N=2m$, when $\mu=\nu\in\bbbr$ there are $m-1$ elementary solutions from 
$\bn_0\in E^2_p$ for each $p=1, 2, \cdots, m-1$, one elementary
solution from $\bn_0 \in E^1_0\bigoplus E^1_m$ and two constant solutions from 
$\bn_0\in E^1_0$ or $\bn_0\in E^1_m$.
we can build up other solutions by taking any combinations of $m-1$ elementary 
solutions alone or together with either one real
or both real eigenvectors.
Thus, the total number of different rank $1$ kink solutions in this case is
$$
1+4 \sum_{k=1}^{m} C_{m-1}^k =4 (2^{m-1}-1)+1=2^{m+1}-3.
$$
When $N=2m$, when $\mu=\nu e^{\frac{\pi\i}{2m}}, \nu\in\bbbr$ there are 
also $m$ elementary solutions. In this case the total number of different rank 
$1$ kink solutions is
$$
\sum_{k=1}^{m} C_m^k=2^m-1.
$$
Hence when $N=2m$, the total number of different rank $1$ kink solutions is
$3\cdot 2^m-4 $. Thus 
we complete the proof.
\end{proof}
Notice that the statement is consistent with the concrete results for $N=3$ and 
$N=4$.
We plot some density plots of $\phi^{(1)}$ and snapshots for $N=5$ when 
$\nu=0.4$ and $\alpha$ are chosen as stated in Figures
\ref{kink51111} and \ref{kink511111}. 
\begin{figure}[ht]
\centering
\includegraphics[scale=0.5]{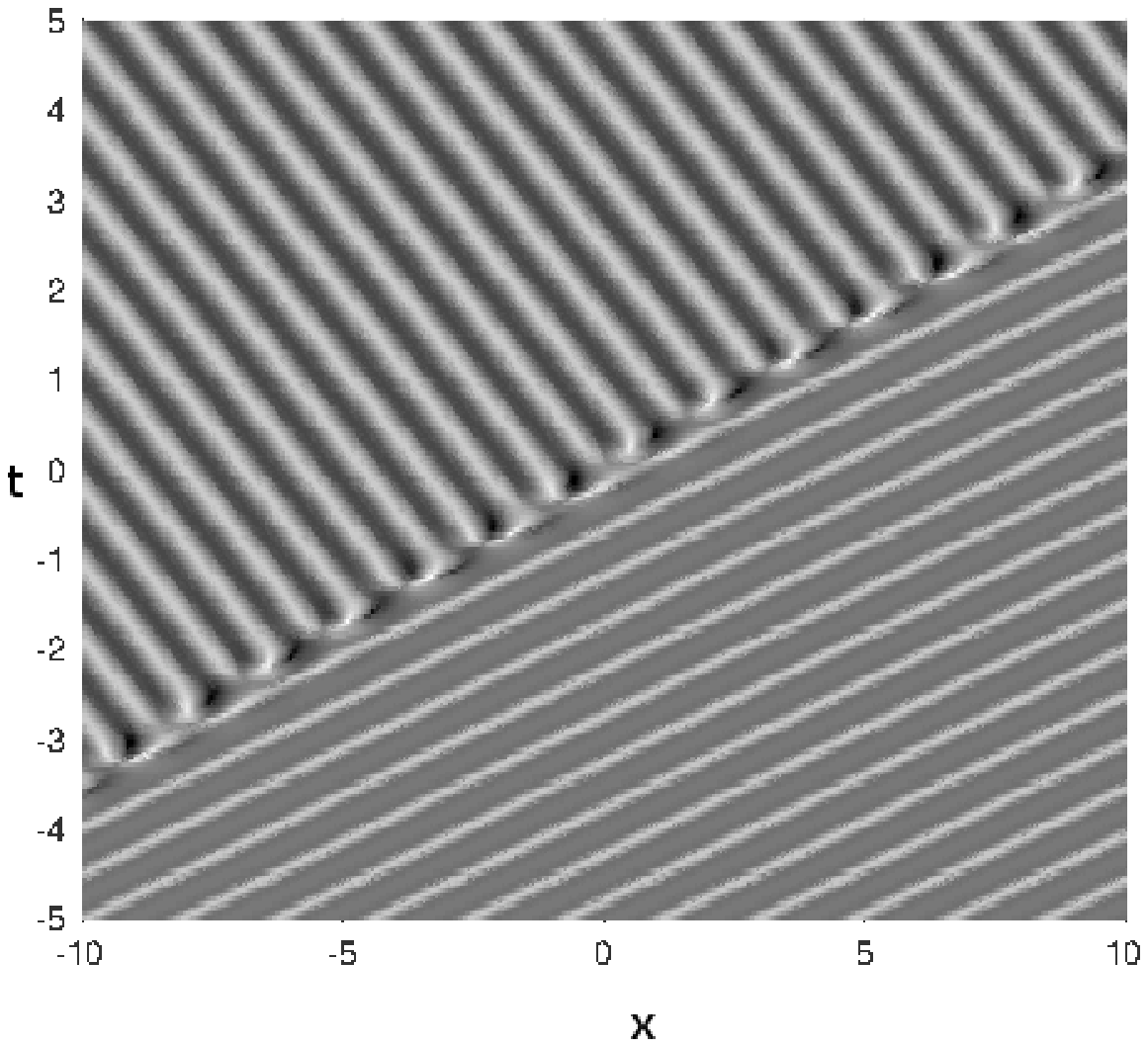}
\includegraphics[scale=0.467]{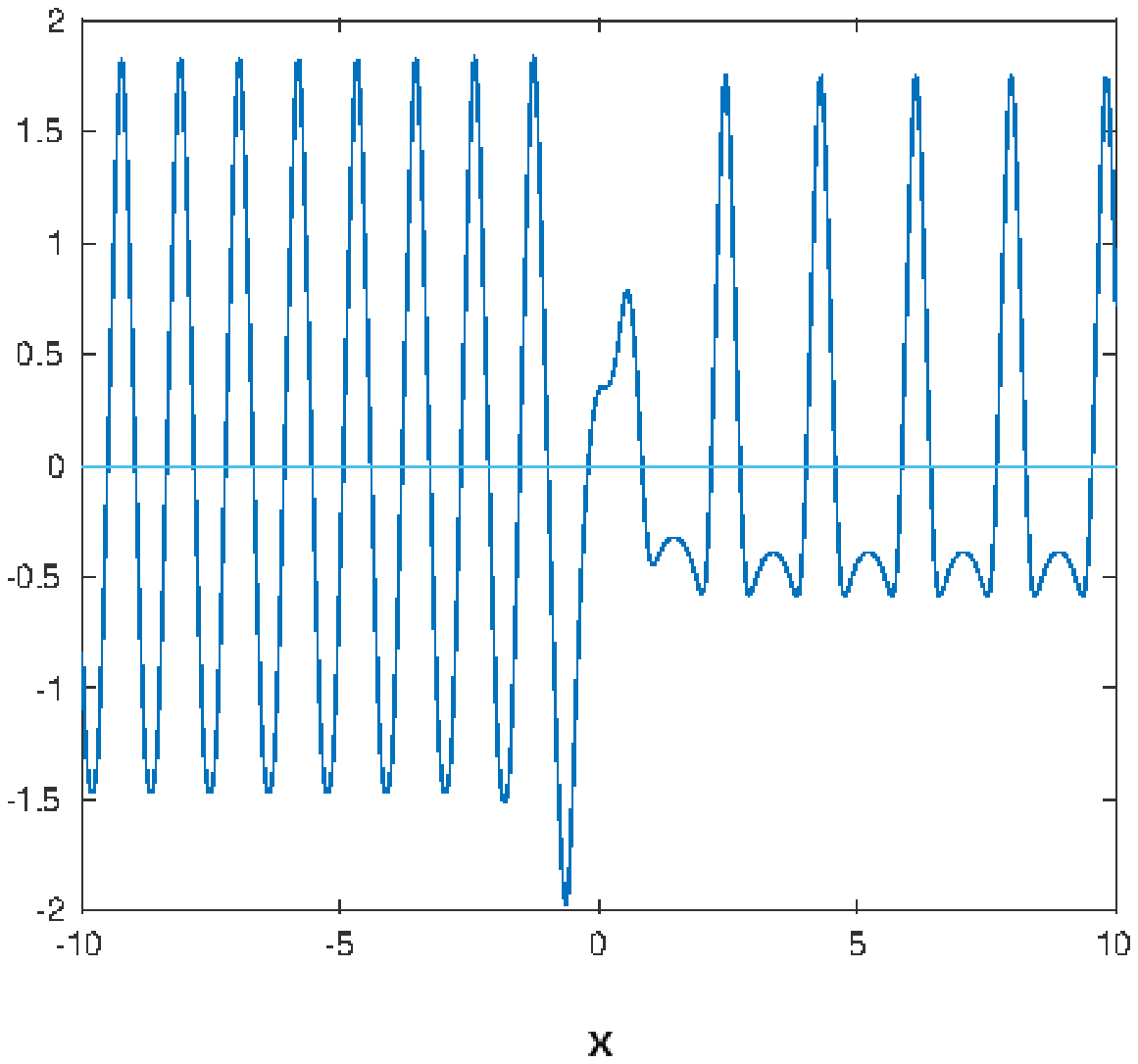}
\caption{Density  plot of $\phi^{(1)}(x,t)$  and a snapshot of $\phi^{(1)}$ at 
t=0 ($\alpha=(1,1,1,1,0)$, $\nu=0.4$) .
}
\label{kink51111}
\end{figure}

\begin{figure}[ht]
\centering
\includegraphics[scale=0.5]{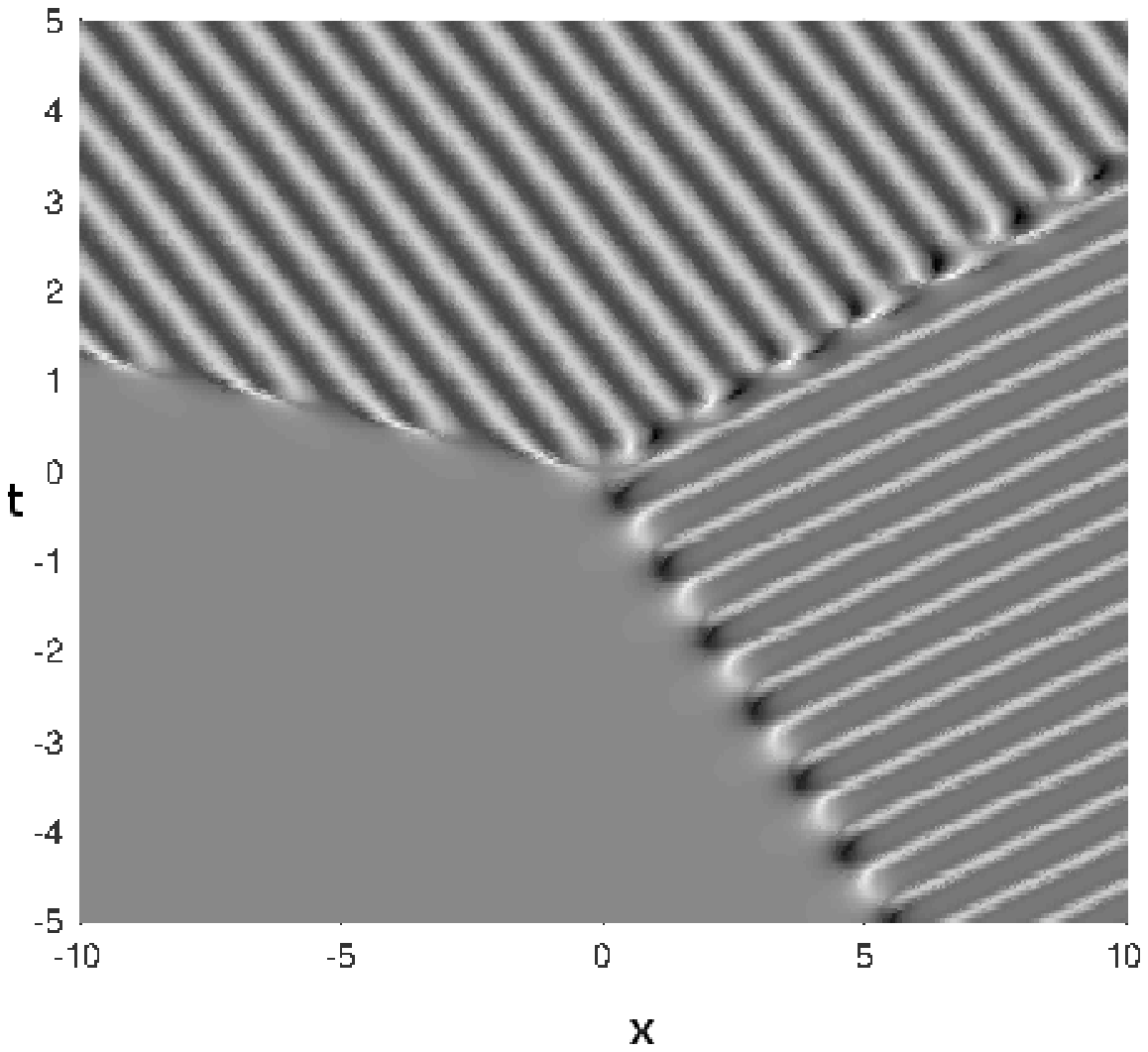}
\includegraphics[scale=0.467]{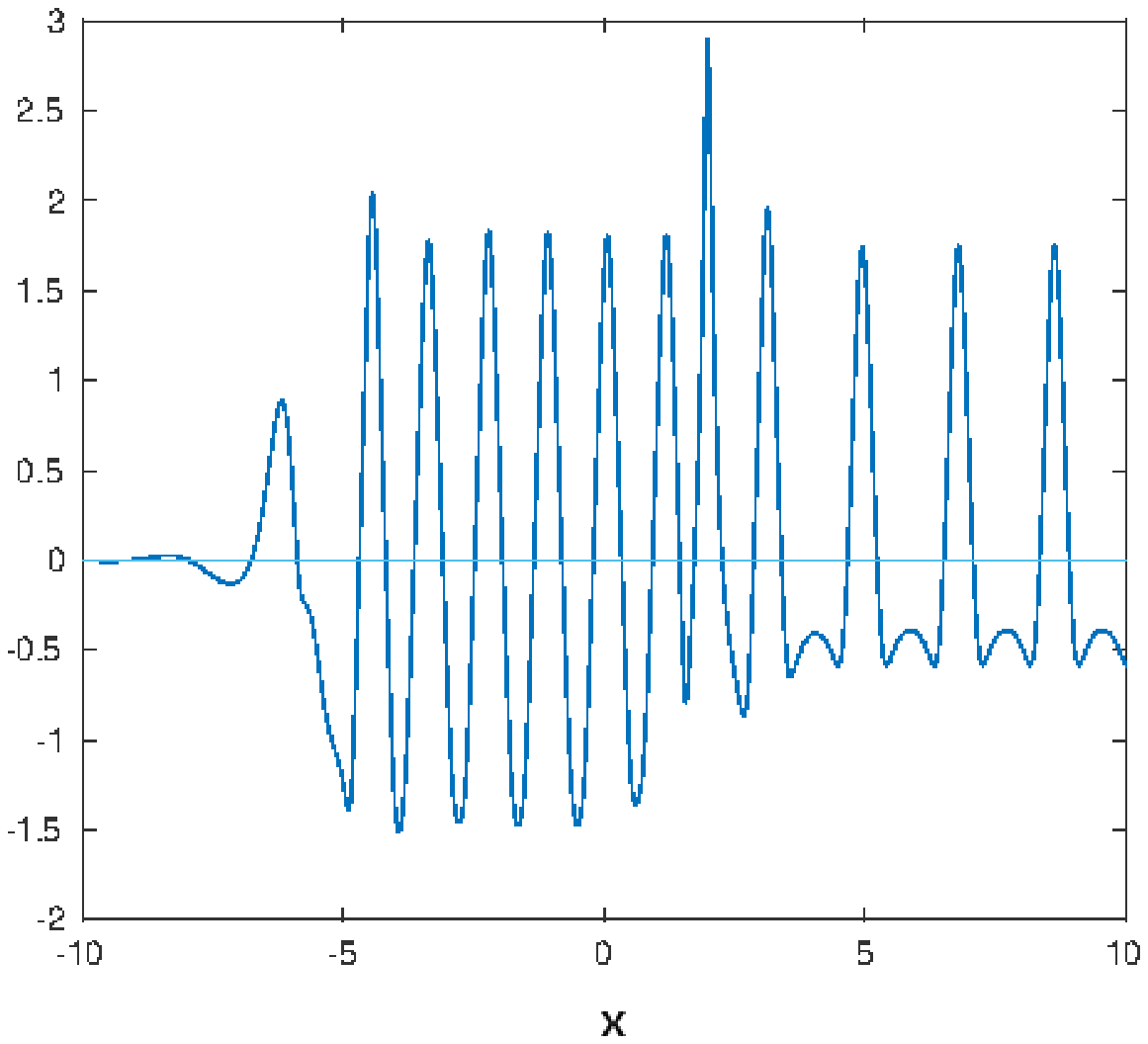}
\caption{Density  plot of $\phi^{(1)}(x,t)$  and a snapshot of $\phi^{(1)}$ at 
t=0.7  ($\alpha=(1,1,1,1,1)$, $\nu=0.4$) .
}
\label{kink511111}
\end{figure}

\subsubsection{Tropicalisation and wave front trajectories} \label{tropsec}

In the general case 
of rank 1 ``kink'' solutions the trajectories of the wave fronts can be 
understood geometrically. According Proposition \ref{pro4} the 
$(x,t)$ dependence of the solution is determined by the vector $\bn$ which can 
be 
presented in the form 
\[
 \bn=\sum_{k=1}^N e^{\Theta_k(x,t)} \be_k ,
\]
where 
\[
 \Theta_k(x,t)=
(\nu\omega^{-k}-\nu^{-1}\omega^{k})x+(\nu^2\omega^{-2k}-\nu^{-2}\omega^{2k
} )t+\log \alpha_k.
\]
The imaginary part $\Theta_k^{\rm Im}(x,t)$ is responsible for oscillations 
of the solution, while the real part 
\[
 \Theta_k^{\rm Re}(x,t)=(\nu-\nu^{-1})\cos(\frac{2\pi k}{N})x+(\nu^2-\nu^{ -2} 
)\cos(\frac{4\pi k}{N}) t+\log| \alpha_k|
\]
tells us which term in the sum is dominant  at a given point $(x,t)$. In a 
region where only one term in the sum is dominant, and thus we can ignore other 
terms, the solution is close to the trivial (zero) solution. In regions where 
two terms have the same real exponent ($ \Theta_k^{\rm Re}(x,t)= 
\Theta_{-k}^{\rm Re}(x,t) $) we observe elementary waves. The boundaries of 
these regions correspond to the wave fronts. Thus the wave fronts can be 
described as follows: We consider a 
set of  linear functions
$\Theta_k^{\rm Re}(x,t),\ k=1,\ldots, N$
and define a continuous piecewise linear function
\begin{equation}\label{tropTheta}
  \Theta(x,t)=\max (\Theta_1^{\rm Re}(x,t),\cdots, \Theta_N^{\rm Re}(x,t)).
\end{equation}
The locus where the function $\Theta(x,t)$ is not smooth corresponds to the 
wave fronts. To compare the numerical result for wave fronts with the locus 
described above one can compare Figures \ref{kink511111} and \ref{kinkTheta}.
\begin{figure}[ht]
\centering
\includegraphics[scale=0.5]{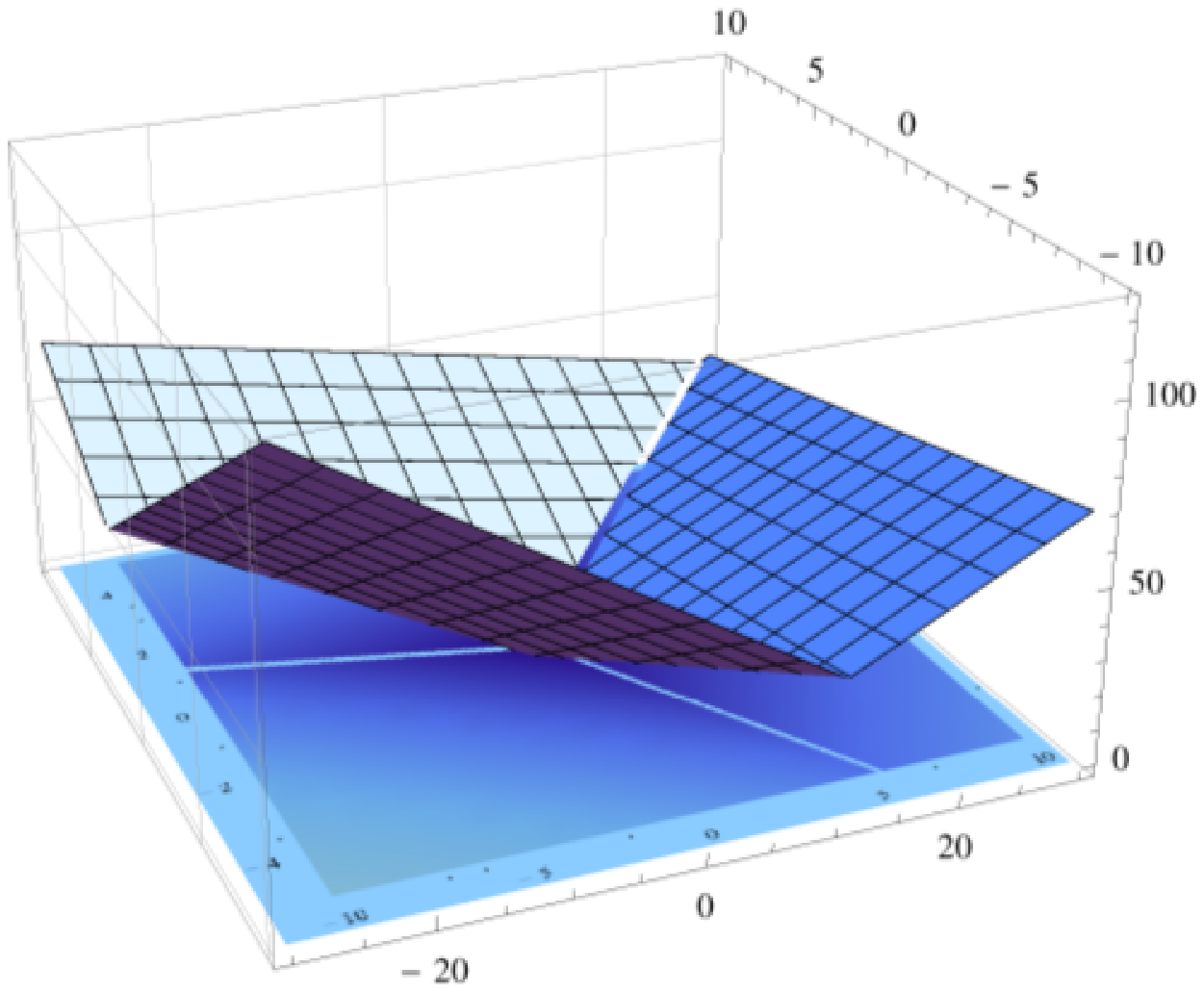}
\includegraphics[scale=0.4]{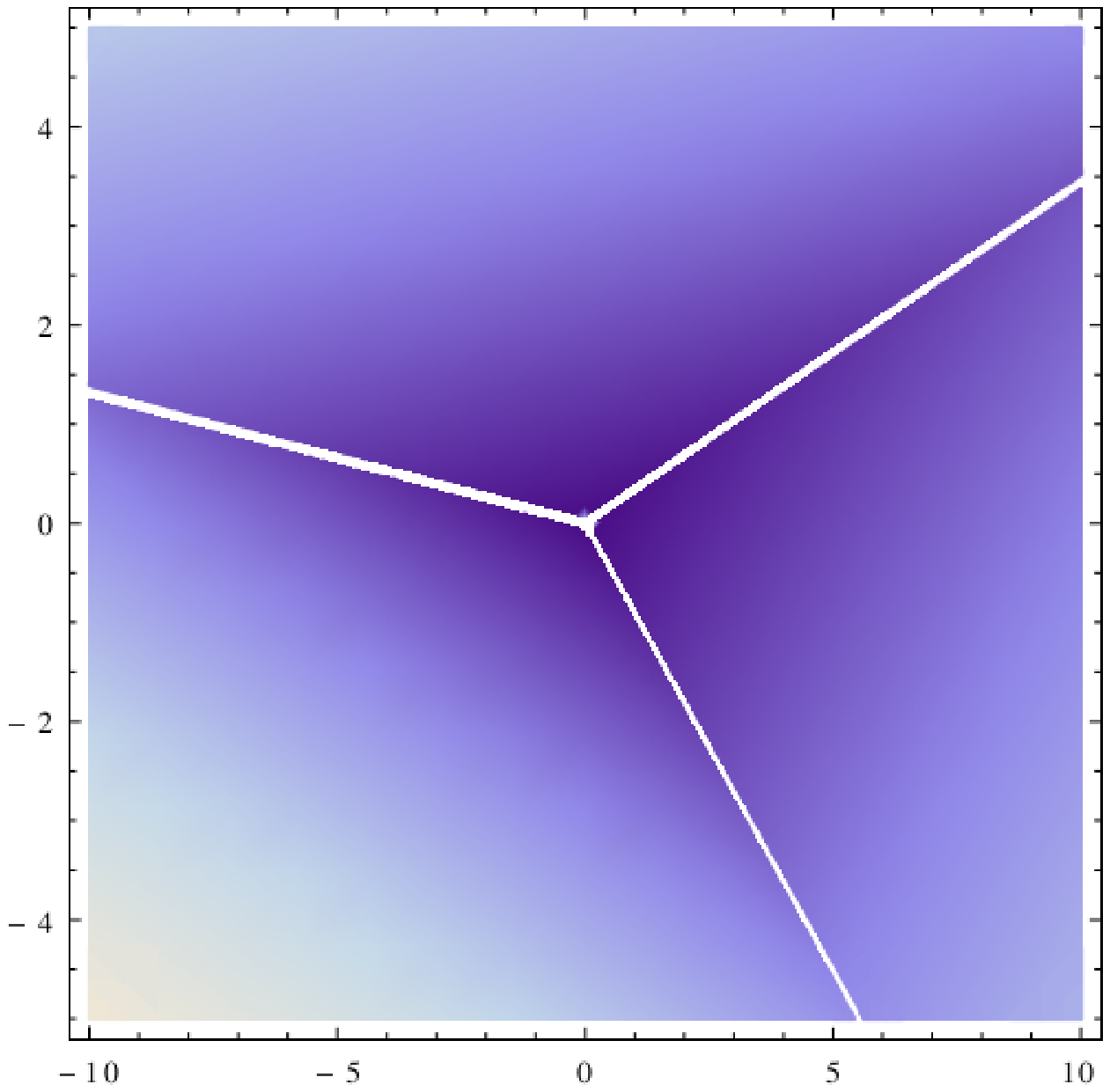}
\caption{Density and 3-d plots for $\Theta(x,t)$, $N=5$, $\alpha=(1,1,1,1,1)$, 
$\nu=0.4$ (compare with Figure \ref{kink511111}).
}
\label{kinkTheta}
\end{figure}
This construction is similar to tropicalisation and soliton graphs proposed by 
Kodama and Williams for the case of KP solitons \cite{kodama2014kp}, although 
there is a slight difference, since we do not use rescaling in our definition 
and keep the logarithmic term $\log|\alpha_k|$, which disappears in the scaling 
limit.

\subsection{Classification of rank $1$ breather solutions for arbitrary 
dimensions}
In this section, we classify all possible rank $1$ breather solutions for 
arbitrary dimension $N$. According to Proposition \ref{pro5}, our soliton 
solution depends only on $\bn_0\in \bbbc^N$, $\mu\in \bbbc$ and $|\mu|\notin 
\{0,1\}$.
In a similar manner to the case for kinks, a natural way to classify possible solutions in 
terms of 
$\bn_0$ is to first consider eigenvectors and eigenvalues of the constant matrix
$\Delta$. We decompose $\bbbc^N$ as a direct sum of invariant 
subspaces of $\Delta$ as follows:
\begin{eqnarray*}
 \mathbb{C}^N=\bigoplus_{p=1}^{N}E^1_p, \qquad  E^1_p=\mbox{span}_{\mathbb{C}}({\bf e}_p) .
\end{eqnarray*}
The 
vector $\bn_0$ 
in this basis 
\begin{eqnarray}\label{n0b}
 \bn_0=\sum_{p=1}^N \alpha_p \bv_p, \qquad \alpha_p\in \bbbc
\end{eqnarray}
is 
given by a matrix $\alpha=(\alpha_1,\ldots ,\alpha_N)$.
We immediately get the following result:
\begin{Pro}\label{pro10} For a constant complex vector $\bn_0$ in the form of 
(\ref{n0b}), if 
there is only one $\alpha_p$ is nonzero, then solutions for (\ref{2+1}) are 
trivial, that is, $\phi^{(i)}=0$.
\end{Pro}
\begin{proof} 
It follows from (\ref{imid}) that
\begin{eqnarray*}
\frac{1}{\mu^{2N}-1}\sum_{l=1}^N \omega^{2pl}   \mu^{2 \{(i-l)\!\!\! 
\mod N\}}=\frac{\omega^{2p(i+1)}}{\mu^2-\omega^{2p}} ; \qquad
\frac{1}{|\mu|^{2N}-1} \sum_{l=1}^N  |\mu|^{2 \{(i-l)\!\!\!  \mod N\}}
=\frac{1}{|\mu|^2-1} .
\end{eqnarray*}
Thus we can compute $\tau_i$ in (\ref{1breather}) as follows:
\begin{eqnarray*}
 \tau_i=|\alpha_p|^4 |e^{(\frac{\mu}{\omega^p}-\frac{\omega^p}{\mu})x +
(\frac{\mu^2}{\omega^{2p}}-\frac{\omega^{2p}}{\mu^2})t}|^4 
\left(\frac{1}{(|\mu|^2-1)^2}-\frac{1}{|\mu^2-\omega^{2p}|^2} \right),
\end{eqnarray*}
which is independent of $i$. According to Proposition \ref{pro5}, we get 
solutions $\phi^{(i)} =0$ as stated.
\end{proof}
We now consider the case when there are only two nonzero components, say 
$\alpha_p$ and $\alpha_q$ among all $\alpha_i,i=1,\cdots, N$, that is,
\begin{eqnarray}
 \bn_0=\alpha_p \bv_p+\alpha_q \bv_q, \qquad \alpha_p \alpha_q \neq 0 , \qquad 
q>p.
\end{eqnarray}
It follows that the $k$-component of vector $\bn$ is
\begin{eqnarray*}
 n_k=\alpha_p e^{(\frac{\mu}{\omega^p}-\frac{\omega^p}{\mu})x +
(\frac{\mu^2}{\omega^{2p}}-\frac{\omega^{2p}}{\mu^2})t}\omega^{kp} +\alpha_q
e^{(\frac{\mu}{\omega^q}-\frac{\omega^q}{\mu})x +
(\frac{\mu^2}{\omega^{2q}}-\frac{\omega^{2q}}{\mu^2})t} 
\omega^{kq}=A(\omega^{kp}+\gamma \omega^{kq}),
\end{eqnarray*}
where we introduce notations for $A, \gamma\in \bbbc$ to shorten the expressions 
of 
$\sigma(j)$ and $\rho(j)$ defined by (\ref{srho}). We have
\begin{eqnarray*}
&&\sigma(j)=A^2(\frac{\omega^{2p(j+1)}}{\mu^2-\omega^{2p}} +2 \gamma 
\frac{\omega^{(p+q)(j+1)}}{\mu^2-\omega^{p+q}}+\gamma^2 
\frac{\omega^{2q(j+1)}}{\mu^2-\omega^{2q}}) ;\\
&&\rho(j)=|A|^2\left((1+|\gamma|^2) \frac{1}{|\mu|^2-1}+\gamma^* 
\frac{\omega^{(p-q)(j+1)}}{|\mu|^2-\omega^{p-q}} +\gamma 
\frac{\omega^{(q-p)(j+1)}}{|\mu|^2-\omega^{q-p}}\right) .
\end{eqnarray*}
According to Proposition \ref{pro5}, the breather solutions depend on 
$\gamma$ and $\mu$ since $A$ is cancelled when we compute the solutions. Let 
$\mu=|\mu|e^{i\delta}$.
The breather trajectory  is 
determined by the condition $|\gamma|=1$, where
\begin{eqnarray*}
&&|\gamma|=|\frac{\alpha_q}{\alpha_p}| 
|e^{(\frac{\mu}{\omega^q}-\frac{\omega^q}{\mu}-\frac{\mu}{\omega^p}+\frac{
\omega^p}{\mu})x 
+(\frac{\mu^2}{\omega^{2q}}-\frac{\omega^{2q}}{\mu^2}-\frac{\mu^2}{\omega^{2p}}
+\frac{\omega^{2p}}{\mu^2})t }|\\
&&\qquad =e^{(|\mu|-\frac{1}{|\mu|})(\cos 
(\delta-\frac{2\pi q}{N})-\cos(\delta-\frac{2\pi p}{N}))x 
+(|\mu|^2-\frac{1}{|\mu|^2})(\cos 
(2\delta-\frac{4\pi q}{N})-\cos(2\delta-\frac{4\pi p}{N}))t +\ln 
|\frac{\alpha_q}{\alpha_p}|  }\\
&&\qquad =e^{2 (|\mu|-\frac{1}{|\mu|})\sin 
(\delta-\frac{(p+q)\pi }{N})\sin\frac{(q-p)\pi }{N}(x-v_{pq}t-x_{pq}^0) } .
\end{eqnarray*}
It reflects the balance between the exponents.
Thus the speed of the breather is given by
\[v_{pq} = -4\left(\frac{1}{|\mu|}+|\mu|\right)
\cos\left(\delta-\frac{\pi (p+q)}{N}\right)\cos\left(\frac{\pi 
(q-p)}{N}\right)\]
and it is shifted to the right along the $x$-axis by
\[x_{pq}^0 = \frac{\ln |\frac{\alpha_p}{\alpha_q}| }{2 
(|\mu|-\frac{1}{|\mu|})\sin 
(\delta-\frac{(p+q)\pi }{N})\sin\frac{(q-p)\pi }{N}} .\]
It is localised in $x$ and of size $L_{pq}$
\[
 L^{-1}_{pq}=2 
(|\mu|-|\mu|^{-1})\sin 
(\delta-\frac{(p+q)\pi }{N})\sin\frac{(q-p)\pi }{N}.
\]
The rank $1$ breather solutions can be obtained in the following way:
\begin{itemize}
 \item There  are $C_N^2$ possible choices of two dimensional 
$\Delta$--invariant subspaces in $\bbbc^N$ and therefore there are $C_N^2$ 
elementary breathers.
\item Solutions corresponding to three dimensional invariant 
subspaces, i.e. 
$$\bn_0=\alpha_p \be_p+\alpha_q \be_q+\alpha_r \be_r,\ \alpha_p\alpha_q\alpha_r 
\ne 0$$  represent decays or fusions of breathers (``Y'' shape), and there are 
$C_N^3$ such solutions.
\item Solutions corresponding to four dimensional invariant subspaces, i.e. 
$$\bn_0=\alpha_p \be_p+\alpha_q \be_q+\alpha_r \be_r+\alpha_s \be_s,\ 
\alpha_p\alpha_q\alpha_r \alpha_s 
\ne 0$$  represent solutions combining 2 ``Y'' shapes (``2Y'' shape solutions). 
There are 
$C_N^4$ such solutions, etc.
\end{itemize}
Examples of ``Y'', ``2Y'' and ``3Y'' configurations in the case $N=5$ are 
presented in Figure \ref{br111}.
\begin{figure}[ht]
\centering
\includegraphics[scale=0.4]{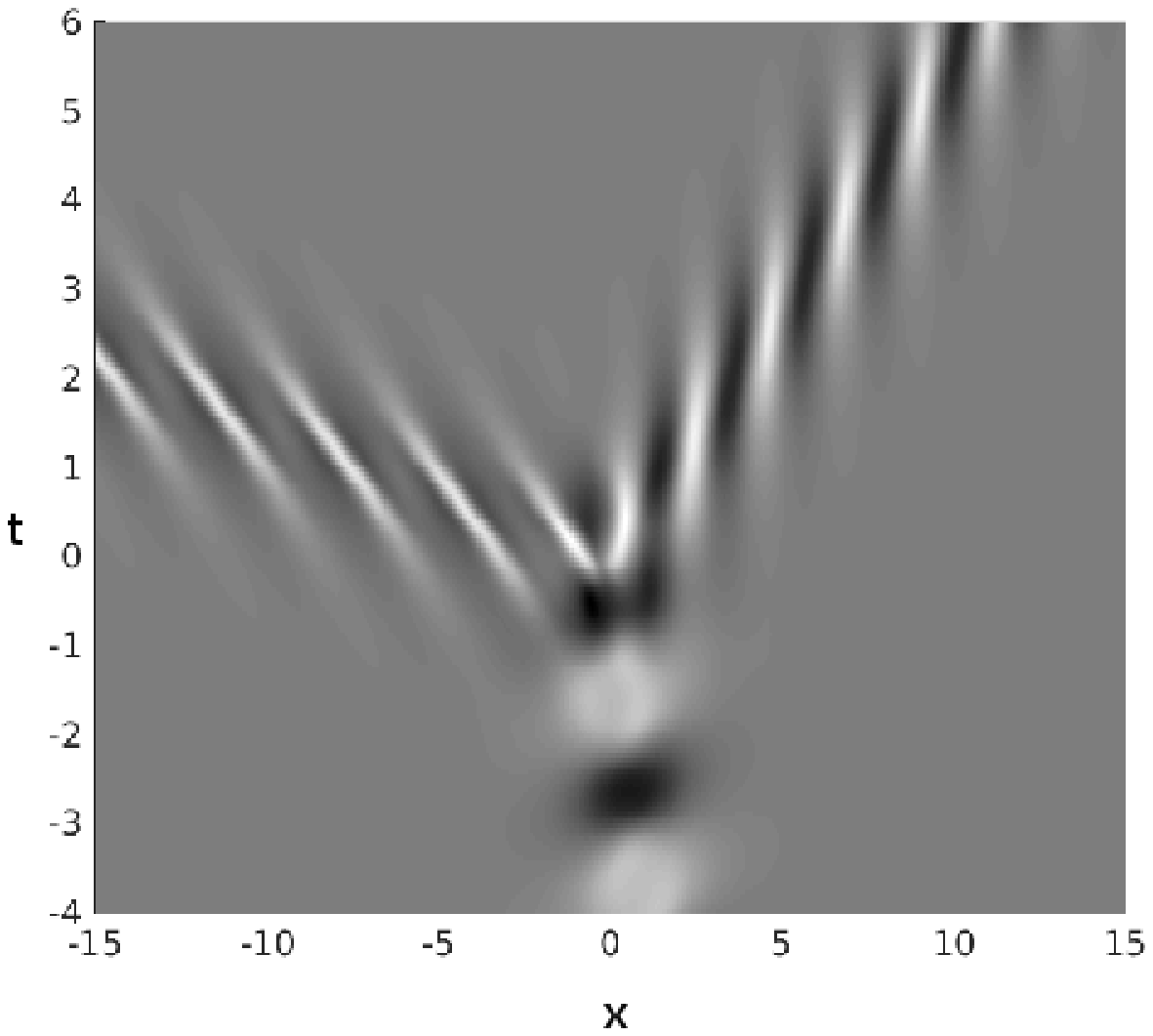}
\includegraphics[scale=0.4]{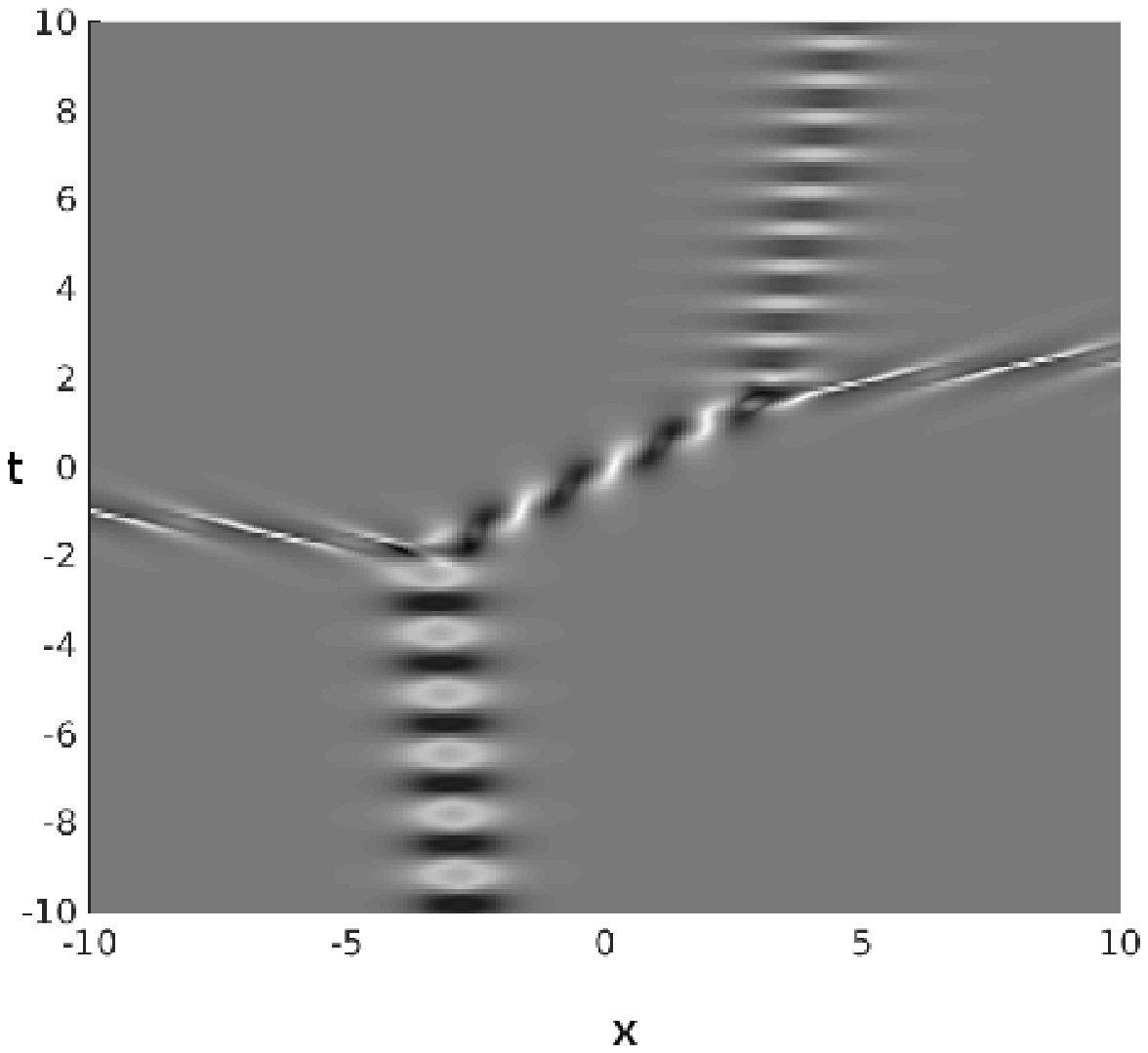} 
\includegraphics[scale=0.4]{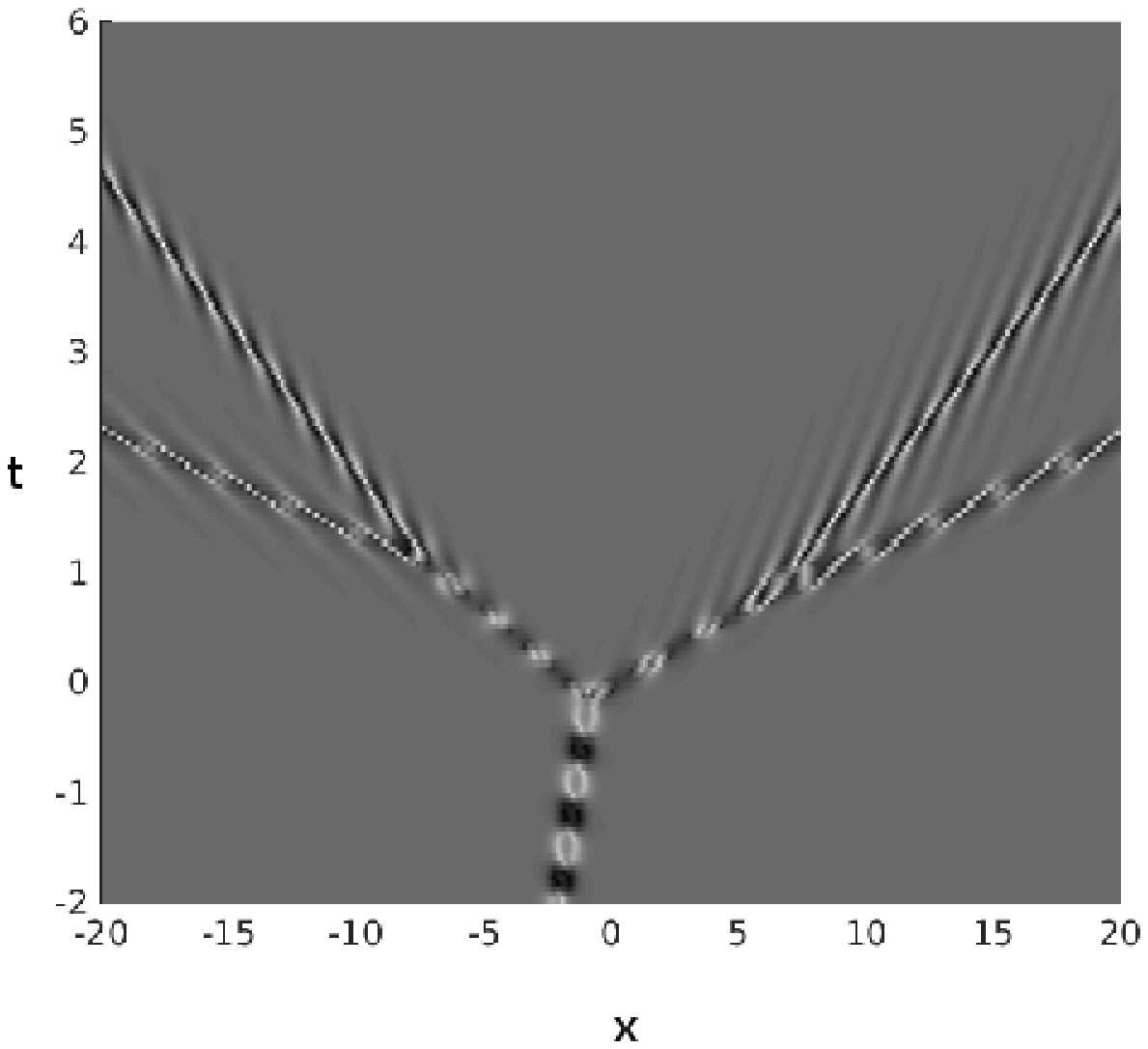}
\caption{Density  plots of $\phi^{(1)}(x,t)$  for breather solutions $N=5$ with
$\mu=0.7+0.15\i
,\ \alpha=(1, 0, 1, 0, 1)$, and 
$\mu=0.5+0.15\i
,\ \alpha=(1, 0.01, 1, 0, 0.0001)$, and
$\mu=0.3+0.15\i
,\ \alpha=(0.01, 1, 10, 10\i, 0.1)$ respectively.}
\label{br111}
\end{figure}

The total number of possible distinct configurations for a breather solution 
given in Proposition \ref{pro5} is
\[
 \sum_{k=2}^N C_N^k=2^N-N-1.
\]

The type of the breather solution depends on the choice of the matrix $\alpha 
=(\alpha_1,\ldots,\alpha_N)$. The explicit expression for the solution is given 
in Proposition \ref{pro5} and it is quite complicated, but the method of 
tropicalisation, which we used in Section  \ref{tropsec}, enables us to give a 
simple description of the soliton graph. We explore the observation that the vector 
\[
 \bn=\sum_{k=1}^N e^{\Theta_p(x,t)}\be_p,
\]
where
\[
 \Theta_p(x,t)=(\frac{\mu}{\omega^p}-\frac{\omega^p}{\mu})x +
(\frac{\mu^2}{\omega^{2p}}-\frac{\omega^{2p}}{\mu^2})t+\log \alpha_p, \quad 
\mu=|\mu| e^\delta
\]
completely determines the $(x,t)$ dependence of the solution. In the regions 
where only one term is dominant the solution is exponentially small. We can 
define the tropical graph of the breather as a locus where two or more terms 
are in balance. More precisely, let us consider the real part of  
$\Theta_p(x,t)$, that is,
\[
 \Theta_p^{\rm Re}(x,t)=(|\mu|-|\mu|^{-1})\cos(\delta-\frac{2\pi 
p}{N})x+(|\mu|^2-|\mu|^{-2})\cos(2\delta-\frac{4\pi p}{N})t+\log |\alpha_p|
\]
and a piecewise linear continuous function of variables $(x,t)$:
\[
 \Theta(x,t)=\max_p \Theta _p^{\rm Re}(x,t).
\]
\begin{Def}
 For rank one breather solutions the tropical soliton graph is defined as a 
locus of points where the function 
$\Theta(x,t)$ is not 
smooth. 
\end{Def}

In order to visualise the tropical plot we present the density plot for the
piecewise constant function  $\partial_x \Theta(x,t)$. In Figure \ref{br111t} 
we show the plots corresponding to solutions plotted in Figure \ref{br111}. 
\begin{figure}[ht]
\centering
\includegraphics[scale=0.5]{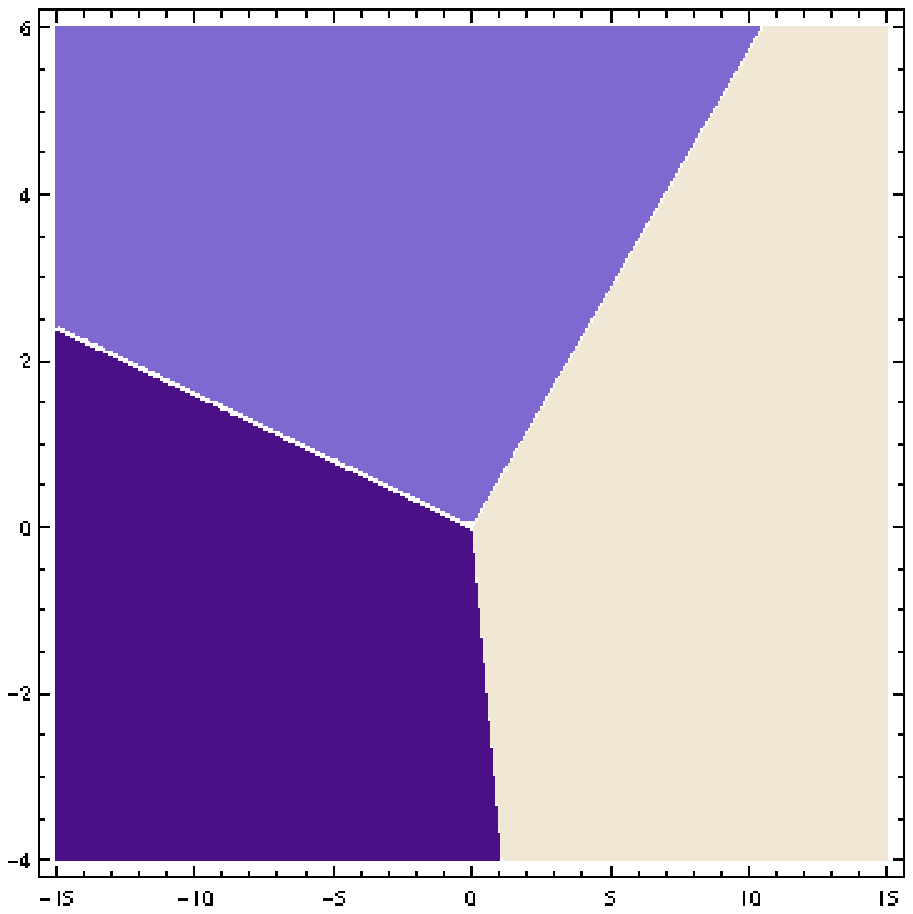}\hspace{5mm}
\includegraphics[scale=0.5]{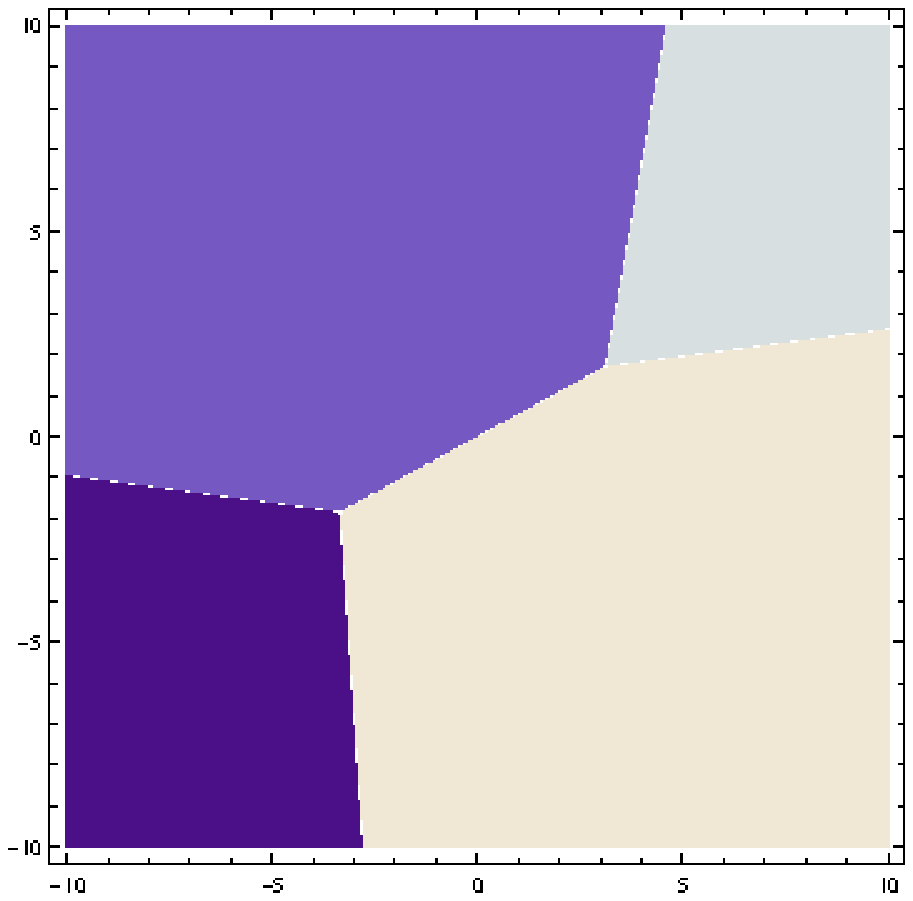} \hspace{5mm}
\includegraphics[scale=0.5]{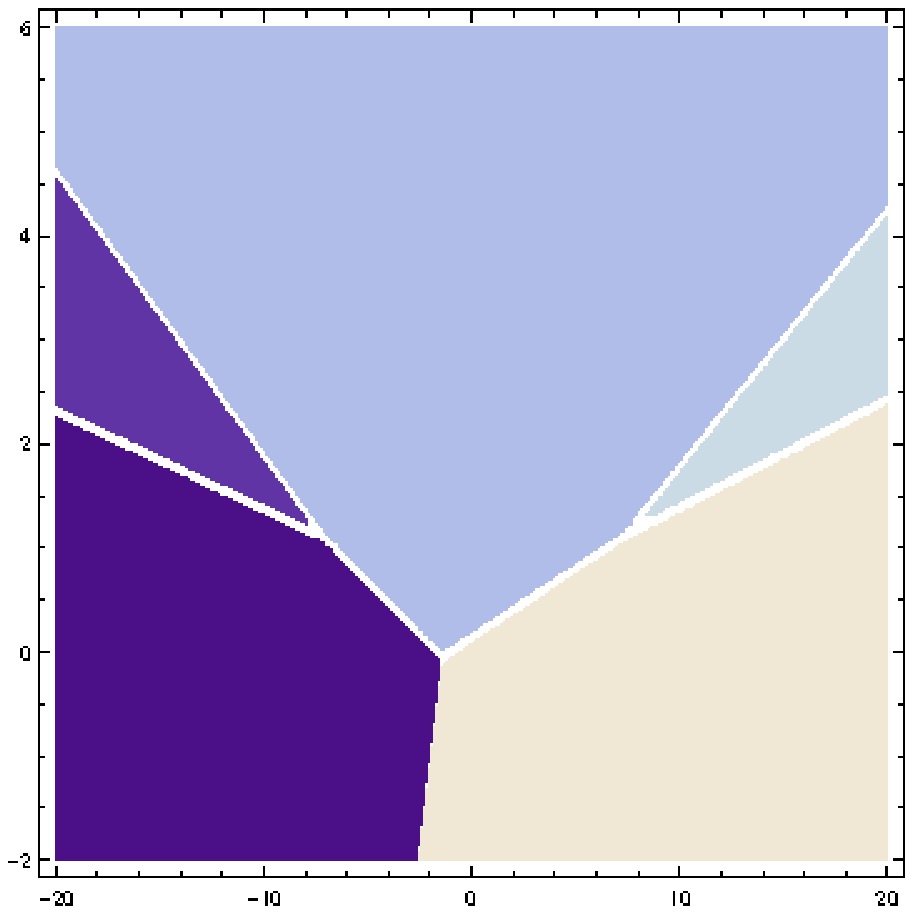}
\caption{Tropical  plots of $\phi^{(1)}(x,t)$  for breather solutions $N=5$ 
with parameters $\mu$ and $\alpha$ corresponding to plots on Figure \ref{br111}.
}
\label{br111t}
\end{figure}

This definition does not reflect the fact that we are dealing with a system of 
equations and thus the graphs corresponding to the variables $\phi^{(i)}(x,t),\ 
i=1,\ldots ,N$ are slightly different (they may depend on the index $i$). It 
can be considered as a first approximation which does capture well the 
trajectories of the solitons (breathers). 

The general approach to visualisation of rank $r$ solutions is similar to the case of 
rank one. We use the fact that a rank $r$ solution is a function of the 
point 
$$\bn(x,t)=\exp((\mu \Delta^{-1}-\mu^{-1} \Delta) x-(\mu^{-2} \Delta^2-\mu^2 
\Delta^{-2}) t)\bn_0 $$ 
on the Grassmannian ${\rm Gr}(r, N)$, where $\bn(x,t)$ is a $N\times r$ full rank 
matrix (as it follows from Proposition \ref{pro5b}). In the basis $\be_k$ 
(\ref{bek}) the matrix $\bn(x,t)$ can be represented as 
\[
 \bn(x,t)=(\be_1,\ldots ,\be_N){\alpha}^{\tr}(x,t),
\]
where $\alpha(x,t)$ is $r\times N$ matrix of full rank and 
\[
 \left(\alpha(x,t)\right)_{pq}=\alpha^{(0)}_{pq}\exp((\mu 
\omega^{-q}-\mu^{-1} \omega^q) x-(\mu^{-2} \omega^{2q}-\mu^2 
\omega^{-2q}) t),\qquad 1\le p\le r,\ 1\le q\le N.
\]
Let 
\[
 I=\{i_1<i_2<\ldots <i_k\}\subset [1,\ldots ,N]
\]
and let $\Delta_I(x,t)$ denote the minor of $\alpha(x,t)$ with columns 
$i_1,\ldots 
, i_k$ (a Pl\"ucker coordinate on the Grassmannian ${\rm Gr}(k, N)$). Let us define 
$\Theta_I(x,t)=\log |\Delta _I(x,t)|$ if  
$\Delta_I(x,t)\ne 0$ and  for such $I$ that 
$\Delta_I(x,t)= 0$ we set $\Theta_I(x,t)=-\infty$. The function $\Theta _I(x,t)$ is a linear function of the 
coordinates $(x,t)$. If there is only one nonzero minor $\Delta_I(x,t)$, then 
it 
is easy to show that the corresponding solution is $\phi^{(i)}(x,t)=0, \ 
i=1,\ldots, N$ (similar to Proposition \ref{pro10}). The solution is concentrated 
near the points where two or more Pl\"ucker coordinates are in balance and we 
can give the following definition of the tropical soliton graph in the case of 
rank $r$ breather solutions.

\begin{Def}
 For rank $r$ breather solutions the tropical soliton graph is defined as a 
locus of points where the function 
$\Theta(x,t)=\max_I \Theta_I(x,t)$ is not smooth. 
\end{Def}

Using the definition, we plot the tropical soliton graph for
$$N=5,\qquad \mu=0.57+0.2\i,\qquad \alpha=\left(\begin{array}{rrrrr}
      1&10&10^4&10^3&1\\10^6&1&10^6&10^4&99.9
                                           \end{array}\right)$$
and compare it to the actual density plot for $\phi^{(3)}(x,t)$ in the following graph:
 \begin{figure}[ht]
\centering
\includegraphics[scale=0.52]{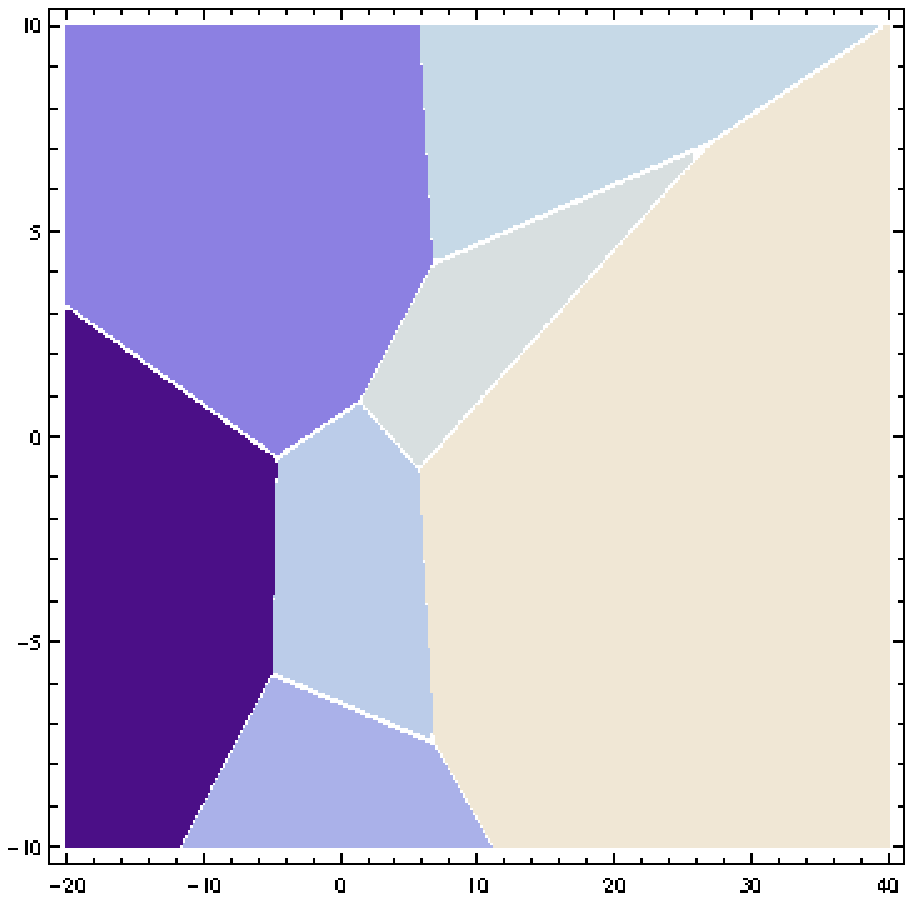}\hspace{9mm}
\includegraphics[scale=0.5]{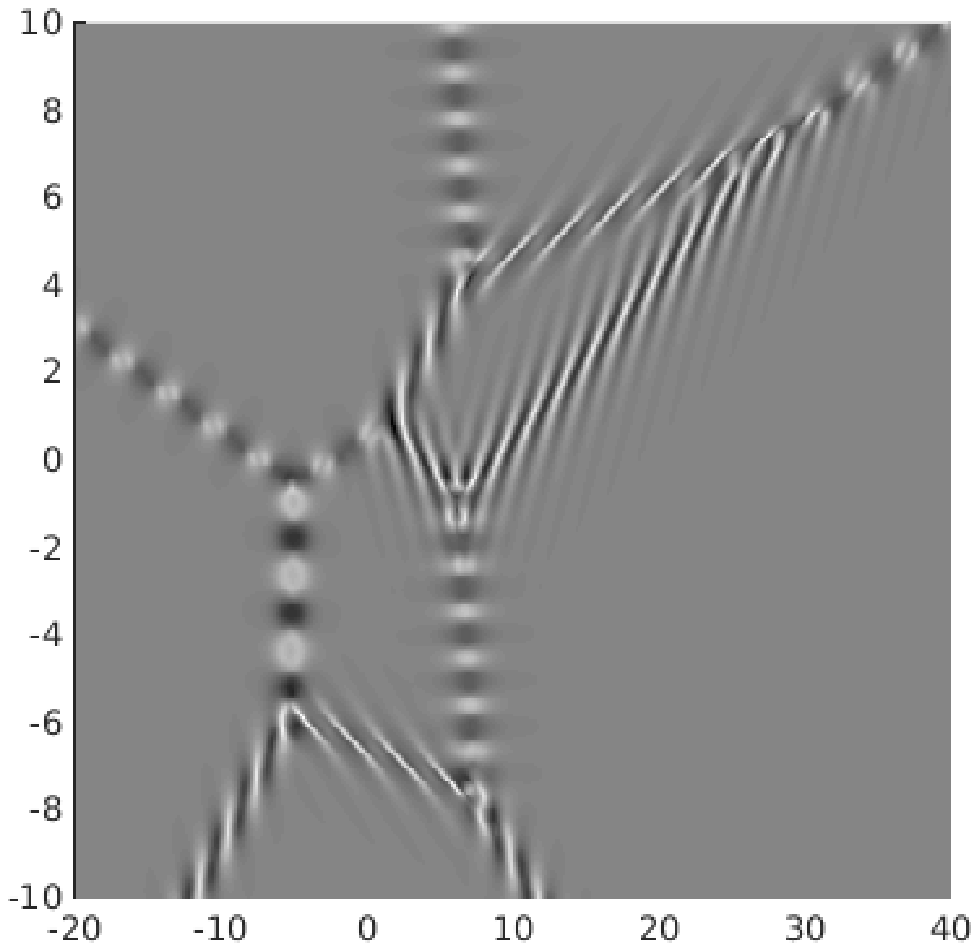}
\caption{Tropical and density plots of $\phi^{(3)}(x,t)$  for rank 2 breather 
solutions $N=5$ }
\label{breatherr2}
\end{figure}

Although the above definition of a tropical soliton graph is not perfect (it does not reflect the dependence of 
the graph on the index $i$ for different components $\phi^{(i)}(x,t)$), it 
reflects well the picture of the breather interactions. It also opens a path 
to classification of possible configurations in the multi-soliton solutions 
of arbitrary rank. 

\section{Conclusion}\label{sec5}

In this paper we developed the dressing method for the two dimensional Volterra 
system (\ref{2+1}). We have constructed two types of exact solutions to the 
system. The first type is rather unusual. It represents propagation of the wave 
fronts. Up to the best of our knowledge it is a new class of solutions in 
integrable models. The second type resembles breathers in the sine-Gordon 
equation. Nonlinear wave (``kink'') solutions are parametrised by a real  
parameter $\nu$ and a point on a real Grassmannian ${\rm Gr}_\bbbr (k,N)$. In the 
case 
of breathers the parameter $\mu$ and Grassmannian ${\rm Gr}_\bbbc (k,N)$ are complex. 
The integer $k$ is the rank of the solution. We have studied in detail the 
properties and configurations of rank 1 solutions, where the Grassmannians are 
real and complex projective spaces respectively. Classification of rank $k$ 
solutions can be linked with the classification of $\Delta$--invariant Schubert 
decompositions of the Grassmannians, where $\Delta$ is the cyclic shift matrix 
from the Lax representation of the two dimensional Volterra system (\ref{2+1}). 

In this paper we have not yet developed a classification of higher rank 
solutions, but we claim  that their  properties are quite different from the 
solutions of rank one. For example the nonlinear wave (``kink'') solutions of 
rank 2 may represent a nonlinear interference of waves (see Figure \ref{T}, 
right) 
which is impossible in the case of rank one solutions.
In the case $N=4$ we have listed in Section \ref{N4k} and presented all 
possible rank 1 kink solutions. Here we present density 
plots for two kink solutions of rank 2 when $N=4$ in Figure \ref{kink4rank2}
and some snapshots in Figure \ref{kink4rank2snap},
which does 
not resemble any rank 1 solution.
\begin{figure}[ht]
\centering
\includegraphics[scale=0.45]{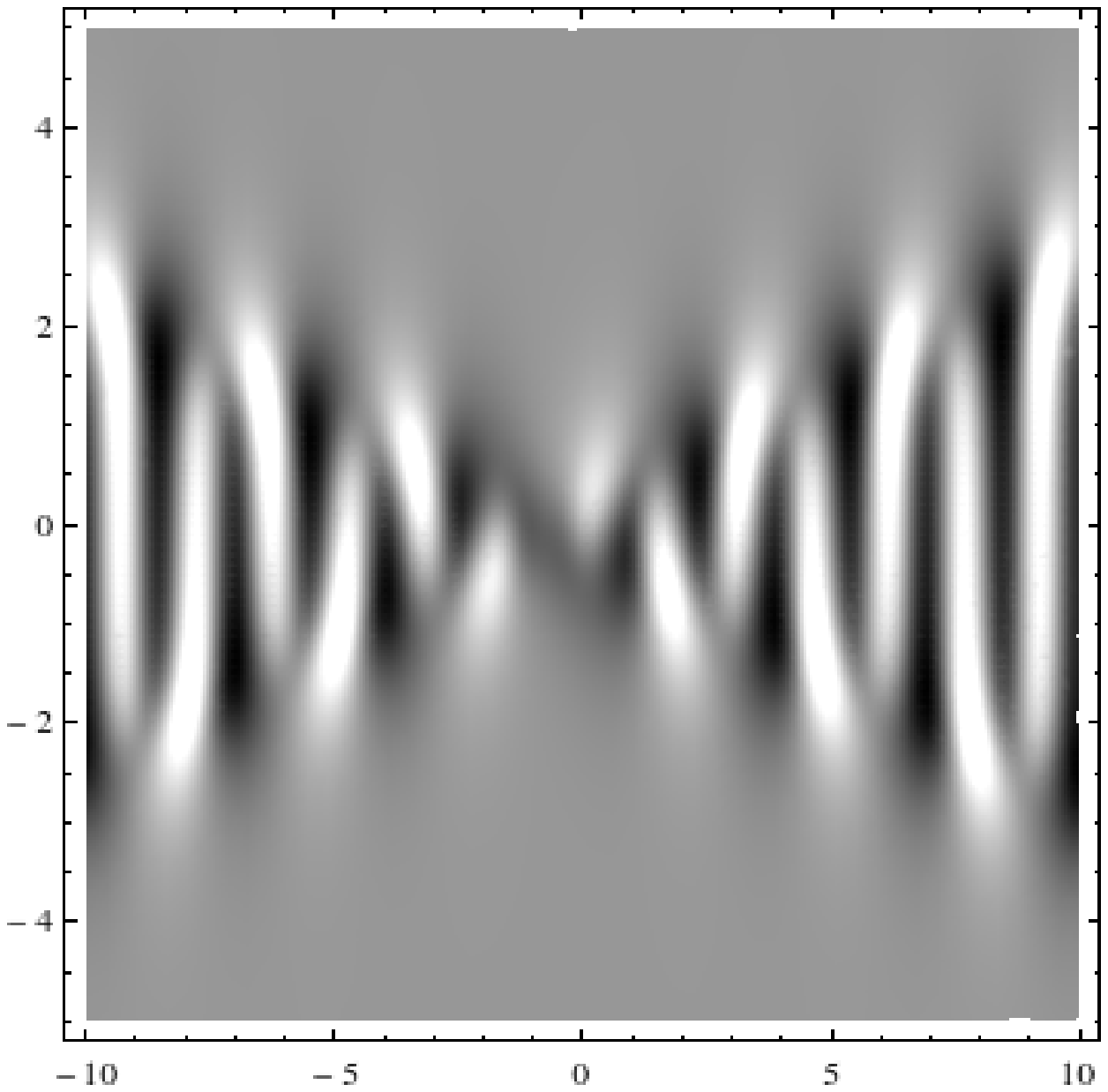} \hspace{9mm}
\includegraphics[scale=0.45]{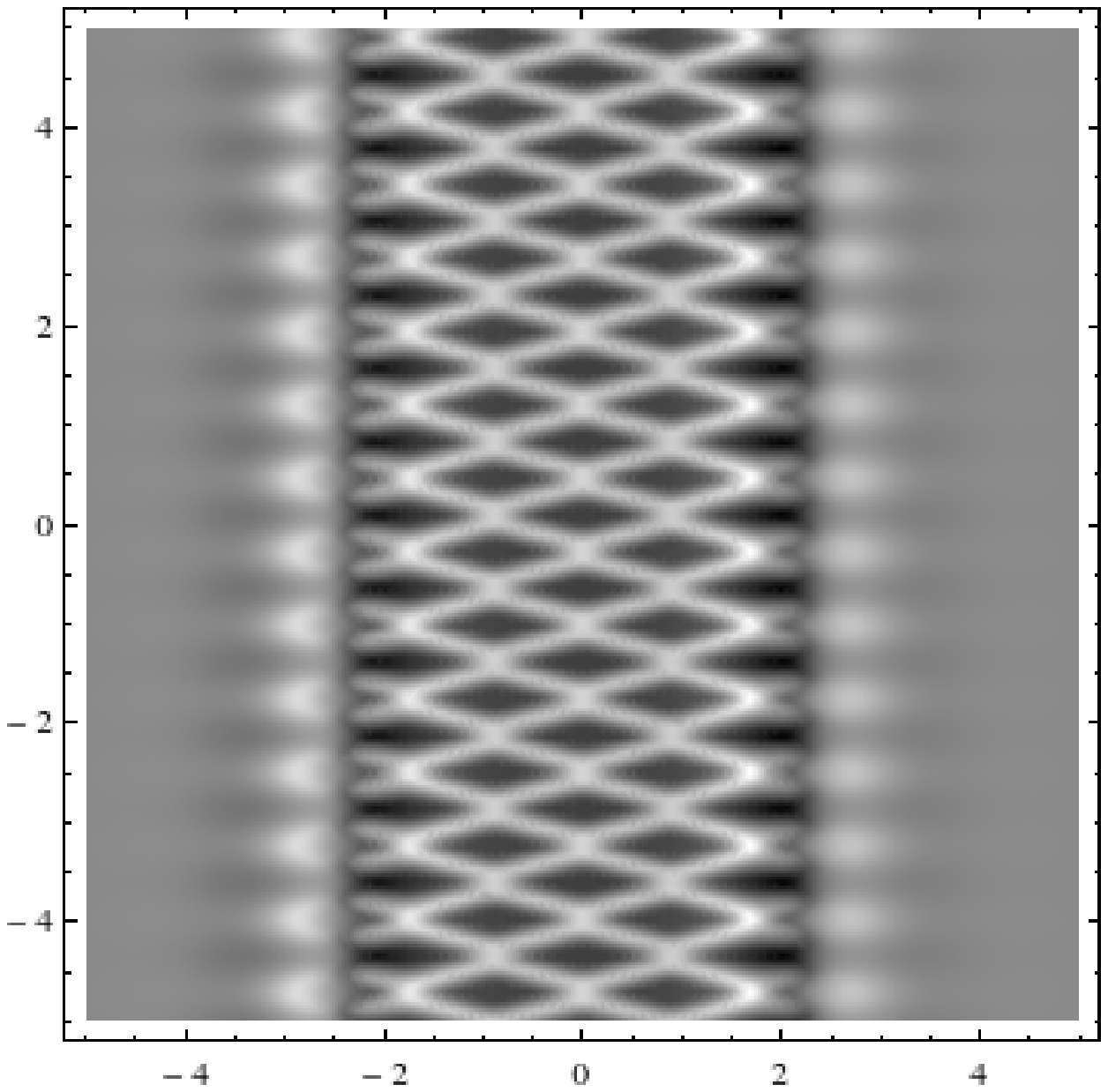}
\caption{Density  plot of $\phi^{(1)}(x,t)$ for a kink solution of rank $2$ when 
$N=4$: $\nu=0.8, 
\bn_0=(\i\be_1-2\be_2-\i\be_3+\be_4,\ 2\be_1+\be_2+2\be_3+\be_4)$ for the left 
plot and $\mu=2 e^{\frac{\pi \i}{4}}, \bn_0=(\be_1 
+(1+100\i)\be_2+\be_3+(1-100\i)\be_4,\ 10\i\be_1+5\i\be_2-10\i\be_3-5\i\be_4)$ 
for the right plot.
}
\label{kink4rank2}
\end{figure}
\begin{figure}[ht]
\centering
\includegraphics[scale=0.41]{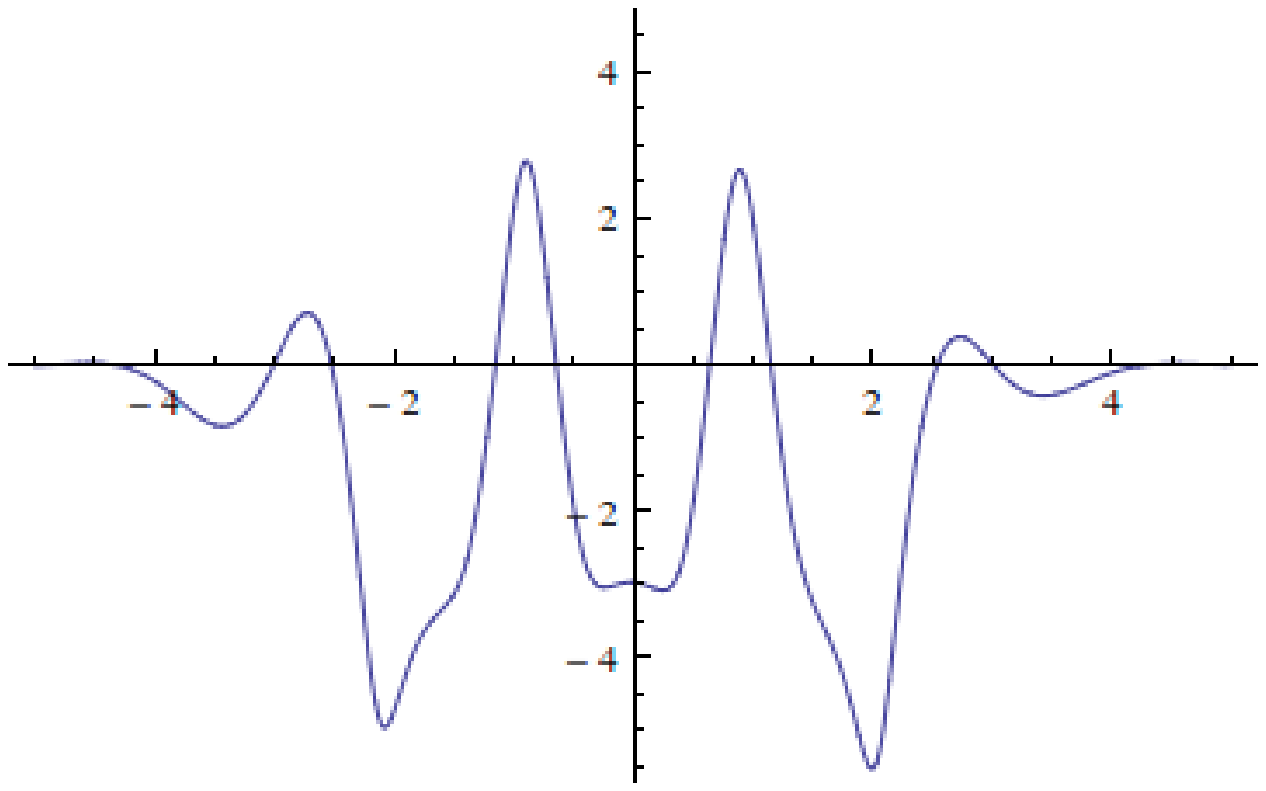} \ \ 
\includegraphics[scale=0.41]{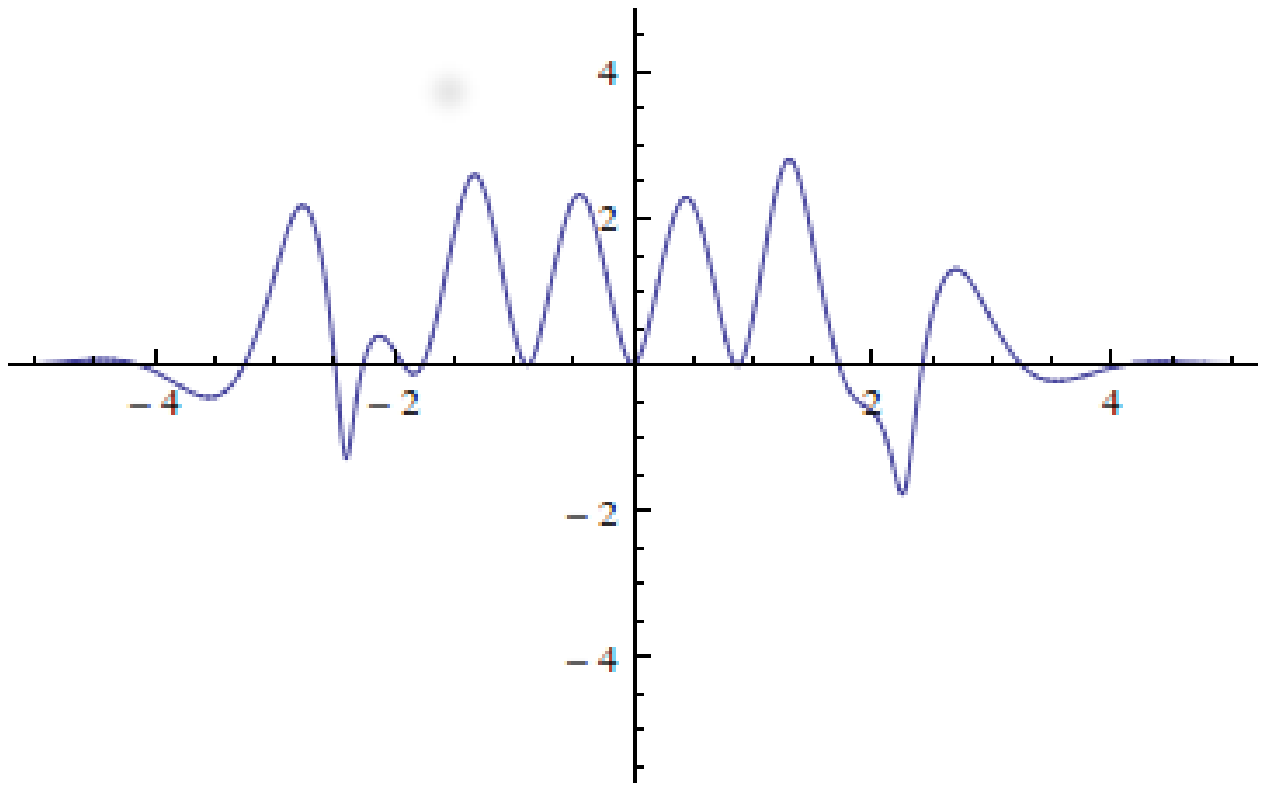} \ \ 
\includegraphics[scale=0.41]{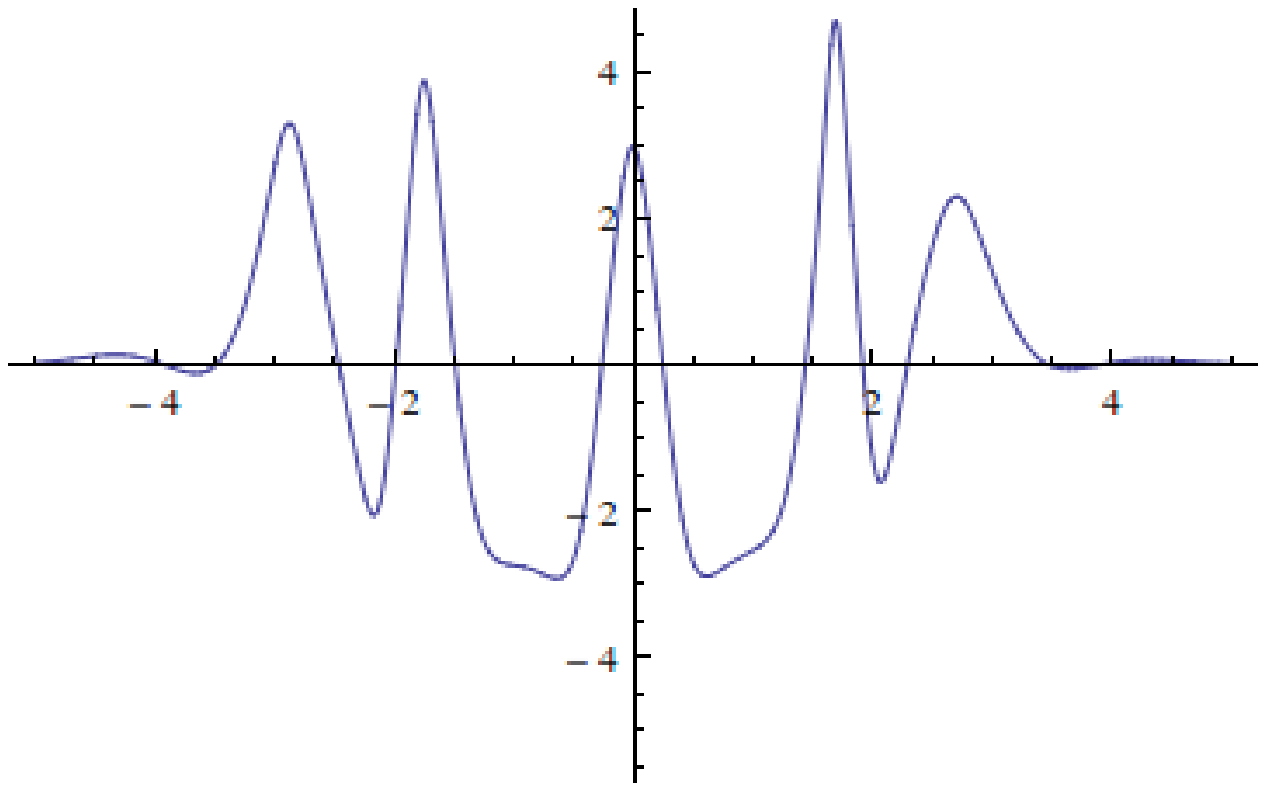}
\caption{Snapshots of $\phi^{(1)}(x,t)$ for rank $2$ kink solution on the right 
plot in Figure \ref{kink4rank2}: $t=0.0945$ for the left plot, $t=-0.09075$ for 
the middle plot and $t=-0.276$ for the right plot.
}
\label{kink4rank2snap}
\end{figure}
Breather solutions of rank 1 do not have closed loops, but in rank 2 loops 
exist (see Figure \ref{Tb}). 
%To visualise higher rank solutions it is convenient to introduce tropical soliton plots. 

To study the structure and classify higher rank wave front and 
breather solutions as well as multi-soliton solutions (with a finite number of 
 orbits of the poles in the dressing matrix $\Phi(\lambda)$) we 
need to develop further methods similar to ones proposed by Kodama et al. for the 
KP equation \cite{kodama2011kp, Kodama2013, kodama2014kp}.
There is also an interesting and as yet unsolved problem to find solutions of 
(\ref{2+1}) which approximate for large $N$ the solutions of the KP equation.

\section*{Acknowledgements}
AVM would like to acknowledge the financial support by the Leverhulme Trust. 
Both AVM and JPW would like to thank Y. Kodama for numerous discussions during 
his one-month visit 
to the UK, which was supported by the LMS, and to L. Williams for her seminal lectures 
during 
the meeting ``Total Positivity: A bridge between Representation Theory and Physics'' at the University of Kent.

%\bibliography{kdv}
\end{document}